\documentclass[12pt]{amsart}
\usepackage{graphicx}
\usepackage{graphicx}
\usepackage{epsfig}
\usepackage{pstricks}

\newtheorem{thm}{Theorem}[section]

\newtheorem{lemma}[thm]{Lemma}

\newtheorem{cor}[thm]{Corollary}
\newtheorem{remark}[thm]{Remark}
\newtheorem{example}[thm]{Example}

\usepackage{amsmath}
\usepackage{amsxtra}
\usepackage{amscd}
\usepackage{amsthm}
\usepackage{amsfonts}
\usepackage{amssymb}
\usepackage{eucal}
\newcommand{\bmb}{\left( \begin{array}{rr}}
\newcommand{\enm}{\end{array}\right)}

\newcommand{\cQ}{\mathcal Q}
\newcommand{\cT}{\mathcal T}

\newcommand{\Z}{{\mathbb Z}}

\newcommand{\bu}{{\mathbf u}}
\newcommand{\bv}{{\mathbf v}}

\newcommand{\al}{{\alpha}}

\numberwithin{equation}{section}

\begin{document}

\title[Arctic curves of the Reflecting Boundary 6V and of the 20V models]{Arctic curves of the Reflecting Boundary Six Vertex and of the Twenty Vertex models}
\author{Philippe Di Francesco} 
\address{$\!\!\!\!\!\!\!\!\!\!\!\!$ Department of Mathematics, University of Illinois, Urbana, IL 61821, U.S.A. 
and \break
Institut de physique th\'eorique, Universit\'e Paris Saclay, 
CEA, CNRS, F-91191 Gif-sur-Yvette, FRANCE\hfill
\break  e-mail: philippe@illinois.edu
}

\begin{abstract}
We apply the Tangent Method of Colomo and Sportiello to predict the arctic curves of the Six Vertex model with reflecting (U-turn) boundary and of the related Twenty Vertex model with suitable domain wall boundary conditions on a quadrangle, both in their Disordered phase. 
\end{abstract}

\maketitle
\date{\today}
\tableofcontents

\section{Introduction}
\medskip
\subsection{Arctic phenomenon}

Geometrically constrained two-dimensional statistical models are known to display the so-called arctic phenomenon in the presence of suitable boundary conditions. This includes ``free fermion" dimer models, where typically dimers choose a preferred crystalline orientation near boundaries while they tend to be disordered (liquid-like) away from the boundaries: the arctic phenomenon is the formation of a sharp phase boundary as the domain is scaled by a large overall factor, the so-called arctic curve separating frozen crystalline from disordered liquid phases. The first observed instance of this phenomenon is the celebrated arctic circle arising in the domino tilings of the Aztec diamond \cite{JPS}, and a general theory was developed for dimers \cite{KO2,KOS}. The free fermion character of these models can be visualized in their formulation in terms of non-intersecting lattice paths, i.e. families of paths with fixed ends, subject to the condition that they share no vertex (i.e. avoid each-other), and can consequently be expressed in terms of free lattice fermions. A manifestation of the free fermion models is that their arctic curves are always analytic, and usually algebraic at ``rational" values of interaction parameters, such as in the uniformly weighted cases.

Beyond free fermions, the archetypical model for paths allowed to interact by ``kissing" i.e. sharing a vertex at which they bounce against each-other, is the Six Vertex (6V) model. The families of paths describing the model are called osculating paths.
With so-called Domain Wall Boundary Conditions (DWBC)
the 6V model exhibits an arctic phenomenon in its disordered phase, which was predicted via non-rigorous methods\cite{CP2010,CNP}, the latest of which being the Tangent Method introduced by Colomo and Sportiello \cite{COSPO}. The new feature arising from these studies is that the arctic curves are generically {\it no longer analytic}, but rather {\it piecewise analytic}. For instance, the arctic curve for large Alternating Sign Matrices (uniformly weighted 6V-DWBC) is made of four pieces of different ellipses as predicted in \cite{CP2010} and later proved in \cite{Aggar}. 

The Tangent Method was validated recently in a number of situations, mostly in free fermion situations \cite{CPS,DFLAP,PRarctic,DFGUI,DFG2,DFG3,CorKeat}. However, a simple transformation using the integrability of the models allowed to deduce from the 6V results the arctic curves for another model of osculating paths: the Twenty Vertex (20V) model with DWBC1,2 \cite{BDFG}. The 20V model is the triangular lattice version of the 6V model: in one formulation the configurations of the model are orientation assignments of all edges of the lattice, in such a way that the {\it ice rule} is obeyed at each vertex, namely that there be an equal number of edges pointing towards and outwards (2+2 for 6V, 3+3 for 20V). In \cite{DFGUI}, four possible variations around DWBC were considered for the 20V model, denoted DWBC1,2,3,4.
In the present paper, we will concentrate on the 20V-DWBC3 model on a quadrangle, which was recently shown to have the same number of configurations as domino tilings of the Aztec Triangle of suitable size \cite{DF20V}. The proof uses again the integrability of the model to relate its partition function to that of the 6V model with another type of DWBC, called U-turn, considered by Kuperberg in \cite{kuperberg2002symmetry}, and whose partition function has a nice determinantal form \cite{tsu,kuperberg2002symmetry}.

In this paper, we set the task of deriving the arctic curves for the U-turn 6V model, and as by-products, those of the 20V-DWBC3, and of the Domino Tiling of the Aztec Triangle.

%
%
%

\subsection{Arctic curves and the Tangent Method}

The systems we are considering in this note are all described in terms of osculating or non-intersecting paths, and are expected to display an arctic curve phenomenon.
The rough idea behind the Tangent Method is as follows. The $n$ paths describing the model's configurations have fixed starting and endpoints, and form a ``soup" whose boundary tends to a subset of the arctic curve. Indeed, this boundary is a solid/liquid separation between an empty phase and one with disordered path configurations. Consider the outermost path forming that boundary: if we displace the endpoint of this path to a point say $L$ outside of the original domain, the path will have to detach itself from the soup, and continue to its endpoint within a mostly empty space, once it gets away from the soup formed by the other paths,
where it is most likely to follow a geodesic (a line in all cases of this paper, due to a general argument of \cite{DFLAP}).
This geodesic is expected to be {\rm tangent} to the arctic curve in the large $n$ limit. The corresponding path is therefore  used as a probe into the arctic curve: the geodesic is determined by the point $L$  and the point $K$ at which it exits the original domain\footnote{Both points $K$ and $L$ scale linearly with the size $n$ so that a thermodynamic limit can be reached.}. 

The partition function $\Sigma_{n,L}$ of the new model is now a sum over the possible positions of $K$ of the product of two partition functions: (1) $Z_{n,K}$ the partition function of the $n$ osculating/non-intersecting paths on the original domain, in which the outer path is conditioned to end at point $K$ instead of its original endpoint. (2) $Y_{K,L}$ the partition function of a single path subject to the same weighting, in some empty space from the point $K$ to the new endpoint $L$.  The quantity $\Sigma_{n,L}=\sum_K Z_{n,K}\, Y_{K,L}$ is dominated at large $n$ by contributions from the most likely exit point $K(n,L)$. The arctic curve
is then recovered as the envelope of the family geodesics through $L$ and $K(n,L)$ for varying $L$ (in rescaled coordinates).

We see that the crucial ingredient in this method is the refined partition function $Z_{n,K}$ in which the outer path is conditioned to exit the domain at point $K$,
or rather its normalized version, the ``refined one-point function" $H_{n,K}=Z_{n,K}/Z_n$ where we have divided by the original partition function $Z_n$.  
Computing exactly the leading large $n,K,L$ asymptotics of $H_{n,K}$
and $Y_{K,L}$ leads to the determination of $K(n,L)$ by solving a steepest descent problem, and eventually to the arctic curve.

After revisiting the case of the 6V model for pedagogical purposes in Section \ref{6vsec}, 
we will apply the Tangent Method in Section \ref{6vpsec} to the case of the 6V' model on the $(2n-1)\times n$ rectangular grid (a simplified version of the U-turn 6V model), in Section \ref{20vsec}
to the case of the 20V model with DWBC3
on the quadrangle $\cQ_n$, and finally in Section \ref{DTsec} to the domino tilings of the Aztec triangle $\cT_n$.
Note that the Tangent Method was previously applied in \cite{PRarctic} to a particular ``free fermion" case of the U-turn 6V model, where the arctic curve is a half-circle: the results of Section \ref{6vsec} extend this to arbitrary values of the parameters.

\subsection{Outline of the paper and main results}

The paper is organized as follows.
In Section \ref{secmodels} we define the four models studied in this paper. These include: the 6V model with 
Domain Wall Boundary Conditions (DWBC), the 6V model with U-turn Boundary Conditions and the related 6V' model,
the 20V model with DWBC3 of Ref.~\cite{DF20V}, and finally the Domino Tiling problem of the Aztec Triangle introduced and studied in Refs.~\cite{DFGUI,DF20V}.  
We show in particular that all models are described by families of weighted osculating/non-intersecting paths. 
In Section \ref{sectan}, we 
describe the Tangent Method in general and how it applies to the determination of the arctic curves of our models.

The next sections are all organized in a similar way, and treat the various models. For each case, we first derive compact relations obeyed 
by the partition function and one-point function of the model, allowing for extracting asymptotic results. The latter are used to apply 
the Tangent Method, and finally obtain the arctic curves of the model.
While Section \ref{6vsec} revisits the 
known case of the 6V model with DWBC, as a pedagogical warmup, the remaining Sections provide new results:  
Section \ref{6vpsec} is about the 6V' model, Section \ref{20vsec} the 20V model with DWBC3, and 
Section \ref{DTsec} the Domino Tilings of the Aztec Triangle. We obtain arctic curves in all cases: Theorems \ref{6VNEthm},\ref{6VpNEthm},\ref{20VNEthm} and \ref{DTthm} cover respectively the cases of 6V,6V',20V and Domino Tilings.


We gather a few concluding remarks in Section \ref{seconc}.

\medskip
\noindent{\bf Acknowledgments.} We are thankful to G.A.P. Ribeiro for bringing 
Refs.~\cite{RIBKOR,PRarctic} to our attention.
We acknowledge support from the Morris and Gertrude Fine endowment, the NSF grant DMS18-02044, and the NSF RTG Grant DMS19-37241.
\bigskip

\section{Models, Paths and the Tangent Method}\label{modelsec}

\subsection{The models}\label{secmodels}

In this paper we consider 4 different models: three vertex models with particular Domain-Wall type boundary conditions (6 Vertex on a $n\times n$ square grid, 6 Vertex with U-turn boundaries on a $2n-1\times n$ rectangular grid,
and 20 Vertex on the quadrangle $\cQ_n$), and one model of domino tilings of the Aztec triangle $\cT_n$.

\subsubsection{6V model with DWBC}

\begin{figure}
\begin{center}
\includegraphics[width=14cm]{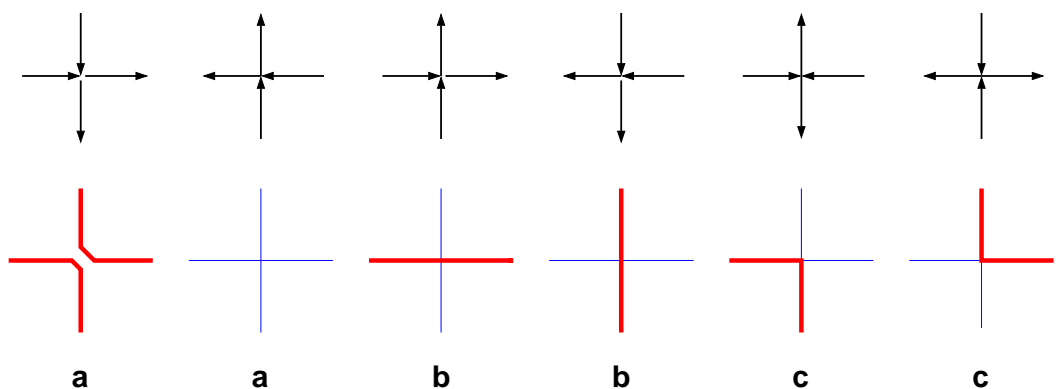}
\end{center}
\caption{\small The 6 vertex environments obeying the ice rule on the square lattice (top) and their osculating path reformulation (bottom). We have indicated the corresponding types a,b,c.}
\label{fig:six}
\end{figure}

\begin{figure}
\begin{center}
\includegraphics[width=14cm]{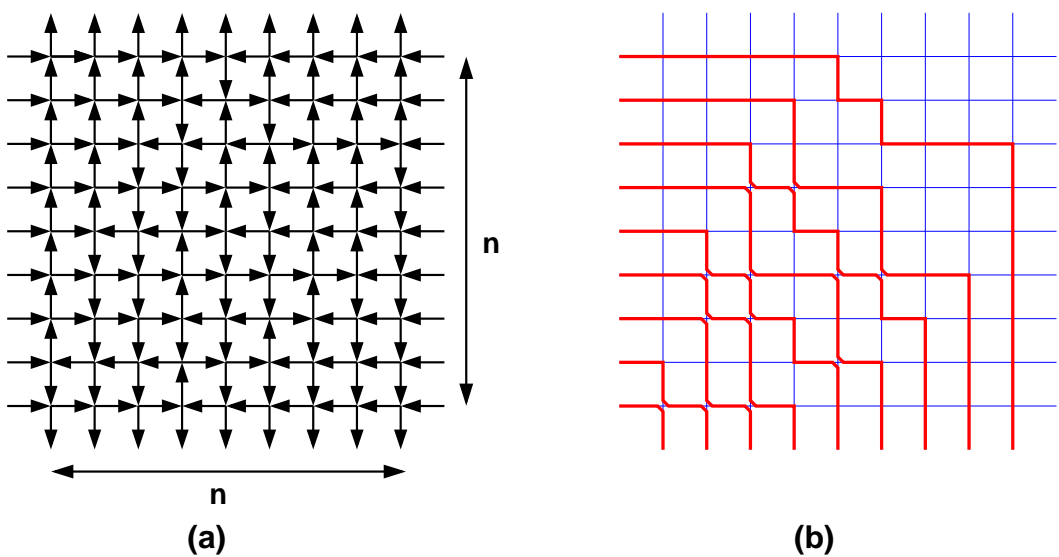}
\end{center}
\caption{\small (a) A sample configuration of the 6V model with DWBC (for $n=9$) and (b) its reformulation in terms of osculating paths.}
\label{fig:oscsixv}
\end{figure}

The 6V model is the archetype of integrable ice-type model on the two-dimensional square lattice.  Its configurations are
obtained by orienting the edges of the lattice (with arrows) in such a way that each vertex has exactly two entering and two out-going
arrows (the so-called ``ice rule"). This gives rise to the ${4\choose 2}=6$ local environments of Fig.~\ref{fig:six} (top row), traditionally called a,b,c types.
Here we consider the 6V model on an $n\times n$ square grid with fixed Domain Wall Boundary Conditions (DWBC), i.e. the $2n$ horizontal boundary arrows 
($n$ on the West (W) and $n$ on the East (E) boundaries) pointing towards the
square domain and the $2n$ vertical ones ($n$ on the North (N) and $n$ on the South (S) boundaries) outwards (see Fig.~\ref{fig:oscsixv} (a) for an illustration). 
Finally, the configurations are weighted by the product of local vertex weights over the domain\footnote{We restrict throughout the paper to the Disordered regime, in which all weights are trigonometric.}, parameterized by real
``spectral parameters" $u,v$ attached to the horizontal and vertical line that intersect at the vertex, taking the following values in the so-called Disordered regime, which we consider in this paper:
\begin{equation}\label{6vweights}
a=\rho\sin(u-v+\eta) ,\qquad  b=\rho\sin(u-v-\eta),\qquad  c=\rho\sin(2\eta)\end{equation}
for a,b,c type vertices respectively. 
The overall fixed factor $\rho>0$ emphasizes the projective nature of the weights and the homogeneity of the partition function (weighted sum over configurations), from which $\rho^{n^2}$
factors out.
Positivity of the weights imposes the condition:
\begin{equation}\label{domain6v} \eta<u-v < \pi-\eta , \qquad 0<\eta<\frac{\pi}{2} .\end{equation}

In the following we shall consider the homogeneous partition function $Z_n^{6V}[u,v]\equiv Z_n^{6V}[u-v]$ in which all horizontal spectral parameters at taking the value $u$ and all vertical ones the value $v$, so that weights are uniformly defined by \eqref{6vweights}, and where we note that both the weights and the partition function only depend on the quantity $u-v$. As stressed in \cite{CP2009}, the partition function enjoys the following crucial symmetry property:
\begin{equation}\label{sym6v}
Z_n^{6V}[\pi-(u-v)]=Z_n^{6V}[u-v]
\end{equation}
This is a consequence of the reflection symmetry of the weights: indeed, the DWBC are unchanged if we reflect the domain w.r.t. say a horizontal line. However, such a reflection interchanges vertices of types  $a \leftrightarrow b$ while keeping $c$-type environments unchanged. The same result is independently obtained by keeping the original setting, but applying the transformation $(u-v)\to \pi-(u-v)$ which leaves the domain \eqref{domain6v} invariant, and under which the weights $a$ and $b$ are interchanged, while $c$ remains invariant, and \eqref{sym6v} follows.

This model was extensively studied \cite{Lieb,DWBC,IKdet}, and turned out to play a crucial role in Kuperberg's proof of the Alternating Sign Matrix (ASM) conjecture \cite{kuperberg1996another,Bressoud}. 
The enumeration of ASM is realized at the ``ASM point" where all weights are equal to $1$, namely:
\begin{equation}\label{asmpoint} \eta=\frac{\pi}{6},\qquad u-v=\frac{\pi}{2},\qquad \rho=\frac{1}{\cos(\eta)} \end{equation}
while the refined enumeration (with a factor $\tau$ per entry $-1$ in the ASM) is provided by picking
\begin{equation}\label{tauasmpoint}u-v=\frac{\pi}{2}, \qquad  \rho=\frac{1}{\cos(\eta)},  \qquad \tau=4\sin^2(\eta),  \qquad 0<\eta<\frac{\pi}{2} \end{equation}
namely $(a,b,c)=(1,1,\sqrt{\tau})$,
with the particular cases of $1,2,3$-enumeration, for the choices $\eta=\frac{\pi}{6},\frac{\pi}{4},\frac{\pi}{3}$ respectively.
The ``20V-DWBC1,2 point" is another interesting combinatorial point, which corresponds to the identification of the number of 20V DWBC1,2 
configurations in terms of 6V DWBC \cite{DFGUI}, with the choice:
\begin{equation}\label{20vpoint}\eta=\frac{\pi}{8},\qquad u-v=\frac{5\pi}{8}, \qquad \rho=\sqrt{2} \end{equation}
corresponding to weights $(a,b,c)=(1,\sqrt{2},1)$.

More recently the thermodynamic free energy of the model was obtained in \cite{KORZIN,PZ6V,BleFok}, and the arctic curves were derived using various semi-rigorous methods, such as the Tangent Method in \cite{COSPO,CPS}, and further used in \cite{BDFG} to determine the arctic curves of the 20V DWBC1,2 models.

The configurations of the model can be rephrased in terms of families of osculating paths as follows. Pick a base orientation of arrows, say to the left and down, and mark
all the edges of any given configurations that respect the base orientation. Note that all the W and S boundary edges are marked, while the N and E ones are not. The marked edges
can be combined into paths say starting at the W boundary and ending at the S one, with right and down steps only, that are non-intersecting but may kiss/osculate by pairs at fully marked vertices: the corresponding six local configurations are depicted on the second row of Fig.~\ref{fig:six}. The osculating path formulation is well adapted to the Tangent Method as we shall see below.
For illustration, we have represented in Fig.~\ref{fig:oscsixv} a sample 6V DWBC configuration both in the arrow (a) and osculating path (b) formulations.

\subsubsection{6V model with U-turn boundary and 6V' model}\label{gen6vpsec}

\begin{figure}
\begin{center}
\includegraphics[width=16cm]{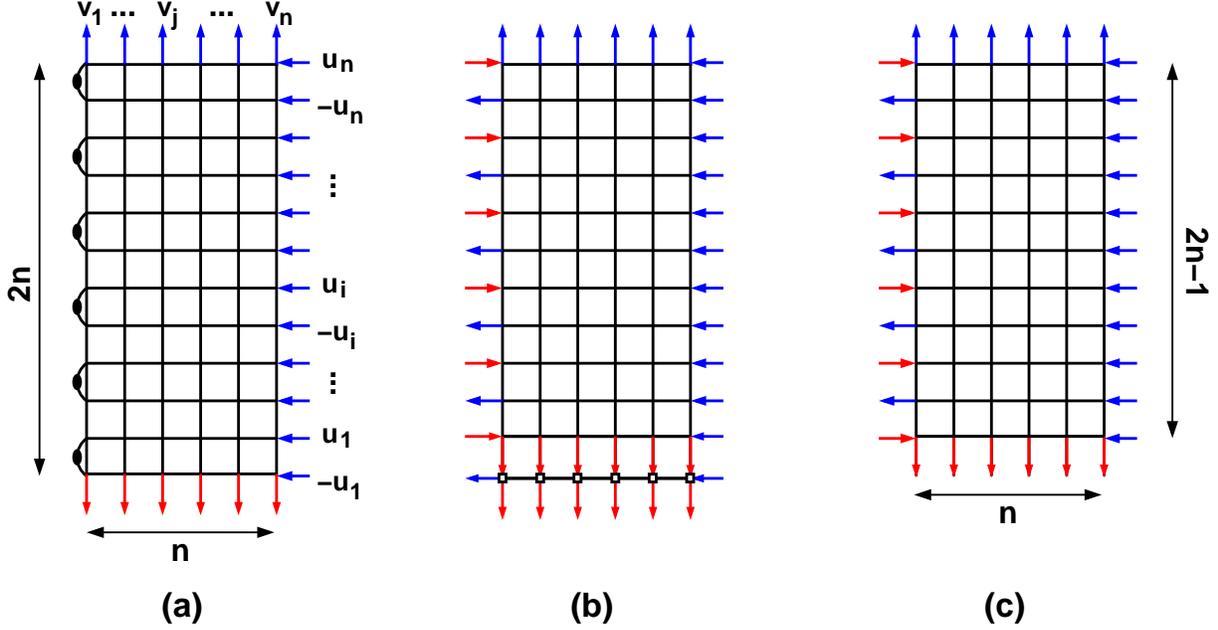}
\end{center}
\caption{\small (a) U-turn boundary 6V model: each U-turn (marked by a black dot along the W boundary) transmits the arrow orientation through the dot. (b) When all $u_i=-\theta-\eta$, the weights $y_d=0$, hence all arrows go up through the U-turns, which may be cut out as shown. The bottom row becomes trivially fixed to the same b-type vertex. (c) the 6V' model is finally obtained by cutting out the trivial b-type vertices.}
\label{fig:6vUtop}
\end{figure}

Kuperberg considered different symmetry classes of ASM, which in turn correspond to different variations around the 6V-DWBC model \cite{kuperberg2002symmetry}. In particular he found a remarkable connection between Vertically Symmetric ASMs (VSASM) and the 6V model with so-called U-turn boundary conditions (6V-U), also considered independently by Tsuchya \cite{tsu}. The 6V-U model is defined on a rectangular grid of square lattice of size $2n\times n$ with the usual DWBC along the N,S boundaries (each with $n$ outgoing vertical arrows) and E boundary (with $2n$ entering horizontal arrows), while the W boundary has U-turns connecting the $2n$ horizontal boundary edges (which we label $0,1,2,...,2n-1$ from bottom to top) by $n$ consecutive pairs $(2i,2i+1)$, $i=0,1,2,...,n-1$
(see Fig.~\ref{fig:6vUtop} (a) for an illustration). Each U-turn transmits the arrow orientation through the marked dot.  
The horizontal lines connected by a U-turn receive horizontal spectral parameters $-u_i$ (even label $2i$)  and $u_i$ (odd label $2i+1$), while vertical spectral parameters are denoted by $v_i$, $i=1,2,...,n$ from left to right.
As before, we consider the Disordered regime, with trigonometric weights depending on horizontal and vertical spectral parameters as in the case of the 6V-DWBC model.
The local weights, say at the intersection of a horizontal line with spectral parameter $u$ and vertical line with spectral parameter $v$ are:
\begin{equation}\label{6voddweights}
a_o=\rho_o\sin(u-v+\eta) ,\ b_o=\rho_o\sin(u-v-\eta),\ c_o=\rho_o\sin(2\eta)\end{equation}
on odd rows, while we must apply the transformation $u\to -u$ on even rows, resulting in:
\begin{equation}\label{6vevenweights}
a_e=\rho_e\sin(\eta-u-v) ,\ b_e=\rho_e\sin(-u-v-\eta),\ c_e=\rho_e\sin(2\eta)\end{equation}
and where the overall constant factors $\rho_o,\rho_e>0$ emphasize the projective character of the weights. 
Finally, U-turns receive weights:
\begin{equation}\label{uweights}
y_u(u)=-\sin(u-\theta+\eta) \qquad y_d(u)= \sin(u+\theta+\eta) \end{equation}
according to whether the transmitted arrow goes up or down. 
We must further constrain $u,v,\theta$ so that the weights of configurations of the 6V-U model are positive. A natural choice is $0<\theta<\frac{\pi}{2}$ and the following domain
for the $u,v,\eta$ parameters:
\begin{equation}\label{domain6}
\eta<u-v<\pi-\eta,\quad  \eta-\pi <u+v<-\eta, \quad 0<\eta<\frac{\pi}{2}
\end{equation}

For the purpose of this paper, we will consider the uniform case, where all horizontal odd spectral parameters are equal with value $u_i=u$ for all $i$, and all vertical ones are equal,
with value $v_j=v$ for all $j$, so that weights of odd/even rows are given by \eqref{6voddweights} and \eqref{6vevenweights} respectively.  Moreover, we pick $\theta=-u-\eta$, thus enforcing that at each U-turn the arrows go up\footnote{This choice simplifies the model by fixing the orientations of all arrows along the W boundary. However, we argue that the thermodynamics of the model are insensitive to that choice. For instance, the thermodynamic free energy, a bulk quantity, is  independent of the choice of $\theta$ (see Remark \ref{thetarem} below). So is the one-point function (see Remark \ref{thetaonerem} below). As a consequence, the arctic curves of the U-turn 6V and of the 6V' models are expected to be identical.}. Dividing each U-turn into two horizontal edges, we now obtain
arrows that alternate in/out along the W boundary (as shown in Fig.~\ref{fig:6vUtop} (b)). Note that the bottom row of vertices has all its edge orientations fixed by the ice rule. Upon dividing by the corresponding product of local even b-type weights, we may safely remove the $n$ vertices of the bottom line. After dividing by the weights of the removed vertices and U-turns, we are left with the 6V model on a rectangular grid of square lattice with size $2n-1\times n$, and
with usual DWBC along the N,E,S boundaries, while arrows alternate in/out from bottom to top along the W boundary (as depicted in Fig.~\ref{fig:6vUtop} (c)). Note that the rows are now labelled $1,2,...,2n-1$ from bottom to top.
By lack of a better name, we shall refer to this model as the 6V' model, and denote by $Z_n^{6V'}[u,v]$ the corresponding homogeneous partition function. 
Similarly to the 6V-DWBC case, this partition function enjoys a reflection symmetry property:
\begin{equation}\label{6vpsym}
Z_n^{6V'}[-u,-\pi-v]=Z_n^{6V'}[u,v]
\end{equation}
Indeed, like in the 6V case, applying a reflection w.r.t. a horizontal line to the rectangular domain interchanges vertices of type $a_o\leftrightarrow b_o$ and $a_e \leftrightarrow b_e$ while $c$-type vertices are unchanged. The same result is obtained in the original setting by applying the transformation $(u,v)\to (-u,-\pi-v)$, which leaves the domain \eqref{domain6} invariant, and \eqref{6vpsym} follows.

We now examine a few ``combinatorial points" in parameter space, where the partition function of the 6V' model has some known combinatorial interpretations.
Similarly to the 6V-DWBC case, the enumeration of Vertically Symmetric ASM (VSASM) is realized \cite{kuperberg2002symmetry} at the ``VSASM point" of the 6V' model, where all weights are equal to $1$, namely:
\begin{equation}\label{vsasmpoint}
\eta=\frac{\pi}{6},\qquad u=0,\qquad v=-\frac{\pi}{2},\qquad \rho_o=\rho_e=\frac{1}{\cos(\eta)}
\end{equation}
while the refined enumeration corresponds to \cite{kuperberg2002symmetry}:
\begin{equation}\label{tauvsasmpoint}
u=0,\qquad v=-\frac{\pi}{2},\qquad \rho_o=\rho_e=\frac{1}{\cos(\eta)}, \qquad \tau=4\sin^2(\eta)
\end{equation}
with even and odd weights $(a_i,b_i,c_i)=(1,1,\sqrt{\tau})$, with the particular cases of $1,2,3$-enumeration of VSASM corresponding respectively 
to $\eta=\frac{\pi}{6},\frac{\pi}{4},\frac{\pi}{3}$.
The ``20V-DWBC3 point" is another interesting combinatorial point, which corresponds to the identification of the number of 20V DWBC3 
configurations on $\cQ_n$ in terms of 6V' DWBC \cite{DF20V}, with the choice:
\begin{equation}\label{20v3point}
\eta=\frac{\pi}{8},\qquad u=\frac{\pi}{8},\qquad v=-\frac{\pi}{2},\qquad \rho_o=\rho_e=\sqrt{2}
\end{equation}

\begin{figure}
\begin{center}
\includegraphics[width=9cm]{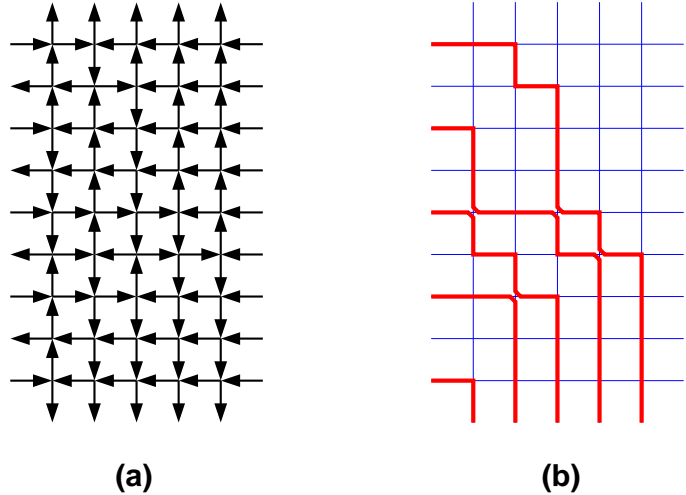}
\end{center}
\caption{\small (a) A sample configuration of the 6V' model (for $n=5$) and (b) its reformulation in terms of osculating paths.}
\label{fig:oscu6vp}
\end{figure}

Like in the 6V-DWBC case, the configurations of the model may be rephrased in terms of osculating paths. Using the same recipe, we see that configurations are in bijection
with families of $n$ osculating paths, starting at odd horizontal edges along the W boundary, and ending at all vertical edges along the S boundary.
For illustration, we have represented in Fig.~\ref{fig:oscu6vp} a sample 6V'  configuration both in the arrow (a) and osculating path (b) formulations.

The U-turn 6V/6V' model thermodynamic free energy was derived in Ref.~\cite{RIBKOR}, and arctic curves were
derived in the VSASM case \cite{DFLAP} and in the particular free fermion case corresponding to $\eta=\frac{\pi}{4}$, $u=0$, $v=-\frac{\pi}{2}$ \cite{PRarctic}.

\subsubsection{20V model with DWBC3}\label{20vpresec}

\begin{figure}
\begin{center}
\includegraphics[width=14cm]{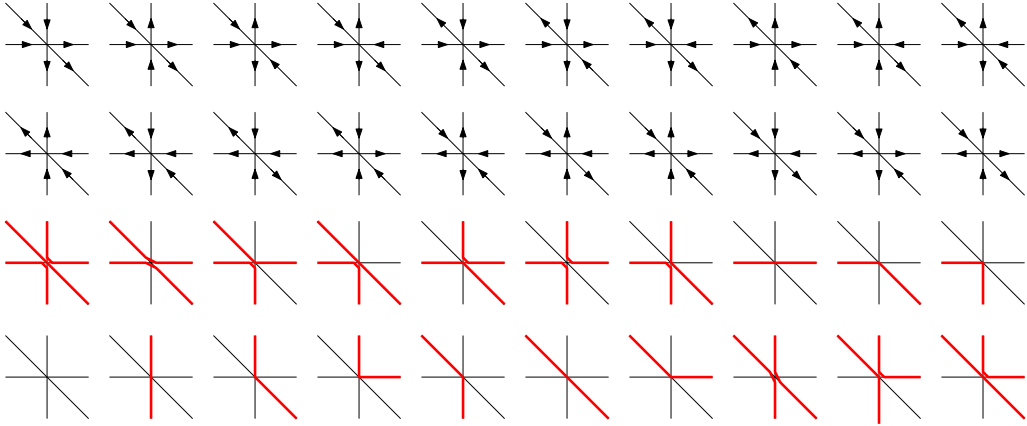}
\end{center}
\caption{\small The twenty vertex environments obeying the ice rule on the triangular lattice (top two rows) together with their osculating Schr\"oder path reformulation (bottom two rows).}
\label{fig:vingt}
\end{figure}

The 20V model is a two-dimensional ice-type model defined on the triangular lattice. As in the 6V case, edges are oriented in such a way that at each vertex there are exactly three incoming and three outgoing arrows. This gives rise to the ${6\choose 3}=20$ local vertex configurations depicted in Fig.~\ref{fig:vingt}. Recently this model was considered with special boundary conditions \cite{DFGUI} emulating DWBC on some particular domains. For simplicity, the triangular lattice is represented with vertices in $\Z^2$, and edges of the square lattice are supplemented by the second diagonal of each square face. Edges are accordingly called horizontal, vertical and diagonal. 

In Ref. \cite{DFGUI}, four types of boundary conditions (DWBC1,2,3,4) were considered on a $n\times n$ square grid in this representation, with remarkable combinatorial properties. 
The DWBC1,2 are closest to the 6V-DWBC, and correspond to arrows entering the domain on the W and E boundaries, and exiting on the N and S boundaries, with a particular choice of the 
NW and SE corner diagonal edges as belonging to the W and S boundaries respectively (DWBC1) or to the N and E (DWBC2). The DWBC3 is a more relaxed version of DWBC, where  only the horizontal arrows point toward the domain on the W boundary, and only vertical arrows point outward on the S boundary, while all other arrows point outward on W and N, and inward on S and E. In \cite{DFGUI}, a family of pentagonal extensions of the grid was considered, and the corresponding 20V configurations were conjectured to correspond to the domino tilings of special domains, viewed as truncations of the Aztec Triangle.  The conjecture was proved in \cite{DF20V} for the maximal extension, namely the 20V model with DWBC3 on the quadrangle $\cQ_n$ of shape $n\times n\times (2n-1)\times n$ (see Fig.~\ref{fig:osc20v} (a) for an illustration), whose partition function was shown to be identical to that of domino tilings of the Aztec triangle $\cT_n$ (see Fig.~\ref{fig:dt0} (a)).  In the present paper, we shall concentrate on this model.

\begin{figure}
\begin{center}
\includegraphics[width=14cm]{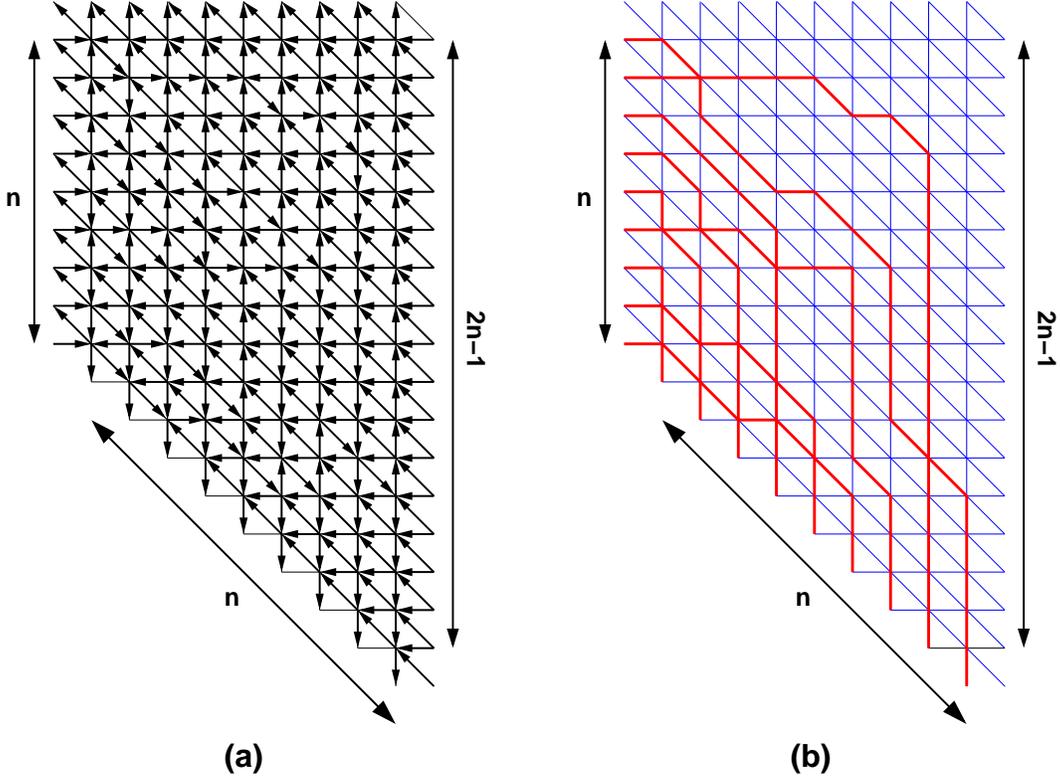}
\end{center}
\caption{\small (a) A sample configuration of the 20V model with DWBC3 on the quadrangle $\cQ_n$ (with $n=9$ here) and (b) its osculating Schr\"oder path reformulation.}
\label{fig:osc20v}
\end{figure}

\begin{figure}
\begin{center}
\includegraphics[width=14cm]{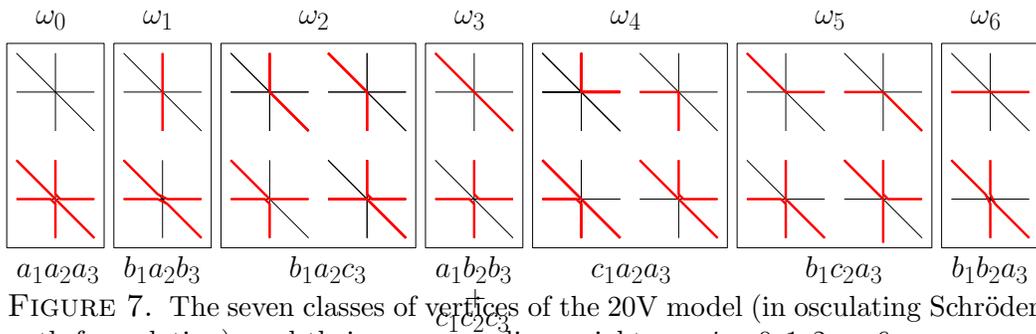}
\end{center}
\caption{\small The seven classes of vertices of the 20V model (in osculating Schr\"oder path formulation), and their corresponding weights $\omega_i$, $i=0,1,2,...,6$.}
\label{fig:weight20v}
\end{figure}

Like in the 6V case, we may rephrase the arrow configurations of the 20V model in terms of osculating paths with horizontal, vertical and diagonal steps along the corresponding edges of the lattice (these are usually called Schr\"oder paths). This is done similarly by picking a base orientation (right, down, and diagonal down and right) of all the edges of the lattice, and marking only those edges of a given configuration of the 20V model that agree with it. The selected edges are assembled again into non-intersecting, but possibly kissing paths travelling to the right and down. We have represented in Fig.~\ref{fig:vingt} (bottom two rows) the 20 local path configurations at a vertex corresponding to the 20 arrow configurations
(top two rows).
In the osculating Schr\"oder path formulation, the DWBC3 on $\cQ_n$ gives rise to families of $n$ osculating paths starting at the $n$ odd horizontal edges along the W boundary, and ending at the $n$ vertical edges of the diagonal SW boundary, as displayed in Fig.~\ref{fig:osc20v} (b).

As detailed in \cite{DFGUI,Kel}, the model receives integrable weights inherited from those of the 6V model upon resolving the triple intersections of spectral parameter lines at each vertex into three simple intersections corresponding to three 6V models on three distinct lattices. Integrability was used in \cite{DF20V} to transform the partition function of the 20V DWBC3 model into that of a  6V' model, for a particular normalization of spectral parameters of the 20V model. With this normalization, the seven local vertex weights corresponding to the dictionary of Fig.~\ref{fig:weight20v} read respectively:
\begin{eqnarray}\label{weights20V}
\omega_0&=&\nu\, \sin(u-v+\eta)\, \sin(\eta-u-v)\,\sin(2u+2\eta)\nonumber \\
\omega_1&=&\nu\, \sin(u-v-\eta)\, \sin(-u-v-\eta)\,\sin(2u+2\eta)\nonumber \\
\omega_2&=&\nu\,  \sin(u-v-\eta)\, \,\sin(2u+2\eta)\,\sin(2\eta)\nonumber \\
\omega_3&=&\nu\, \{ \sin^3(2\eta)+\sin(u-v+\eta)\, \sin(-u-v-\eta)\,\sin(2u) \}\nonumber \\
\omega_4&=&\nu\,\sin(2u+2\eta) \, \sin(\eta-u-v)\,\sin(2\eta) \nonumber \\
\omega_5&=&\nu\, \sin(u-v-\eta)\,\sin(\eta-u-v)\,\sin(2\eta) \nonumber \\
\omega_6&=&\nu\, \sin(u-v-\eta)\,\sin(\eta-u-v)\,\sin(2u) 
\end{eqnarray}
where again the fixed overall factor $\nu>0$ emphazises the projective nature of the weights. Note that each vertex
is the intersection of three lines (horizontal, vertical, diagonal) each of which carries a spectral parameter ($\eta+u,\ v,\ -u$
respectively).
The domain of parameters ensuring positivity of the weights is :
\begin{equation}\label{domain20}
0<u<\frac{\pi}{2}-\eta \qquad  \eta <u-v<\pi-\eta \qquad \eta<-u-v<\pi-\eta \qquad 0<\eta<\frac{\pi}{2} 
\end{equation}
(Note the similarity with the domain \eqref{domain6} for the 6V' model, the only extra condition being that $u>0$.).

Note the existence of a combinatorial point where the weights are uniform and all equal to $1$: 
\begin{equation}\label{combipoint20v}\eta=\frac{\pi}{8},\qquad u=\eta=\frac{\pi}{8},\qquad v=-4\eta=-\frac{\pi}{2}, \qquad  \nu=\sqrt{2} . \end{equation}
identical to the 20V-DWBC3 point of the 6V' model, where the partition functions of both models are related \cite{DF20V}.

\subsubsection{Domino Tilings of the Aztec Triangle}\label{dtrefsec}

\begin{figure}
\begin{center}
\includegraphics[width=16cm]{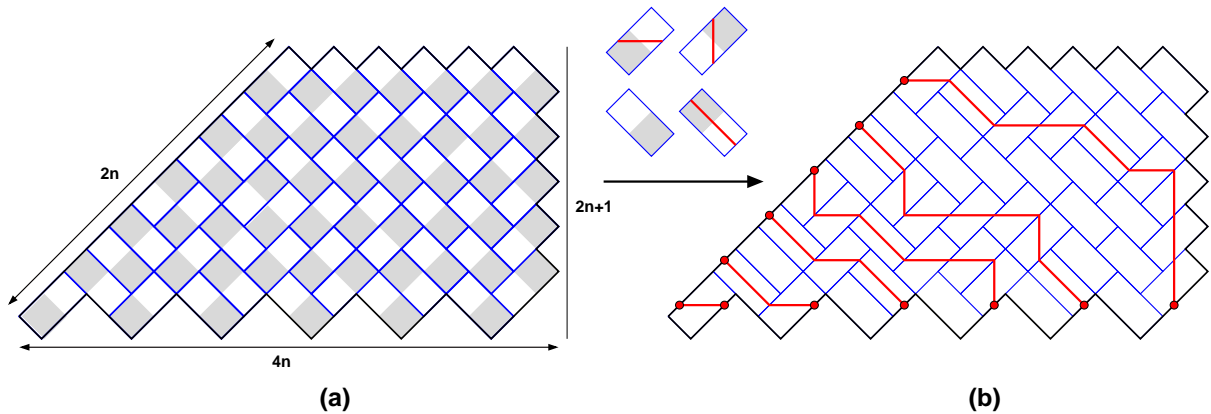}
\end{center}
\caption{\small (a) A sample domino tiling of the Aztec triangle $\cT_n$ (here for $n=6$) and (b) its non-intersecting Schr\"oder path formulation. The dictionary for the path steps (horizontal, vertical, empty, diagonal) is indicated for each of the four (bicolored) domino configurations. }
\label{fig:dt0}
\end{figure}
Our fourth class of objects is the tiling configurations by means of $2\times 1$ dominos of the ``Aztec Triangle" of order $n$ \cite{DFGUI,DF20V}, denoted $\cT_n$, depicted in Fig.~\ref{fig:dt0} (a). 
The identity between the number of 20V-DWBC3 configurations on $\cQ_n$ and the number of domino tilings of the Aztec triangle $\cT_n$  was conjectured in \cite{DFGUI} and proved in \cite{DF20V}. In Sect.~\ref{DTsec} below, we will make use of this correspondence to determine the limit shape of typical domino tilings of $\cT_n$ for large $n$.

It proves useful to rephrase the domino tiling problem in terms of non-intersecting lattice paths, as indicated in Fig.~\ref{fig:dt0} (b), where the indicated dictionary between bi-colored dominos and path steps has been used to reexpress bijectively the tiling configuration into a family of non-intersecting lattice paths with fixed ends on the diagonal NW and S boundaries of the domain. As indicated, the paths may have horizontal, vertical and diagonal steps and are therefore non-intersecting Schr\"oder paths.

\subsection{Tangent Method: combining one-point functions and paths}\label{sectan}

This section details how the Tangent Method of \cite{COSPO} works and how we are going to apply it to the four models studied in this paper: the 6V-DWBC, 6V', 20V-DWBC3 and finally the Domino Tiling of the Aztec Triangle, all expressed in the (possibly osculating) path formulation.

\subsubsection{The Tangent Method}\label{sectanbis}

\begin{figure}
\begin{center}
\includegraphics[width=16cm]{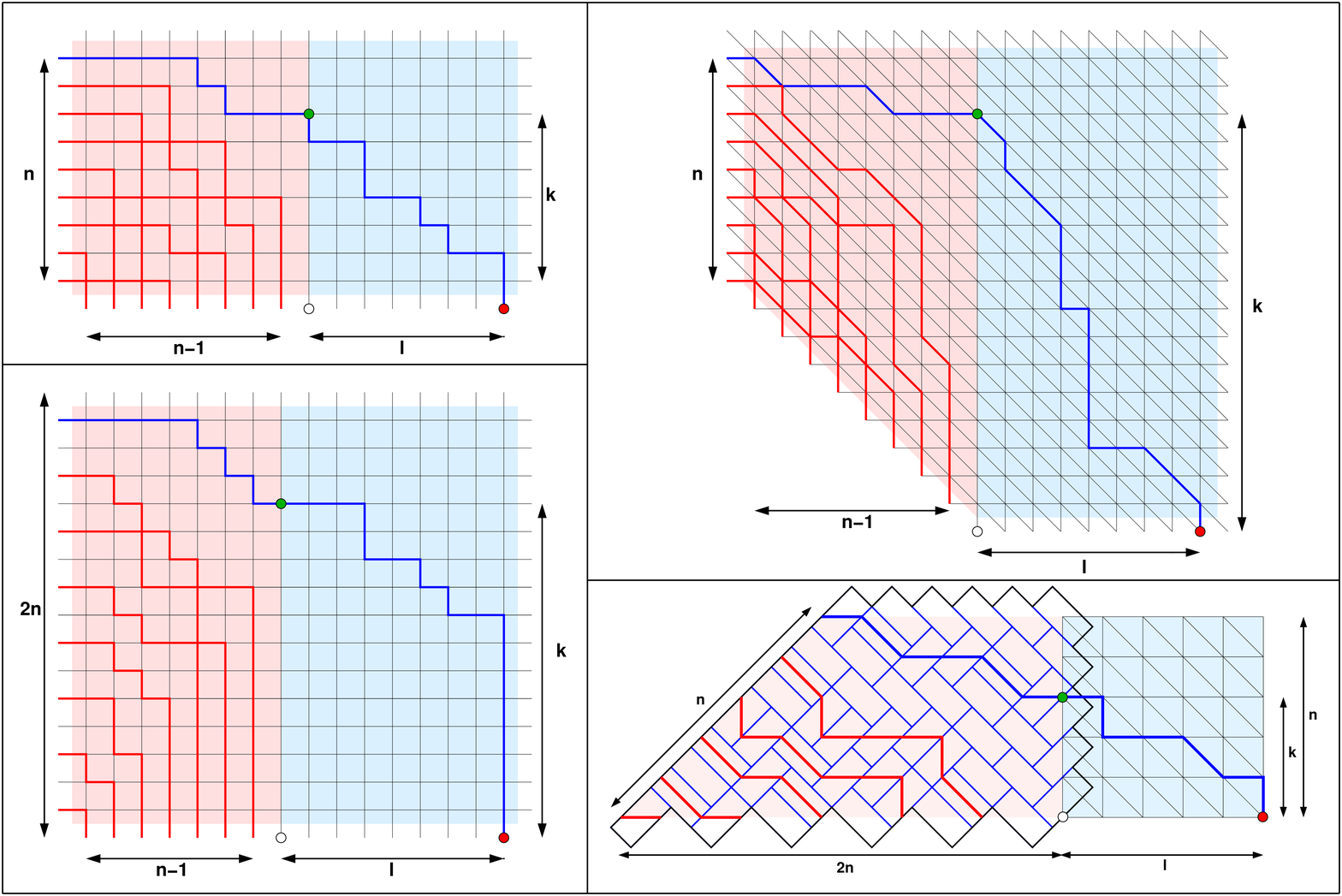}
\end{center}
\caption{\small Application of the Tangent Method to determine the NE branch of the arctic curve, illustrated for the four models studied in this paper: 6V-DWBC (top left), 6V' (bottom left), 20V-DWBC3 (top right) and Domino Tilings of the Aztec Triangle (bottom right). In all cases, the endpoint of the topmost path is displaced from its original position (white dot) to the right, at some distance $\ell$ (red dot). The partition function splits into the modified partition function of the model $Z_{n,k}$ with exit point at position $k$ (pink domain) and that, $Y_{k,\ell}$, of a single path from the exit point to the displaced endpoint with the same ambient weights (light blue domain).
The Tangent Method uses the most likely position $k=k(\ell)$ namely the one giving the largest contribution to the total partition function. The relevant portion of arctic curve is given by the envelope of the family of lines through $(\ell,0)$ and $(0,k(\ell))$ (the green and red points), in rescaled coordinates with the origin at the SE corner of the original domain. All dimensions are expressed in units of the underlying lattice grid.}
\label{fig:alltgt}
\end{figure}

As explained in the Introduction, the Tangent Method consists in finding the most likely exit point from the original domain of the topmost path, given that its end has
been displaced away from the domain. To determine this point, we consider the full partition function $\Sigma_{n,\ell}$ of the model, which is made of two pieces (corresponding respectively to the pink and light blue domains in Fig.~\ref{fig:alltgt}):
\begin{itemize}
\item{} The modified partition function $Z_{n,k}$ for the set of $n$ weighted paths in the original domain, with the (topmost) $n$-th path constrained to exit the domain along the E border at a fixed height $k$ (green dot in Fig.~\ref{fig:alltgt}), normalized into the ``refined one-point function" $H_{n,k}=Z_{n,k}/Z_n$.
\item{} The partition function $Y_{k,\ell}$ of a single weighted path constrained to start at the previous exit point and end at a fixed endpoint 
(displaced at distance $\ell$ from its original position, second green dot in Fig.~\ref{fig:alltgt}).
\end{itemize}
The full partition function reads:  $\Sigma_{n,\ell}= \sum_{k=1}^{\mu n} H_{n,k}\, Y_{k,\ell}$, where $\mu=1$ for the 6V-DWBC and the Domino Tiling models, and $\mu=2$ in the other
models. Note that all weights (including those of the single path) are those of the underlying vertex model; 
in particular, the vertices not visited by the single path (in the light blue zones of Fig.~\ref{fig:alltgt}) receive the weight of the empty vertex configuration.

Next we go to the large $n=N$ scaling limit and use large $N$ estimates of both partition functions to find the leading contribution to the sum in $\Sigma_{n,\ell}$.
More precisely, setting $n=N$, $\ell=\lambda N$, the limiting solution of the saddle-point approximation to the sum, in the form of some 
function $\kappa(\lambda)$ where $k(\ell)=\mu N \kappa(\lambda)$ is the most likely position of the exit point.  The (rescaled) arctic curve is then obtained as the envelope of the family of lines through the most likely exit point and the fixed endpoint, both functions of the parameter $\lambda$.  More precisely, we must estimate the large $N$ behavior of
the total partition function $\Sigma_{n,\ell}$:
$$\Sigma_{N,\lambda N} \simeq \mu N\int_{0}^{1} d\kappa H_{N,\mu \kappa N} Y_{\mu\kappa N,\lambda N}$$
where we have replaced the summation by an integral over the rescaled variable $\kappa=k/(\mu N)$. In Sections \ref{6vsec},\ref{6vpsec},\ref{20vsec} and \ref{DTsec} below, we work out the explicit asymptotics of both
functions in the integrand, in the form $H_{N,\mu\kappa N}\simeq e^{NS_H(\kappa)}$ and $Y_{\mu\kappa N,\lambda N}\simeq e^{NS_Y(\kappa)}$. The leading contribution 
to the integral comes from the
solution $\kappa=\kappa(\lambda)$ to the saddle-point equation $\partial_\kappa S_H+\partial_\kappa S_Y=0$. This gives the most likely exit point $\big(0,\mu\kappa(\lambda)\big)$ (in rescaled variables).
The tangent line in rescaled variables is the line through $\big(0,\mu\kappa(\lambda)\big)$ and $(\lambda,0)$, with equation 
\begin{equation}\label{tgteq}
y+A x-B=0, \qquad A=\frac{\mu\kappa(\lambda)}{\lambda}, \quad B=\mu \kappa(\lambda)
\end{equation}
As we shall see, the family of tangent lines is best described in terms of the parameter $\xi$ (the deviation from uniform vertical spectral parameter in the last (E-most) column $v_n=v+\xi$). In particular the relationship between $\lambda$ and $\kappa(\lambda)$ takes the parametric form: $(\kappa(\lambda),\lambda)=(\kappa[\xi],\lambda[\xi])$, for $\xi\in I$, 
$I$ an interval determined by the conditions that $\lambda[\xi]>0$ and $\kappa[\xi]\in [0,1]$.
The envelope of the family of lines 
\begin{equation}\label{tgtfamily}
F_\xi(x,y):=y+A[\xi]x-B[\xi]=0\end{equation}
is determined as the solution of the linear system $F_\xi(x,y)=\partial_\xi F(x,y)=0$, and gives rise
to the parametric equations for the arctic curve:
\begin{equation}
\label{acurve}
x=X[\xi]:=\frac{B'[\xi]}{A'[\xi]},\qquad y=Y[\xi]:=B[\xi]-\frac{A[\xi]}{A'[\xi]}B'[\xi] , \qquad (\xi\in I)
\end{equation}

By the geometry of the problem, only a portion of the arctic curve can be obtained in this way: moving the exit point to the right along the line through all other exit points covers a portion of arctic curve between a point of tangency to the E vertical border of the original domain (when the endpoint tends to its original position) and a point of tangency to the horizontal N border of the domain (when the endpoint tends to infinity on the right along the line), or equivalently corresponding to the slope $A[\xi]\in [0,\infty)$.
This condition was used to restrict the domain of the variable $\xi\in I$. This portion of arctic curve is on the NE corner of the domain, and we shall refer to it as the NE branch of the arctic curve.

The case of the Domino Tiling of the Aztec Triangle is simpler: as a (free fermion) dimer model, it is expected on general grounds \cite{KOS} to have an analytic arctic curve, equal to the analytic continuation of its NE branch. As we shall see, the cases of 6V-DWBC, 6V' and 20V-DWBC3 are more involved, and lead in general to non-analytic arctic curves.

\subsubsection{Other branches}\label{obsec}

To reach other portions of the arctic curve, we will have to resort to various tricks, all based on the same principle: we switch to a different interpretation of the configurations of the original model, to express them  in terms of different families of paths, to which the Tangent Method can be applied again. 

\noindent{\bf 6V-DWBC and 6V' cases.}

\begin{figure}
\begin{center}
\includegraphics[width=14cm]{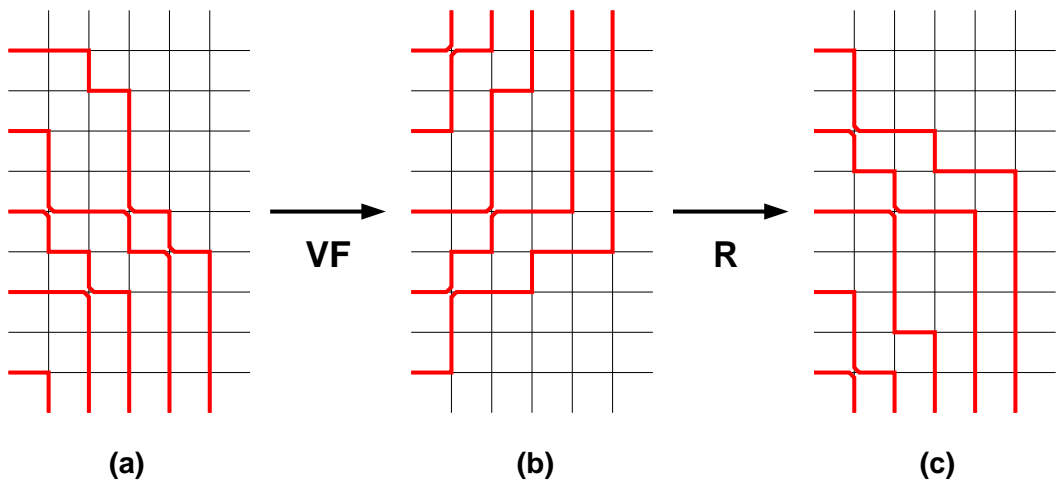}
\end{center}
\caption{\small From NE to SE branch in the 6V' model: (a) configuration of the 6V' model (b) after application of the Vertical Flip (VF) (c) after application of the Reflection (R), leading to another 6V' configuration, with weights of $a_i$ and $b_i$ types interchanged for $i=e,o$.}
\label{fig:flip6vp}
\end{figure}

In the case of the 6V-DWBC/6V' model arctic curves, we have access to the SE branch by reinterpreting the 6V configurations in terms of paths with the same starting points (every point/every other point along the W vertical border) but with endpoints along the N border (see Fig.~\ref{fig:flip6vp} (a-b) for an illustration in the 6V' case). This corresponds to redefining the base orientation of edges (and direction of travel of the paths) to be to the right and up: we call this transformation on the paths Vertical Flip ($\bf VF$). The osculation at each vertex must be {\it redefined} so that all paths now go horizontally right and vertically up. The SE branch of the arctic curve is obtained by applying the Tangent Method to this new family of paths. The easiest way to do so is to reflect the picture w.r.t. a horizontal line, so that the setting is that of the 6V' model again: we call this transformation Reflection ($\bf R$) as shown in Fig.~\ref{fig:flip6vp} (c). 
The net effect of the composition of $\bf VF$ followed by $\bf R$ on the model is simply to interchange the a and b type weights, a transformation also implemented by the involution $*$ acting on the spectral parameters as follows: 
\begin{eqnarray}
&&u\to u^*=\pi-u \qquad\qquad\qquad ({\rm 6V-DWBC})\label{6vstar}\\ 
&&(u,v)\to (u^*=-u,v^*=-\pi-v)\qquad ({\rm 6V'})\label{6vpstar}
\end{eqnarray} 
This gives rise\footnote{Here and in the following the superscript $*$ indicates that the 
corresponding quantity is obtained by changing $u\to u^*$ (6V-DWBC) or $(u,v)\to (u^*,v^*)$ (6V').}
 to the new weights $a^*=b,b^*=a,c^*=c$ (6V-DWBC) or $a_i^*=b_i,b_i^*=a_i, c_i^*=c_i$, $i=o,e$ (6V'). 
 
The ``upside-down" one-point function must also be reinterpreted as $H_{n,k}=H_{n,\mu n-k}^*$ where the latter is computed with the new transformed weights. Similarly, the path partition function $Y_{k,\ell}$
with starting point $(0,k)$ and endpoint $(\ell,0)$ is reinterpreted as the partition function $Y_{\mu n-k,\ell}^*$ with the new weights.  
More precisely setting the origin of the rescaled domain 
at the SE corner of the domain, the vertices of the rescaled domain are: SE:(0,0), NE:(0,$\mu$), NW:(-1,$\mu$), SW:(-1,0). The large $n=N$
optimization problem leading to the most likely exit point $(0,\mu\kappa(\lambda))$ leads now to the most likely exit point $(0,\mu(1-\kappa^*(\lambda)))$ and the associated family of rescaled lines. 
The SE branch of arctic curve is obtained by reflecting back the envelope of this family, and effectively amounts to applying $*$ to the NE branch and reflecting it w.r.t. 
the line $y=\mu/2$, namely applying the transformation 
\begin{equation} \label{flip6v6vp} (x,y)\mapsto (x,\mu-y)
\end{equation}
with $\mu=1$ for the 6V-DWBC model, and $\mu=2$ for the 6V' model.

\begin{figure}
\begin{center}
\includegraphics[width=16cm]{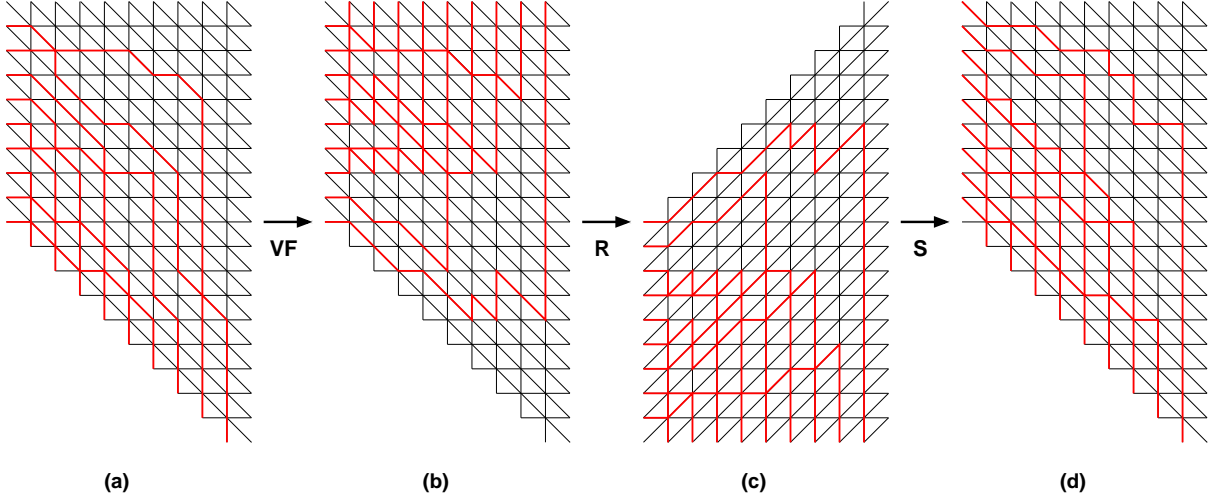}
\end{center}
\caption{\small The shear trick: (a) configuration of the 20V-DWBC3 (b) after application of the Vertical Flip (VF) (c) after application of the Reflection (R) (d) after the Shear transformation (S). }
\label{fig:shear}
\end{figure}

\noindent{\bf 20V-DWBC3 case.}
In the case of the 20V-DWBC3 model illustrated in Fig.~\ref{fig:shear} (a), the same idea leading to the SE branch must be adapted (by adapting the ``shear trick" devised in Ref.~\cite{BDFG}). More precisely, as in the 6V-DWBC and 6V' cases, we first redefine the base edge orientation so that horizontal/diagonal edges point right/down, but vertical edges point up. Alternatively, compared to the original base orientation, this simply interchanges vertical edges which are occupied by path steps with vertical edges which are empty and vice versa (like in the 6V,6V' cases, we call this $\bf VF$=Vertical Flip, see Fig.~\ref{fig:shear} (b)). The osculation at each vertex must be redefined so that all paths now go horizontally right, diagonally down and vertically up: in particular the paths now end on the $n$ vertical edges of the N boundary (see Fig.~\ref{fig:shear} (b)).
To match this with a 20V configuration of $\cQ_n$, we must reflect the configuration w.r.t. a horizontal line (we call this again $\bf R$=Reflection, see Fig.~\ref{fig:shear} (c)). However,  the quadrangular domain 
$\cQ_n$ is not invariant under horizontal reflection: to recover it, we apply a shear transformation as indicated in Fig.~\ref{fig:shear} (d) (we call this $\bf S$=Shear). More precisely setting the origin of the rescaled domain 
at the SE corner of the domain, the vertices of the rescaled quadrangle $\cQ_N/N$ are: SE:(0,0), NE:(0,2), NW:(-1,2), SW:(-1,1). 
Applying successively (in rescaled variables $(x,y)$) the reflection $(x,y)\to (x,-y)$, shear $(x,y)\to (x,y-x)$, and finally translation by $(0,2)$ leaves the domain invariant, 
but effectively flips the orientation of the vertical edges in the 20V-DWBC3 configuration. Note that the final configuration after application of $\bf VF,\ R, \ S$ (Fig.~\ref{fig:shear} (d))
is slightly different from a 20V-DWBC3 configuration, as all the starting steps along the W boundary are diagonal, as opposed to horizontal. This clearly makes no difference in the case of uniform weights, however for non-uniform weights this changes the weights along the W boundary. We argue nonetheless that this mild boundary effect does not affect the asymptotic behavior of bulk quantities as all other weights are the same in both situations. In particular, we expect the arctic curve to be the same in the original 20V-DWBC3 model with
horizontal starting steps along the W boundary, and in the modified one, where all starting steps are diagonal.

Let us now examine the fate of the local vertex environments of Fig.~\ref{fig:weight20v} under the sequence of transformations $\bf VF,\  R, \ S$. It is easy to see that under this transformation the types of vertices are mapped as follows: $\omega_0\leftrightarrow \omega_1$, $\omega_2\leftrightarrow \omega_4$, while all other types are preserved. For illustration, the top left vertex of type $\omega_2$ of Fig.~\ref{fig:weight20v} is successively transformed into the top left vertex of type $\omega_4$ as follows:
$$
\raisebox{-1.cm}{\hbox{\epsfxsize=14cm\epsfbox{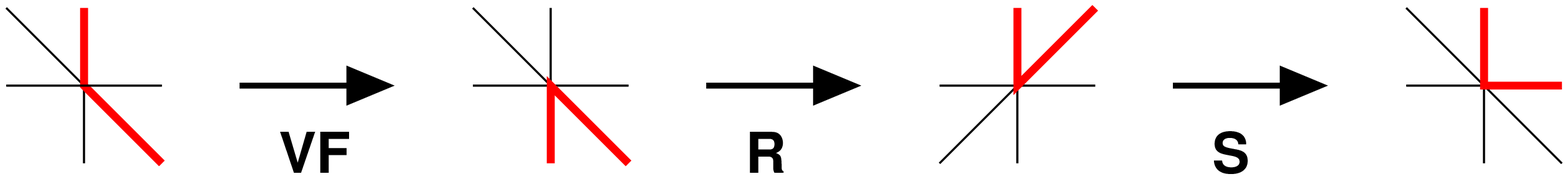}}}
$$
We finally note that the involution $*$ which maps the weights
$(\omega_0,\omega_1,\omega_2,\omega_3,\omega_4,\omega_5,\omega_6)\mapsto (\omega_1,\omega_0,\omega_4,\omega_3,\omega_2,\omega_5,\omega_6)$
is simply given by 
\begin{equation}\label{20vstar}(u,v)\mapsto (u^*,v^*)=(u,-v-\pi)\qquad ({\rm 20V-DWBC3}) .
\end{equation}
As before, we have to reinterpret $H_{n,k}=H_{n,2n-k}^*$ and $Y_{k,\ell}=Y_{2n-k,\ell}^*$ leading to the 
most likely exit point $2N(1-\kappa^*(\lambda))$ in rescaled variables $\kappa=k/(2N)$ and $\lambda=\ell/N$.
The SE branch is finally obtained by applying the reflection/shear/translation $(x,y)\mapsto (x,2-x-y)$ to the NE one after applying $*$, namely changing $v\to -v-\pi$.

This gives access to the SE branch in all 6V-DWBC, 6V' and 20V-DWBC3 cases.  Except in the free fermion cases, where arctic curves are expected to be analytic, 
we have no prediction for other portions of arctic curve when they exist.

\section{6V Model}\label{6vsec}

\subsection{Partition function and one-point function}

\subsubsection{Inhomogeneous partition function}

A general result \cite{IKdet} provides a determinant formula for the partition function $Z_n^{6V}[\bu,\bv]$ of the inhomogeneous 6V-DWBC model, with horizontal/vertical spectral parameters $\bu=u_1,u_2,...,u_n$/$\bv=v_1,v_2,...,v_n$.
\begin{thm}
Let
$$m(u,v):=m(u-v):=\frac{1}{\sin(u-v+\eta)\sin(u-v-\eta)}$$
The full inhomogeneous 6V-DWBC partition function reads:
\begin{equation}\label{sixinho}
Z_{n}^{6V}[\bu,\bv]=\rho^{n^2}\sin^n(\eta)\det_{1\leq i,j \leq n}\big(m(u_i,v_j)\big)
\times 
\frac{\prod_{i,j=1}^n 
\sin(u_i-v_j+\eta)\sin(u_i-v_j-\eta)}{\prod_{1\leq i<j \leq n}\sin(u_i-u_j)\,
\sin(v_j-v_i)}
\end{equation}
\end{thm}

\subsubsection{Homogeneous limit}

The homogeneous limit $Z_n^{6V}[u-v]$ of the inhomogeneous partition function $Z_n^{6V}[\bu,\bv]$ in which $u_i\to u$ and $v_i\to v$ for all $i$ involves the quantity
\begin{equation}\label{deltahomo6v}
\Delta_n[u-v]:=\lim_{u_1,u_2,...,u_n\to u\atop v_1,v_2,...,v_n\to v} \frac{\det_{1\leq i,j \leq n}\big(m(u_i,v_j)\big)}{\prod_{1\leq i<j \leq n}(u_i-u_j)(v_j-v_i)}=: \frac{1}{\prod_{i=1}^{n-1} (i!)^2} D_n[u-v]
\end{equation}
Using Taylor expansion of rows and columns leads to the determinant
$$D_n[u]=\det_{0\leq i,j\leq n-1}\left(\partial_u^{i+j}m(u)\right)$$

This determinant obeys a simple quadratic relation as a consequence of Pl\"ucker/Desnanot-Jacobi relations (up to some permutation of rows and columns) 
relating a determinant to some of its minors of size 1 and 2 less, summarized in the following lemma. 

\begin{lemma}\label{desnajac}
Given an $n+1\times n+1$ square matrix $M$, its determinant $|M|$ and minors $|M|_{a}^{b}$ (with row $a$ and column $b$ removed), and $|M|_{a_1,a_2}^{b_1,b_2}$ (with rows $a_1,a_2$ and columns $b_1,b_2$ removed are related via:
$$ |M| \times |M|_{n,n+1}^{n,n+1} =|M|_n^n \times |M|_{n+1}^{n+1}-|M|_{n+1}^n \times |M|_n^{n+1} $$
\end{lemma}

Applying this to the matrix $M=\big(\partial_u^{i+j}m(u)\big)_{0\leq i,j\leq n}$, we easily get
$$ D_{n+1}[u]\, D_{n-1}[u]= D_{n}[u] \, \partial_u^2 D_{n}[u]- (\partial_u D_{n}[u])^2 =D_n[u]^2 \partial_u^2 {\rm Log}(D_n[u]) $$
As a direct consequence,  we have
\begin{thm}\label{6vpartthm}
The quantity $\Delta_n[u]$ obeys the following recursion relation for all $n\geq 1$:
\begin{equation}\label{6vdelta} \frac{\Delta_{n+1}[u] \, \Delta_{n-1}[u] }{\Delta_n[u]^2} = \frac{1}{n^2} \partial_u^2 {\rm Log}(\Delta_n[u]) \end{equation}
with the convention that $\Delta_0[u]=1$.
\end{thm}
Note that this relation determines $\Delta_n[u]$ recursively, from the initial data $\Delta_0[u]=1$ and $\Delta_1[u]=m[u]$.
Finally the homogeneous partition function $Z_n^{6V}[u,v]$ is expressed as
\begin{equation}\label{6Vpart} \frac{Z_n^{6V}[u-v]}{\sin^n(2\eta)}= \Delta_n[u-v] \, \left(\rho \sin(u-v+\eta)\sin(u-v-\eta)\right)^{n^2} \end{equation}

\subsubsection{One-point function}

We now consider a slightly more general limit, in which we take $u_1,u_2,...,u_n\to u$ and $v_1,v_2,...,v_{n-1}\to v$ but the last vertical spectral parameter is kept arbitrary,
say $v_n=w=v+\xi$. The corresponding partition function $Z_n^{6V}[u-v;\xi]$ is again obtained as a limit of the inhomogeneous formula \eqref{sixinho}.
We have
\begin{eqnarray} \frac{Z_n^{6V}[u-v;\xi]}{\sin^n(2\eta)}&=& \Delta_n[u-v;\xi] \ \left(\rho \sin(u-v+\eta)\sin(u-v-\eta)\right)^{n(n-1)} \nonumber \\
&&\qquad\qquad \times \,\left(\rho \sin(u-v-\xi+\eta)\sin(u-v-\xi-\eta)\right)^{n} \label{ref6vZ}
\end{eqnarray}
in terms of a function
$$\Delta_n[u-v;\xi]:= \lim_{{u_1,u_2,...,u_n\to u\atop v_1,v_2,...,v_{n-1}\to v}\atop v_n\to v+\xi} \frac{\det_{1\leq i,j \leq n}\big(m(u_i-v_j)\big)}{\prod_{1\leq i<j \leq n}\sin(u_i-u_j)\sin(v_j-v_i)}=: \frac{(-1)^{n-1}\, (n-1)!}{\sin^{n-1}(\xi)\, \prod_{i=1}^{n-1} (i!)^2} D_n[u-v;\xi]$$
where
\begin{equation}\label{6vDN}
D_n[u;\xi]=\det\left( \{\partial_u^{i+j} m(u)\}_{0\leq i\leq n-1\atop 0\leq j\leq n-2} \Bigg\vert \{\partial_u^i m(u-\xi) \}_{0\leq i\leq n-1}\right) 
\end{equation}
identical to $D_n[u]$ except for its last column.

We define the ``one-point function" as the ratio
$$H_n^{6V}[u;\xi]:= \frac{Z_n^{6V}[u;\xi]}{Z_n^{6V}[u]}= \frac{\Delta_n[u;\xi]}{\Delta_n[u]} \, \left(\frac{\sin(u-\xi+\eta)\sin(u-\xi-\eta)}{\sin(u+\eta)\sin(u-\eta)}\right)^{n}
$$

Applying again Lemma \ref{desnajac} this time with the matrix $M$ in the expression \eqref{6vDN}  for $D_{n+1}[u;\xi]$, we find that
\begin{equation}\label{newdenjac6v} D_{n+1}[u;\xi]\, D_{n-1}[u] = D_n[u] \, \partial_u D_n[u;\xi]- D_n[u;\xi]\, \partial_u D_n[u] \end{equation}
We also introduce the reduced one-point function
\begin{equation}\label{red6v}
H_n[u;\xi]:=(-1)^{n-1}\,(n-1)! \frac{D_n[u;\xi]}{D_n[u]}= \sin^{n-1}(\xi)\, \frac{\Delta_n[u;\xi]}{\Delta_n[u]} 
\end{equation}
in terms of which
\begin{equation}
\label{oneptfromred6v}
H_n^{6V}[u;\xi]=\sin(\xi) \,H_n[u;\xi]\, \left(\frac{\sin(u-\xi+\eta)\sin(u-\xi-\eta)}{\sin(u+\eta)\sin(u-\eta)\sin(\xi)}\right)^{n}
\end{equation}

The reduced one-point function is determined by the following relation, as a direct consequence of \eqref{newdenjac6v}.
\begin{thm}\label{6v1ptthm}
The reduced one-point function $H_n[u;\xi]$ of the 6V-DWBC model satisfies the following relation:
\begin{equation}\label{6vonept} \frac{H_{n+1}[u;\xi]}{H_n[u;\xi]}\, \frac{\Delta_{n+1}[u]\,\Delta_{n-1}[u]}{\Delta_n[u]^2} +\frac{1}{n} \partial_u {\rm Log}(H_n[u;\xi])=0 \end{equation}
\end{thm}

\subsection{Large $n$ limit: free energy and one-point function asymptotics}

In the following sections, we use the fact that the 6V weights depend on the quantity $u-v$ only. Without loss of generality we shall set $v=0$ from now on.

\subsubsection{Free energy}
In this section, we reproduce an argument of \cite{KORZIN,CP2009} leading to the large $n$ asymptotics of the partition function of the 6V-DWBC model.

The free energy per site $f^{6V}[u]$ of the 6V-DWBC model is defined via the large $n=N$ limit 
$$ f^{6V}[u] =-\lim_{N\to \infty} \frac{1}{N^2} {\rm Log}(Z_N^{6V}[u]) $$
or equivalently as the leading asymptotics $Z_N^{6V}[u]\simeq_{N\to\infty}  e^{-N^2f^{6V}[u]}$. Substituting this into \eqref{6Vpart}, we get:
\begin{equation}\label{sixvf}
f^{6V}[u] =f[u] -{\rm Log}\left(\rho \sin(u+\eta)\sin(u-\eta)\right)  \end{equation}
in terms of the limit
$$f[u]:=-\lim_{N\to \infty} \frac{1}{N^2} {\rm Log}(\Delta_N[u]) $$
Finally, substituting the large $N$ asymptotics $\Delta_N[u]\simeq_{N\to\infty} e^{-N^2 f[u]}$ into \eqref{6vdelta} yields the following differential equation (1D Liouville equation)
$$e^{-2 f[u]} +\partial_u^2 f[u]=0 $$
To fix the solution, let us derive some symmetry and some limit of $\Delta_N[u]$. 

First note that the reflection symmetry $u\to \pi-u$ from \eqref{sym6v}, together with the relation \eqref{sixvf} imply that
$$f[\pi-u]=f[u] .$$
Next let us consider the limit $u \to \eta$ of $\Delta_n[u]$. Setting $u=\eta+\epsilon$, we may perform the homogeneous limit \eqref{deltahomo6v} by setting $u_i=u+\epsilon x_i$,
$v_i=\epsilon y_i$ (recall we have set $v=0$) and taking all $x_i,y_i\to 0$. We have for small $\epsilon$
$$ m(u_i-v_j)= \frac{1}{\epsilon \sin(2\eta)\, (1+x_i-y_j)} +O(1) $$
Using the Cauchy determinant formula, we find that
$$\frac{\det_{1\leq i,j \leq n}(m(u_i-v_j))}{\prod_{1\leq i<j\leq n} (u_i-u_j)(v_j-v_i)}\simeq_{\epsilon\to 0} \frac{1}{\epsilon^{n^2}} \frac{\det\left( \frac{1}{1+x_i-y_j} \right) }{\prod_{i<j} (x_i-x_j)(y_j-y_i)} =\frac{1}{\epsilon^{n^2}}\prod_{i,j=1}^n \frac{1}{1+x_i-y_j}\to_{x_i\to 0\atop y_j\to 0} \frac{1}{\epsilon^{n^2}}$$
We deduce that 
$$\left\{\lim_{N\to \infty} -\frac{1}{N^2} {\rm Log}(\Delta_N[\eta+\epsilon])\right\}\Bigg\vert_{\epsilon\to 0} \simeq   {\rm Log}(\epsilon) \ \Rightarrow \ f[\eta+\epsilon]\vert_{\epsilon\to 0}\simeq {\rm Log}(\epsilon)$$

Defining $W[u]:=e^{f[u]}$, we find that $W$ satisfies the following conditions:
\begin{eqnarray*}
W[u]\, \partial_u^2W[u]-(\partial_u W[u])^2&=&\left\vert \begin{matrix} W[u] & \partial_u W[u]\\
\partial_u W[u] & \partial_u^2 W[u] \end{matrix}\right\vert=1\\
W[\pi-u]=W[u], \quad 
W[\eta]&=&0 .
\end{eqnarray*}
The constant Wronskian condition implies that $W[u]$ obeys a second order linear differential equation, with general solution of the form $W[u]=\frac{\sin(\al u +\beta)}{\al}$.
The parameter $\beta$ is fixed by the vanishing condition $W[\eta]=0$, and $\al$ by the symmetry condition $W[\pi-u]=W[u]$. We find the solution 
\begin{equation}\label{W6v}W[u]=\frac{\sin(\al(u-\eta))}{\al}
\end{equation}
where $\al(\pi-u-\eta)=\pi-\al(u-\eta)$, and finally
\begin{equation}\label{fsol}f[u]={\rm Log}\left(\frac{\sin(\al(u-\eta))}{\al} \right), \qquad  \al=\frac{\pi}{\pi-2\eta} \end{equation}
Substituting this into \eqref{sixvf}, we finally get the thermoynamic free energy of the 6V-DWBC model in the Disordered regime:
\begin{equation}\label{6vfreeenergy}
 f^{6V}[u]= -{\rm Log}\left(\frac{\al \rho \sin(u+\eta)\sin(u-\eta)}{\sin(\al(u-\eta))} \right)  \end{equation}
with $\al$ as in \eqref{fsol}.

\subsubsection{One-point function}
We present now a simplified version of the argument given in \cite{CP2009} to derive the asymptotics of the one-point function $H_n[u;\xi]$.
By eq.\eqref{6vonept} we may infer the large $n=N$ leading behavior of $H_n[u;\xi]$ to be:  
\begin{equation}\label{asymtoH6v} H_N[u;\xi]\simeq_{N\to \infty} e^{-N \psi[u;\xi]} \end{equation}
Substituting this into \eqref{6vonept} yields the differential equation:
$$ e^{-\psi[u;\xi]-2 f[u]} -\partial_u \psi[u;\xi]=0  \ \Rightarrow \ \partial_u e^{\psi[u;\xi]} =e^{-2 f[u]} $$
Using the result \eqref{fsol} for $f[u]$, this is easily integrated into
$$e^{\psi[u;\xi]}= c[\xi] -\al\cot(\al(u-\eta)) $$
for some integration constant $c[\xi]$ independent of $u$, and $\al$ as in \eqref{fsol}.
To fix the integration constant, let us consider the limit when $u-\xi-\eta\to 0$, by setting $\xi=u-\eta+\epsilon$ for a small $\epsilon\to 0$.
Noting that
$$\left\{\partial_u^i  m[u-\xi]\right\}\Big\vert_{\xi=u-\eta+\epsilon}=-\frac{i!}{\sin(2\eta)\, \epsilon^{i+1}}+O(\epsilon^{-i}) $$
we see that the determinant for $D_N[u;\xi]$ \eqref{6vDN} is dominated by the term in the last row and column $i=j=N-1$, resulting in the leading behavior
\begin{eqnarray*}\frac{D_N[u;\xi]}{D_N[u]} &\simeq& -(N-1)!\,\frac{D_{N-1}[u]}{\sin(2\eta) \, \epsilon^{N}\, D_N[u]}\simeq  (N-1)!\,\frac{W^{2N}}{\epsilon^N}\\
&\Rightarrow& H_N[u;u-\eta+\epsilon]\simeq_{\epsilon\to 0}  
\left(\frac{W^2}{\epsilon\, \sin(u-\eta)} \right)^N
\end{eqnarray*}
where we have used the defining relation \eqref{red6v} for $H_N[u;\xi]$.
Sending $\epsilon\to 0$, we conclude that
$$\lim_{\xi \to u-\eta} e^{\psi[u;\xi]}=0$$
This immediately gives $c[\xi]=\al \cot(\al\,\xi)$, and finally
\begin{equation}\label{psi6v}
e^{\psi[u;\xi]}=\frac{\al \sin(\al(u-\xi-\eta))}{\sin (\al\,\xi)\sin(\al(u-\eta))} 
\end{equation}
Collecting all the above results, we finally get the asymptotics of the 6V one-point function.

\begin{thm}\label{oneptasy6v}
The 6V-DWBC one-point function $H_n^{6V}[u;\xi]$ has the following large $n=N$ behavior:
\begin{eqnarray*}H_N^{6V}[u;\xi]&\simeq& e^{-N \psi^{6V}[u;\xi]} \\
\psi^{6V}[u;\xi]&=& -{\rm Log}\left(
\frac{\sin (\al\,\xi)\sin(\al(u-\eta))\sin(u-\xi+\eta)\sin(u-\xi-\eta)}{\al \sin(\al(u-\xi-\eta))\sin(\xi)\sin(u+\eta)\sin(u-\eta)} \right) \end{eqnarray*}
with $\al$ as in \eqref{fsol}.
\end{thm}
\begin{proof}
By the relation \eqref{oneptfromred6v}, we immediately get
$$\psi^{6V}[u]=\psi[u]+{\rm Log}\left(\frac{\sin(\xi)\sin(u+\eta)\sin(u-\eta)}{\sin(u-\xi+\eta)\sin(u-\xi-\eta)} \right) $$
and the Theorem follows.
\end{proof}

\subsubsection{Refined partition functions and one-point functions}

To apply the Tangent Method, we need the large $n,k$ asymptotics of the refined partition functions $Z_{n,k}^{6V}[u]$, $k=1,2,...,n$, defined as follows.  
Given a configuration of $n$ osculating paths contributing to $Z_{n}^{6V}[u]$ (with all horizontal spectal parameters equal to $u$ and all vertical ones to $0$), let us focus on the topmost path: let us record the first visit of this path to the east-most vertical line, say at the intersection with the $k$-th horizontal line from the bottom. Note that the path accesses the last vertical via a horizontal step, and ends with $k$ vertical steps until the east-most endpoint. We define the refined partition functions $Z_{n,k}^{6V}[u]$ to be the sum of all contributions in which the topmost path has these $k+1$ last steps. 


The quantities $Z_{n,k}^{6V}[u]$ turn out to be generated by the semi-inhomogeneous partition function $Z_n^{6V}[u;\xi]$
\eqref{ref6vZ}, for which the last vertical spectral 
parameter is replaced by $\xi$. 
Introducing relative weights $(\bar a,\bar b,\bar c)$ for the last column, as the following ratios of weights at $v=\xi$ by those at $v=0$:
$${\bar a} =\frac{\sin(u-\xi+\eta)}{\sin(u+\eta)},\quad {\bar b}= \frac{\sin(u-\xi-\eta)}{\sin(u-\eta)},\quad {\bar c}=1\ , $$
we have the following decomposition:
$$Z_n^{6V}[u;\xi]=\sum_{k=1}^n 
Z_{n,k}^{6V}[u]\, {\bar b}^{k-1} \,{\bar c}\, {\bar a}^{n-k}={\bar a}^{n-1} \sum_{k=1}^n \tau^{k-1} Z_{n,k}^{6V}[u]$$
in terms of a parameter
\begin{equation}\label{tauparam}
\tau:= \frac{\bar b}{\bar a} =\frac{\sin(u-\xi-\eta)\,\sin(u+\eta)}{\sin(u-\xi+\eta)\,\sin(u-\eta)}\end{equation}
In applying the Tangent Method, we truncate the topmost path after its last horizontal step (see an illustration in the top left of Fig.~\ref{fig:alltgt} in the pink domain). The effect of removing the last $k$ vertical steps and replacing them by empty edges is an overall multiplication by a factor $(1/c)(a/b)^{k-1}$ (as we have cut\footnote{Here we choose not to attach any weight to the end vertex at height $k$, as it will be part of the partition function of the single path treated in next section.} 
the turning c-type vertex and replaced the $k-1$ b-type vertices by a-type ones). This suggests to 
define {\it refined one-point functions} as the ratios
$$ H_{n,k}[u]:=\frac{1}{c}\left(\frac{a}{b}\right)^{k-1}\frac{Z_{n,k}^{6V}[u]}{Z_{n}^{6V}[u]} $$
The above relation between refined partition functions turns into the following relation between one-point function and refined one-point functions:
\begin{equation}\label{relaoneptonept} c \frac{H_n^{6V}[u;\xi]}{{\bar a}^{n-1}}=\sum_{k=1}^n t^{k-1} H_{n,k}[u] \end{equation}
where we have used a new parameter
\begin{equation}\label{tparam}
t=\frac{b}{a} \tau=\frac{\sin(u-\xi-\eta)}{\sin(u-\xi+\eta)}=:t_{6V}[\xi] \ .\end{equation}

Let us now consider the large $n=N$ scaling limit of $H_{n,k}[u]$ in which the ratio $\kappa=k/N$ is kept finite. Using the relation \eqref{relaoneptonept}
and the asymptotics of the one-point function $H_N^{6V}[u;\xi]$ of Theorem \ref{oneptasy6v}, 
we have at leading order as $N\to \infty$:
\begin{equation} H_{N,\kappa N}[u]\simeq \oint \frac{dt}{2i\pi t} e^{-N S_0(\kappa,t)}, \quad S_0(\kappa,t)= \kappa\, {\rm Log}(t) +\varphi^{6V}[u;\xi]  
\end{equation}
where we have defined
$$\varphi^{6V}[u;\xi] := \psi^{6V}[u;\xi]+{\rm Log}(\bar a) =-{\rm Log}\left(
\frac{\sin (\al\,\xi)\sin(\al(u-\eta))\sin(u-\xi-\eta)}{\al \sin(\xi)\sin(u-\eta)\sin(\al(u-\xi-\eta))} \right) .$$
Note that we are using the variable $t$ as integration variable (to extract the coefficient $H_{n,k}[u]$ of $t^{k-1}$), and that $\xi$ is an implicit function of $t$ 
upon inverting the equation $t_{6V}[\xi]=t$. The integral is dominated by the solution of the saddle-point equation $\partial_t S_0(\kappa,t)=0$ or equivalently $\partial_\xi S_0(\kappa,t_{6V}[\xi])=0$ 
resulting in
\begin{eqnarray}\label{kappa6v} \kappa&=&\kappa_{6V}[\xi]:= -\frac{t_{6V}[\xi]}{\partial_\xi t_{6V}[\xi]} \partial_\xi \varphi^{6V}[u;\xi] \nonumber \\
&&\!\!\!\!\!\!\!\!\!\!\!\!=
\left\{ \cot(u-\xi-\eta)+\cot(\xi)-\al\cot(\al \xi)-\al\cot(\al(u-\xi-\eta)) \right\}\frac{\sin(u-\xi+\eta)\sin(u-\xi-\eta)}{\sin(2\eta)} \nonumber \\
\end{eqnarray}
with $\al$ as in \eqref{fsol}.

\subsection{Paths}
\subsubsection{Partition function}
The second ingredient of the Tangent Method is the partition function $Y_{k,\ell}$ for a single weighted path in empty space with the same weights as the 6V osculating paths (see an example of such a path in the light blue domain of Fig.~\ref{fig:alltgt} top left).
Note that the path starts where the topmost one in $H_{n,k}[u]$ stopped, namely with a preliminary horizontal step, and end with a vertical step at position $\ell$ (measured from the original position in $Z_n$) on the S boundary.
Note also that all empty vertices receive the weight $a$ of \eqref{6vweights}. We may therefore factor out an unimportant overall weight $a^{n\ell}$, and 
weight the path by the product of its relative vertex weights:
$$a_0=1,\qquad b_0=\frac{b}{a}=\frac{\sin(u-\eta)}{\sin(u+\eta)},\qquad c_0=\frac{c}{a}=\frac{\sin(2\eta)}{\sin(u+\eta)} \ ,$$
for respectively a,b and c type vertices.

Let us use a step-to-step transfer matrix formulation of the path, namely a matrix $T$
describing the transfer from a step to the next. Each step may be in either of two states: horizontal or vertical, and the matrix entry is the corresponding 6V weight at the vertex shared by the step and its successor, which we multiply by an extra weight $z, w$ if the next step is horizontal, vertical respectively. This gives the matrix
$$T_{6V}=\begin{pmatrix} b_0\, z & c_0\, z\\ c_0\, w & b_0\, w \end{pmatrix}$$
allowing to express the generating function $P(z,w):=\sum_{k,\ell\geq 0} Y_{k,\ell}\, z^\ell w^{k+1}$ as
$$P(z,w)=\begin{pmatrix} 0 & 1\end{pmatrix} \cdot ({\mathbb I}- T_{6V})^{-1} \begin{pmatrix} 1\\ 0\end{pmatrix}= \frac{c_0 w}{1-b_0 z- b_0 w(1+\frac{c_0^2-b_0^2}{b_0} z )}$$
Using the new weights
$$ \gamma_1:= b_0= \frac{\sin(u-\eta)}{\sin(u+\eta)}, \quad \gamma_2:=\frac{c_0^2-b_0^2}{b_0}=\frac{\sin(3\eta-u)}{\sin(u-\eta)} $$
We deduce that
\begin{equation}\label{path6v} Y_{k,\ell} = c_0 \gamma_1^k \frac{(1+\gamma_2\, z)^k}{(1-\gamma_1\, z)^{k+1}} \Bigg\vert_{z^\ell}=c_0
\sum_{{P_1\geq 0\atop 0\leq P_2\leq k}\atop P_1+P_2= \ell} {P_1+k\choose k} {k\choose P_2} \, \gamma_1^{k+P_1}\, \gamma_2^{P_2} \end{equation}

\subsubsection{Asymptotics}

We now consider the scaling limit of large $n=N$ and $\kappa=k/N, \lambda=\ell/N$ fixed. Replacing the summation in \eqref{path6v} with an integral over $p_2=P_2/N$
and using the Stirling formula, we find the leading large $N$ behavior of $Y_{k,\ell}$:
\begin{eqnarray*}
Y_{\kappa N,\lambda N} &\simeq& \int_0^\kappa dp_2 e^{-N S_1(\kappa,p_2)}\\
S_1(\kappa,p_2)&=& p_2 {\rm Log}(\frac{p_2}{\gamma_2})+(\kappa-p_2){\rm Log}(\kappa-p_2)+(\lambda-p_2){\rm Log}(\lambda-p_2)\\
&&\qquad -(\kappa+\lambda-p_2){\rm Log}(\gamma_1(\kappa+\lambda-p_2))
\end{eqnarray*}

\subsection{Arctic curves}

We now apply the Tangent Method. We must solve for the saddle-point equations for the total action $S(\kappa,\xi,p_2)=S_0(\kappa,t[\xi])+S_1(\kappa,p_2)$, namely
$\partial_\kappa S=0$ and $\partial_{p_2} S=0$, while the last equation $\partial_\xi S=0$ eventually allows us to solve for $\kappa=\kappa[\xi]$ by using 
the result \eqref{kappa6v}.
We get:
$$\frac{ t (\kappa-p_2)}{\gamma_1(\kappa+\lambda-p_2)} =1 ,\qquad \frac{\gamma_1}{\gamma_2} \frac{p_2(\kappa+\lambda-p_2)}{(\kappa-p_2)(\lambda-p_2)} =1 $$
with the following unique solution:
$$ \frac{p_2}{\kappa}= \frac{\sin(u-3\eta)\sin(\xi)}{\sin(u-\xi-\eta)\sin(2\eta)} , \qquad \frac{\kappa}{\lambda}= \frac{\sin(u-\xi+\eta)\sin(u-\xi-\eta)}{\sin(\xi)\sin(\xi-2\eta)} $$
parameterized by $\xi$ via $t=t_{6V}[\xi]$ \eqref{tparam} and $\kappa=\kappa_{6V}[\xi]$ \eqref{kappa6v}. In particular this determines $\kappa$ as a function of $\lambda$ in the
parametric form $(\kappa,\lambda)=(\kappa_{6V}[\xi],\lambda_{6V}[\xi])$, where
$$\lambda_{6V}[\xi]:= \kappa_{6V}[\xi] \ \frac{\sin(\xi)\sin(\xi-2\eta)}{\sin(u-\xi+\eta)\sin(u-\xi-\eta)} $$
This allows us to identify the slope $A[\xi]=\frac{\kappa_{6V}[\xi]}{\lambda_{6V}[\xi]}$ and the intercept 
$B[\xi]=\kappa_{6V}[\xi]$
for the family of tangents: $F_\xi[x,y]= y+ A[\xi]\, x  -B[\xi]=0$. The parameter $\xi$ is constrained by the condition that $A[\xi]>0$ which implies that $\xi\in [u+\eta-\pi,0]$.
Using the expression for the envelope \eqref{acurve}, we arrive at the final result.

\begin{thm}\label{6VNEthm}
The NE portion of the arctic curve for the 6V-DWBC model in the Disordered regime is predicted by the Tangent Method to be given parametrically by:
$$ x=X_{NE}^{6V}[\xi]:=\frac{B'[\xi]}{A'[\xi]},\qquad y=Y_{NE}^{6V}[\xi]:=B[\xi]-\frac{A[\xi]}{A'[\xi]}B'[\xi] ,\qquad (\xi\in [u+\eta-\pi,0])$$
where
\begin{eqnarray*}
A[\xi]&=&\frac{\sin(u-\xi+\eta)\sin(u-\xi-\eta)}{\sin(\xi)\sin(\xi-2\eta)}\\
B[\xi]&=&\left\{ \cot(u-\xi-\eta)+\cot(\xi)-\al\cot(\al \xi)-\al\cot(\al(u-\xi-\eta)) \right\}\\
&&\qquad \qquad \times \, \frac{\sin(u-\xi+\eta)\sin(u-\xi-\eta)}{\sin(2\eta)}
\end{eqnarray*}
with $\al$ as in \eqref{fsol}.
\end{thm}

As explained in Sect.~\ref{obsec}, we easily get the SE portion of the arctic curve, by applying the transformation $*$: $u\mapsto u^*=\pi- u$, and the change of coordinates \eqref{flip6v6vp} for $\mu=1$. As a result we have the following.

\begin{thm}\label{6VSEthm}
The SE portion of the arctic curve for the 6V-DWBC model in the Disordered regime is predicted by the Tangent Method to be given parametrically by:
$$ x=X_{SE}^{6V}[\xi]:=X_{NE}^{6V}[\xi]^*,\qquad y=Y_{SE}^{6V}[\xi]:=1-Y_{NE}^{6V}[\xi]^* ,\qquad (\xi\in [\eta-u,0])$$
with $X_{NE}^{6V},Y_{NE}^{6V}$ as in Theorem \ref{6VNEthm}.
\end{thm}

Finally, we note that the weights \eqref{6vweights} and the DWBC are invariant under central symmetry, which reflects all arrow orientations. As a consequence,
the arctic curve of the 6V-DWBC model is centro-symmetric as well, and it can be easily completed by applying the central symmetry
$(x,y)\mapsto (-1-x,1-y)$ to the NE and SE branches to respectively produce the SW and NW ones.

\begin{remark}\label{rem6v}
At the self-dual point $u=u^*=\frac{\pi}{2}$, the arctic curve is symmetric w.r.t. the horizontal line $y=1/2$, as well as the vertical line $x=-1/2$ by the central symmetry. The full curve is then obtained by successive reflections of the NE branch; as an example the limit shape of ASMs \cite{CP2010} is made of 4 reflected portions of ellipse.
This is no longer true if $u\neq \frac{\pi}{2}$.
\end{remark}

\section{6V' model}\label{6vpsec}

\subsection{Partition function and one-point function}

\subsubsection{Inhomogeneous partition function}

The partition function of the inhomogeneous U-turn boundary 6V model was derived by Kuperberg and independently by Tsuchya \cite{kuperberg2002symmetry,tsu}.
Let 
$$m_U(u,v):=\frac{1}{\sin(u-v+\eta)\sin(u-v-\eta)}-\frac{1}{\sin(u+v+\eta)\sin(u+v-\eta)}$$
Note that as opposed to the 6V case, this is no longer a function of $u-v$ only, but includes a reflected term which is a function of $u+v$.
\begin{thm}\label{kupthm}
The U-turn boundary 6V partition function reads:
\begin{eqnarray}\label{Udelta}
&&Z_{n}^{6V-U}[\bu,\bv;\theta]=(\rho_e\rho_o)^{n^2}\, \det_{1\leq i,j \leq n}\big(m_U(u_i,v_j)\big) \nonumber \\
&&\qquad
\times{
\scriptstyle 
\frac{\left\{\prod_{i=1}^n \sin(\theta-v_i)\sin(2u_i+2\eta)\sin(2\eta) \right\}\left\{\prod_{i,j=1}^n 
\sin(u_i-v_j+\eta)\sin(u_i-v_j-\eta)\sin(u_i+v_j+\eta)\sin(u_i+v_j-\eta)\right\}}{\left\{\prod_{1\leq i<j \leq n}\sin(u_i-u_j)\,
\sin(v_j-v_i)\right\}\left\{\prod_{1\leq i\leq j\leq n}\sin(u_i+u_j)\,\sin(v_i+v_j)\right\}}}
\end{eqnarray}
\end{thm}

As mentioned above and illustrated in Fig.~\ref{fig:6vUtop}, the 6V' model corresponds to the choice of parameter 
$\theta=-u-\eta$, which ensures that $y_u(u)=0$
for all U-turns.
The partition function corresponding to this choice, where we cut out the U-turns of Fig.~\ref{fig:6vUtop} (a) and remove their weights, 
as well as the weights of the trivially fixed b-type vertices of the bottom row in Fig.~\ref{fig:6vUtop} (b), reads:
\begin{equation}\label{UUp}
Z_n^{6V'}[u,\bv]=\lim_{u_i\to u\atop \theta\to -u-\eta} \frac{(-1)^n\, Z_n^{6V-U}[\bu,\bv;\theta]}{\sin^n(2u+2\eta)\,\rho_e^n\, \prod_{i=1}^n \sin(-u-v_i-\eta) } 
\end{equation}
where we have identified the limit of the U-turn weights to be $y_d(u_i)\to -\sin(2u+2\eta)$ and that of the $b$ weights of the bottom (even) row to be $b_e\to \rho_e\sin(-u-v_i-\eta)$.

\begin{remark}\label{thetarem}
Note that in \eqref{Udelta} the dependence on the parameter $\theta$ is only through the prefactor $\prod_i \sin(\theta-v_i)$. The ``worst case scenario" is the homogeneous limit where all $v_i\to v$, and where this gives a factor $\sin^n(\theta-v)$. In any case, this does not affect the value of the
thermodynamic free energy $f=\lim_{n\to\infty} -\frac{1}{n^2}{\rm Log}(Z_n^{6V'})$, which is independent of $\theta$. We may therefore safely fix the value of $\theta$ to suit our needs. 
\end{remark}

%
%
%
%
%
%
%

\subsubsection{Homogeneous limit}

Like in the 6V case, the homogeneous limit where we take all $u_i\to u$ and all $v_i\to v$ involves the quantity:

$$\Delta_n[u,v]:= \lim_{u\to u_i\atop v\to v_i}  \frac{\det_{1\leq i,j \leq n}\left( m_U(u_i,v_j)\right)}{\prod_{1\leq i<j\leq n} (u_i-u_j)(v_j-v_i)} $$
Upon Taylor-expanding rows and columns, we may rewrite:
$$\Delta_n[u,v]=(-1)^{n(n-1)/2}\, \det_{0\leq i,j\leq n-1}\left( \frac{\partial_u^i\partial_v^j m_U(u,v)}{i!j!}  \right)=:\frac{1}{\prod_{i=0}^{n-1}(i!)^2} D_n[u,v]$$
where the determinant $D_n[u,v]$ reads
\begin{equation}\label{dofuv}
D_n[u,v]=\det_{0\leq i,j\leq n-1}\left((-1)^j\partial_u^i\partial_v^j m_U(u,v) \right)\end{equation}

Using the relation \eqref{UUp} and the result of Theorem \ref{kupthm}, we obtain the homogeneous partition function of the 6V' model:
\begin{equation}\label{sixvpf}
\frac{Z_n^{6V'}[u,v]}{\rho_e^{n^2-n}\rho_o^{n^2}\sin^n(2\eta)}=\Delta_n[u,v]\,
\frac{\big( \sin(u-v+\eta)\,\sin(u-v-\eta)\,\sin(u+v+\eta)\,\sin(u+v-\eta)\big)^{n^2}}{\big(\sin(2u)\,\sin(2v)\big)^{n(n+1)/2}}
\end{equation}

To determine $\Delta_n[u,v]$ one uses like in the 6V case the Pl\"ucker/Desnanot-Jacobi relation of Lemma \ref{desnajac} applied to the $n+1\times n+1$ matrix $M$
in the definition of $D_{n+1}[u,v]$ \eqref{dofuv}:
$$D_{n+1}[u,v]\,D_{n-1}[u,v]=\partial_u D_n[u,v]\, \partial_vD_n[u,v]-D_n[u,v]\, \partial_u\partial_v D_n[u,v]$$
which implies
\begin{equation}\label{pluone}
\frac{D_{n+1}[u,v]\,D_{n-1}[u,v]}{(D_n[u,v])^2}+ \partial_u\partial_v {\rm Log}\big(D_n[u,v]\big) =0
\end{equation}
As a direct consequence, we have:
\begin{thm}\label{6vpdeltathm}
The quantity $\Delta_n[u,v]$ obeys the following recursion relation:
\begin{equation}\label{deltarecur}
\frac{\Delta_{n+1}[u,v]\,\Delta_{n-1}[u,v]}{\Delta_{n}[u,v]^2}+\frac{1}{n^2} \partial_u\partial_v {\rm Log}\big(\Delta_{n}[u,v]\big) =0
\end{equation}
\end{thm}
Note that the latter can be used to determine $\Delta_n[u,v]$ recursively, starting with $\Delta_0[u,v]=1$ and $\Delta_1[u,v]=m_U(u,v)$, as we illustrate now with a few simple examples.

\begin{example}\label{classicex}
Let us consider the ``classical limit" $\eta\to 0$, where:
$$ m_U(u,v)=\frac{1}{\sin^2(u-v)}-\frac{1}{\sin^2(u+v)} =\frac{\sin(2u)\sin(2v)}{\sin^2(u-v)\,\sin^2(u+v)}$$
We have:
\begin{thm}\label{classithm} In the classical case $\eta=0$, we have for all $n\geq 1$:
$$
\Delta_{n}[u,v]=n! \left(\frac{\sin(2u)\sin(2v)}{\sin^2(u-v)\,\sin^2(u+v)} \right)^{n(n+1)/2}
$$
\end{thm}
\begin{proof}
The proof is by induction on $n$, using \eqref{deltarecur}, and follows from the relation
$$\frac{\sin(2u)\sin(2v)}{\sin^2(u-v)\,\sin^2(u+v)}-\partial_u\partial_v {\rm Log}\left(\sin(u-v)\,\sin(u+v)\right) =0 \ .$$
\end{proof}

\noindent Note that the corresponding 6V' partition function vanishes, however we get a finite limit for the quantity
$$\lim_{\eta\to 0} \frac{Z_n^{6V'}[u,v]}{\rho_e^{n^2-n}\rho_o^{n^2}\sin^n(2\eta)}=n!\left( \sin(u-v)\sin(u+v)\right)^{n(n-1)} $$
This result has a simple interpretation: sending $\eta\to 0$ implies both $c_e$ and $c_o$ type vertices have
vanishing weights. However, without the ability to turn right, none of the osculating paths can satisfy the boundary conditions, unless each path is allowed at least one right turn. In this formulation, we must no longer see the paths as osculating, but rather as {\it crossing} at fully occupied $a$-type vertices.  The minimal case is if each path has exactly one turn
(and the vanishing weight $\sin^n(2\eta)$ is divided before taking the $\eta\to 0$ limit). For each $i=1,2,...,n$, the $i$-th path from the bottom starts with say $j=\sigma(i)$ horizontal steps, then turns right and ends with $2i-1$ vertical steps at the $j$-th endpoint. Clearly there are as many such configurations as permutations $\sigma$ of the $n$ path ends, which accounts for an overall factor of $n!$. Collecting all the Boltzmann weights gives the remaining factor.
\end{example}

\begin{example}\label{picex}
We now consider the ``free fermion" case $\eta=\frac{\pi}{4}$, where
\begin{eqnarray*}
m_U(u,v)&=&\frac{1}{\sin(u-v+\frac{\pi}{4})\sin(u-v-\frac{\pi}{4})} -\frac{1}{\sin(u+v+\frac{\pi}{4})\sin(u+v-\frac{\pi}{4})}\\
&=&\frac{4\sin(2u)\sin(2v)}{\cos(2(u-v))\cos(2(u+v))}
\end{eqnarray*}
\begin{thm}\label{picthm} In the free fermion case $\eta=\frac{\pi}{4}$, we have for all $n\geq 1$:
$$
\Delta_{n}[u,v]=\frac{\left(4\sin(2u)\sin(2v)\right)^{n(n+1)/2}\left(4\,\cos(2u)\cos(2v)\right)^{n(n-1)/2}}{\left(\cos\big(2(u-v)\big)\,\cos\big(2(u+v)\big)\right)^{n^2}} 
$$
\end{thm}
\begin{proof}
The proof is by induction on $n$ using \eqref{deltarecur}, and follows from the relation
$$ \frac{4\sin(4u)\sin(4v)}{\cos^2(2(u-v))\cos^2(2(u+v))}+\partial_u\partial_v{\rm Log}\left(\cos(2(u-v))\cos(2(u+v))\right)=0 $$
\end{proof}
\noindent The corresponding 6V' partition function reads:
\begin{equation}\label{freepicex} 
\frac{Z_n^{6V'}[u,v]}{\rho_e^{n^2-n}\rho_o^{n^2}}=\left(\cos(2u)\cos(2v)\right)^{n(n-1)/2}
\end{equation}
\end{example}

\subsubsection{One-point function}

As in the case of the 6V model, we consider the semi-homogeneous partition function $Z_n^{6V'}[u,v;\xi]$ with the same boundary conditions as $Z_n^{6V'}[u,v]$
but with a different vertical spectral parameter in the last column, set to $v_n=v+\xi$.
It is again obtained as a limit of \eqref{UUp} and reads:
\begin{eqnarray}\label{pluone}
&&\!\!\!\!\!\!\!\!\!
\frac{Z_n^{6V'}[u,v;\xi]}{\rho_e^{n^2-n}\rho_o^{n^2}\sin^n(2\eta)}\!\!=\!\!\Delta_n[u,v;\xi] \frac{\big( \sin(u-v+\eta)\sin(u-v-\eta)\sin(u+v+\eta)\sin(u+v-\eta)\big)^{n(n-1)}}{\sin(2\xi+2v)\big(\sin(\xi+2v)\big)^{n-1}\big(\sin(2u)\big)^{\frac{n(n+1)}{2}} \big(\sin(2v)\big)^{\frac{n(n-1)}{2}}}\nonumber\\
&&\quad\qquad \qquad \times \, \left(\sin(u-v-\xi+\eta)\sin(u-v-\xi-\eta)\sin(u+v+\xi+\eta)\sin(u+v+\xi-\eta)\right)^{n}\nonumber\\
\end{eqnarray}
in terms of  the semi-homogeneous quantity
$$\Delta_n[u,v;\xi]:=\lim_{u_1,...u_n\to u\atop v_1,...,v_{n-1}\to v, v_n\to v+\xi} 
\frac{\det_{1\leq i,j\leq n} \left( m_U(u_i,v_j)\right)}{\prod_{1\leq i<j\leq n}\sin(u_i-u_j)\,\sin(v_j-v_i)}$$
Repeating the Taylor expansion of rows and columns except the last one, we may rewrite:
\begin{eqnarray}\Delta_n[u,v;\xi]&=&
\frac{(-1)^{n(n-1)/2}}{\sin^{n-1}(\xi)}\, \det\left\{\left(\frac{\partial_u^i\partial_v^j m_U(u,v)}{i!j!} \right)_{i=1,..,n-1\atop j=0,...,n-2} \Big\vert\left(\frac{\partial_u^{i}m_U(u,w)}{i!}\right)_{i=0,...,n-1}\right\}\nonumber \\
&=:&\frac{(-1)^{n-1}\, (n-1)!}{\sin^{n-1}(\xi)\,
\prod_{i=0}^{n-1}(i!)^2} D_n[u,v;\xi] \label{deltad6vp}
\end{eqnarray}
where the determinant $D_n[u,v;\xi]$ reads
\begin{equation}\label{deteronept6vp}
D_n[u,v;\xi]=\det\left\{\left\{(-1)^j\partial_u^i\partial_v^j m_U(u,v) \right\}_{i=1,..,n-1\atop j=0,...,n-2}\Big\vert \left\{\partial_u^{i}m_U(u,v+\xi)\right\}_{i=0,...,n-1}\right\}
\end{equation}

As before, we define the one-point function $H_n^{6V'}[u,v;\xi]$ as the ratio:
\begin{eqnarray}
&&H_n^{6V'}[u,v;\xi]:=\frac{Z_n^{6V'}[u,v;\xi]}{Z_n^{6V'}[u,v]}= \frac{\Delta_n[u,v;\xi]}{\Delta_n[u,v]}\, \frac{\sin(\xi+2v)}{\sin(2\xi+2v)}\,\left(\frac{\sin(2v)}{\sin(\xi+2v)}\right)^n
\nonumber \\
&&\quad \qquad \times\left( \frac{\sin(u-v-\xi+\eta)\,\sin(u-v-\xi-\eta)\,\sin(u+v+\xi+\eta)\,\sin(u+v+\xi-\eta)}{\sin(u-v+\eta)\,\sin(u-v-\eta)\,\sin(u+v+\eta)\,\sin(u+v-\eta)}\right)^n
\label{ratioref}
\end{eqnarray}

The Pl\"ucker/Desnanot-Jacobi relation of Lemma \ref{desnajac} applied to the $n+1\times n+1$ matrix $M$
in the definition of $D_{n+1}[u,v;\xi]$ \eqref{deteronept6vp} implies the following:
$$D_{n+1}[u,v;\xi]\, D_{n-1}[u,v] = D_n[u,v;\xi]\, \partial_u D_n[u,v] -D_n[u,v]\, \partial_u D_n[u,v;\xi]$$
Introducing the  reduced one-point function 
\begin{equation}\label{redone6vp}
H_n[u,v;\xi]:= (-1)^{n-1} (n-1)! \frac{D_n[u,v;\xi]}{D_n[u,v]} =\sin^{n-1}(\xi)  \frac{\Delta_n[u,v;\xi]}{\Delta_n[u,v]} ,
\end{equation}
we may recast the above into the following.
\begin{thm}\label{6vponeptthm}
The reduced one-point function of the 6V' model obeys the following relation:
\begin{equation}\label{hone}
\frac{H_{n+1}[u,v;\xi]}{H_n[u,v;\xi]} \frac{\Delta_{n-1}[u,v]\,\Delta_{n+1}[u,v]}{\Delta_n[u,v]^2}+\frac{1}{n} \partial_u {\rm Log}(H_n[u,v;\xi] ) =0
\end{equation}
\end{thm} 
Together with \eqref{deltarecur}, this determines $H_n[u,v;\xi]$ recursively, using the initial data $H_1[u,v;\xi]=\frac{m_U(u,v+\xi)}{m_U(u,v)}$, and in turn
the one-point function $H_n^{6V'}[u,v;\xi]$ via:
\begin{eqnarray}
&&H_n^{6V'}[u,v;\xi]= H_n[u,v;\xi]\,\frac{\sin(\xi)\sin(\xi+2v)}{\sin(2\xi+2v)}\,\left(\frac{\sin(2v)}{\sin(\xi)\sin(\xi+2v)}\right)^n
\nonumber \\
&&\quad \qquad \times\left( \frac{\sin(u-v-\xi+\eta)\,\sin(u-v-\xi-\eta)\,\sin(u+v+\xi+\eta)\,\sin(u+v+\xi-\eta)}{\sin(u-v+\eta)\,\sin(u-v-\eta)\,\sin(u+v+\eta)\,\sin(u+v-\eta)}\right)^n
\label{oneptfin6vp}
\end{eqnarray}

\subsection{Large $n$ limit: free energy and one-point function asymptotics}

\subsubsection{Free energy}

For large $n=N$, like in the 6V case, the relation \eqref{deltarecur} leads to the following leading behavior for the function $\Delta_n[u,v]$:
\begin{equation}\label{asymptoD}
\Delta_N[u,v]\simeq e^{-N^2 f[u,v]}
\end{equation}
for some function $f[u,v]$ to be determined (see \cite{RIBKOR} for a full derivation). 


\noindent{\bf Liouville equation and free energy.}

For large $n=N$, substituting the behavior \eqref{asymptoD} into eq.\eqref{deltarecur}, and
expanding at leading order in $N^{-1}$, we get the following 2D Liouville partial differential equation equation for the function $f[u,v]$:
\begin{equation}\label{liouville}
\partial_u\partial_v f[u,v]-e^{-2 f[u,v] } =0
\end{equation}
Introducing the function 
$W[u,v]:=e^{f[u,v]}$ this may be rewritten as:
$$W\,\partial_u\partial_vW-\partial_uW\,\partial_vW =1$$

The general solution $W$ of this equation is known to be \cite{Liouville,Crowdy}:
\begin{equation}\label{gensol}
W[u,v]=\frac{g(u)-h(v)}{|g'(u)h'(v)|^{\frac{1}{2}}} \end{equation}
for some arbitrary differentiable functions $g,h$.  In \cite{RIBKOR}, the functions $g,h$ are fixed by use of symmetries and known limits of $W$, leading to the following.

\begin{thm}[\cite{RIBKOR}]\label{freeconj}
The leading asymptotics of the determinant $\Delta_n[u,v]$ is given by  $W[u,v]=\lim_{N\to\infty} \Delta_N[u,v]^{-\frac{1}{N^2}}$ where:
\begin{equation}\label{conjW}
W[u,v]=\frac{\sin(\al(u-v-\eta))\,\sin(\al(-u-v-\eta))}{\al\, |\sin(2\al u)\, \sin(2\al (v+\eta))|^\frac{1}{2}}
\end{equation}with
\begin{equation}\label{vala}\al=\frac{\pi}{\pi-2\eta}
\end{equation}
\end{thm}


Theorem \ref{freeconj} gives access to the full free energy $f^{6V'}$ of the 6V' model, as defined by the large $N$ asymptotics 
$Z_N^{6V'}[u,v]\simeq e^{-N^2\, f^{6V'}[u,v]}$, where as a consequence of \eqref{sixvpf}, we have:
\begin{equation}\label{f6vprime}
 f^{6V'}[u,v]=f[u,v]+{\rm Log}\left( \frac{\sqrt{|\sin(2u)\,\sin(2v)|}}{\rho_e\rho_o\sin(u-v+\eta)\,\sin(u-v-\eta)\,\sin(u+v+\eta)\,\sin(u+v-\eta)} \right)
 \end{equation}
 This leads immediately to the following.

\begin{cor}[\cite{RIBKOR}]
The free energy of the 6V' model in the Disordered regime reads:
\begin{eqnarray}\label{fren6vp}
f^{6V'}[u,v]&=&{\scriptstyle \frac{1}{2} }\,{\rm Log} \left\vert \frac{\sin(2u)\,\sin(2v)}{\sin(2\al u)\, \sin(2\al (v+\eta))}\right\vert \nonumber \\
&&\quad + {\rm Log}\left( \frac{\sin(\al(u-v-\eta))\,\sin(\al(-u-v-\eta))}{\al\, \rho_e\rho_o\sin(u-v+\eta)\,\sin(u-v-\eta)\,\sin(u+v+\eta)\,\sin(u+v-\eta)} \right)
\end{eqnarray}
\end{cor}

We also have access to the free energy $f^{20V}$ of the 20V DWBC3 model defined in Sect.~\ref{20vpresec}, which will be studied in Section \ref{20vsec} below. 
The free energy is defined via $Z_N^{20V}[u,v]\simeq e^{-N^2\, f^{20V}[u,v]}$
for large $N$. As a consequence of \eqref{vingtvpf} which relates the partition functions of the 20V-DWBC3 and 6V' models (see also Ref.~\cite{DF20V}), we have the relation:
\begin{equation}\label{f20v}
f^{20V}[u,v]=f^{6V'}[u,v]+\frac{1}{2}{\rm Log} \left(\nu^3\, \sin^3(2u+2\eta)\,\sin(u-v-\eta)\,\sin(u+v-\eta)\right)
\end{equation}

Let us apply this to the uniform case \eqref{combipoint20v}, where the partition function $Z_n^{6V'}$ of the $6V'$ model on the $(2n-1)\times n$ grid is related to the number of configurations $Z_n^{20V}$ of the 20V model with DWBC3 on the quadrangle $\cQ_n$ \cite{DF20V} (see Sect.~\ref{20vpresec}). Using Theorem \ref{freeconj} and the relations \eqref{fren6vp} and \eqref{f20v}, and approaching the desired value  $v=-4\eta+\epsilon$, while $u=\eta$, we get for $\eta=\frac{\pi}{8}$, $\al=\frac{4}{3}$, $\nu=\sqrt{2}$:
\begin{eqnarray*}
e^{f^{20V}}&=&\lim_{\epsilon\to 0}  \left\vert\frac{\sin(2\eta)\,\sin(-8\eta+2\epsilon)}{\sin(\frac{8}{3}\eta))\, \sin(-8\eta+\frac{8}{3}\epsilon)} \right\vert^{\frac{1}{2}} \frac{3\, \sin(\frac{16}{3}\eta)\sin(\frac{8}{3}\eta)}{4\, \nu^{3/2}\,\sin^2(4\eta)\,\sin(6\eta)\,\sin(2\eta)}=\frac{3^{9/4}}{2^{9/2}} \ .
\end{eqnarray*}
This is in agreement with the asymptotics of the exact conjectured formula of Ref.~\cite{DF20V} for the uniformly weighted partition function, namely:
\begin{equation}\label{20vpf}
Z_N^{20V}=2^{N(N-1)/2}\prod_{i=0}^{N-1}\frac{(4i+2)!}{(n+2i+1)!}\simeq \left(\frac{2^{9/2}}{3^{9/4}}\right)^{N^2}\ ,
\end{equation}
easily derived by use of the Stirling formula.

\subsubsection{One-point function}

We now derive the large $n=N$ asymptotics of the one-point function $H_n^{6V'}[u,v;\xi]$ \eqref{ratioref}. From Eq. \eqref{oneptfin6vp},
the latter is simply expressed in terms of the 
reduced one-point function $H_n[u,v;\xi]$ \eqref{redone6vp}. Like in the 6V case, we first derive a differential equation governing the asymptotic behavior
of $H_n[u,v;\xi]$, and compute a number of limits to fix integration constants. It turns out that our Conjecture \ref{freeconj} is sufficient
to determine asymptotics completely.

By Theorem \ref{6vponeptthm}, $H_n[u,v;\xi]$ must satisfy \eqref{hone}, which implies
the leading asymptotic behavior
\begin{equation}\label{asymptoh} H_N[u,v;\xi]\simeq_{N\to\infty}   e^{-N \psi[u,v;\xi]}
\end{equation}
for some function $\psi[u,v;\xi]$. As a simple confirmation, using the definition \eqref{redone6vp} and the fact that $\Delta_n[u,v;0]=\Delta_n[u,v]$, we find that
$H_n[u,v;\xi]\simeq_{\xi\to 0} \xi^{n-1}$, resulting in:
\begin{equation}\label{firsxtlimxi}
\psi[u,v;\xi]\simeq_{\xi \to 0}  -{\rm Log}(\xi)\end{equation}

\noindent{\bf Differential equation.}

Substituting the expressions \eqref{asymptoD} and \eqref{asymptoh} into eq. \eqref{hone} for $n=N$, 
and expanding to leading order in $N^{-1}$, we get the following partial differential equation:
\begin{equation}
\partial_u\psi[u,v;\xi]-e^{-2f[u,v]-\psi[u,v;\xi]}=0 \label{psione}
\end{equation}

\noindent{\bf Limits.}

In addition to the limit \eqref{firsxtlimxi} above, let us consider the limit $u-v-\xi-\eta\to 0$,
by setting $\xi=u-v-\eta-\epsilon$ and sending $\epsilon\to 0$. The entries of the last column of the determinant $D_N[u,v;v+\xi]$ \eqref{deteronept6vp} read:
$$ \partial_u^i m_U(u,u-\eta-\epsilon) = \frac{1}{\sin(2\eta)}\frac{(-1)^i i!}{\epsilon^{i+1} }+O(\epsilon^{-i})$$
The dominant term is in the last row and results in 
$$D_N[u,v;u-\eta-\epsilon] \simeq \frac{(-1)^{N-1} (N-1)!}{\sin(2\eta)\,\epsilon^{N} } D_{N-1}[u,v] $$
We deduce that
\begin{eqnarray*}
H_N[u,v;u-v-\eta-\epsilon]&=&(-1)^{N-1}\, (N-1)!\frac{D_N[u,v;u-\eta-\epsilon]}{D_N[u,v]}\\
&&\!\!\!\!\!\!\!\!\!\! \simeq_{\epsilon\to 0} \frac{ (N-1)!^2}{\sin(2\eta)\,\epsilon^{N} }\,\frac{D_{N-1}[u,v]}{D_{N}[u,v]}
\simeq \frac{1}{\epsilon^{N} }\,\frac{\Delta_{N-1}[u,v]}{\Delta_{N}[u,v]}\simeq \frac{W[u,v]^{2N}}{\epsilon^{N} }
\end{eqnarray*}
where we have used the large $N$ asymptotics $\Delta_N[u,v]\simeq W[u,v]^{-N^2}$. Matching this with the asymptotics \eqref{asymptoh}, we conclude that 
\begin{equation}\label{seclimxipsi}
\psi[u,v,u-v-\eta+\epsilon]_{\epsilon\to 0}\simeq {\rm Log}\left\vert\frac{\epsilon}{W[u,v]^2}\right\vert 
\end{equation}
Repeating the analysis for $\xi=\eta-u-v+\epsilon$, we find analogously:
\begin{equation}\label{thirdlimxipsi}\psi[u,v,\eta-u-v+\epsilon]_{\epsilon\to 0}\simeq {\rm Log}\left\vert\frac{\epsilon}{W[u,v]^2}\right\vert 
\end{equation}

\noindent{\bf Solution.}

Note that Eq.~\ref{psione} may be rewritten  in the form $\partial_u(e^{\psi})=e^{-2f}=W^{-2}$. This can be integrated w.r.t. the variable $u$ as follows:
\begin{equation}\label{integ}
e^{\psi[u,v,\xi]}= c[v,\xi] -\frac{\al \sin(2\al(v+\eta))}{\sin(\al(u-v-\eta))\,\sin(\al(u+v+\eta))} 
\end{equation}
for some integration constant $c[v,\xi]$ independent of $u$.

We now use the limit \eqref{seclimxipsi} to express that, for $\xi=u-v-\eta+\epsilon$ and $\epsilon\to 0$, we have $e^{\psi[u,v;\xi]}\to 0$. This gives:
$$ c[v,u-v-\eta] =\frac{\al \sin(2\al (v+\eta))}{\sin(\al(u-v-\eta))\,\sin(\al(u+v+\eta))}$$
which is valid for all $u,v$. In particular, setting $u=v+\eta+\xi$ yields the integration constant
$$c[v,\xi] =\frac{\al \sin(2\al(v+\eta))}{\sin(\al \xi)\,\sin(\al(\xi+2v+2\eta))} $$
which we plug back into \eqref{integ} to finally get:
\begin{equation}\label{solpsi}
\psi[u,v;\xi]= {\rm Log}\left( \frac{\al \,\sin(2\al (v+\eta))\,\sin(\al(u-v-\xi-\eta))\,\sin(\al(u+v+\xi+\eta))}{\sin(\al(u-v-\eta))\,\sin(\al(u+v+\eta))\,\sin(\al\xi)\,\sin(\al(\xi+2v+2\eta))} \right)
\end{equation}

Using the relation \eqref{oneptfin6vp} this leads to the following result for the one-point function asymptotics.
\begin{thm}\label{asympto6vponept}
The one-point function $H_n^{6V'}[u,v;\xi]$ has the following large $n=N$ behavior:
\begin{eqnarray*}&&H_N^{6V}[u,v;\xi]\simeq e^{-N \psi^{6V'}[u,v;\xi]} \\
&&\psi^{6V'}[u,v;\xi]= -{\rm Log}\left(
\frac{\sin(\al(u-v-\eta))\,\sin(\al(u+v+\eta))\,\sin(\al\xi)\,\sin(\al(\xi+2v+2\eta))}{\al \,\sin(2\al (v+\eta))\,\sin(\al(u-v-\xi-\eta))\,\sin(\al(u+v+\xi+\eta))} \right) \\
&&\ - {\rm Log}\left(\frac{\sin(2v) \sin(u-v-\xi+\eta)\sin(u-v-\xi-\eta)\sin(u+v+\xi-\eta)\sin(u+v+\xi+\eta)}{\sin(\xi)\sin(\xi+2v)\sin(u-v+\eta)\sin(u-v-\eta)\sin(u+v-\eta)\sin(u+v+\eta)}\right) \end{eqnarray*}
with $\al$ as in \eqref{fsol}.
\end{thm}

As a consistency check, we find that $\lim_{\xi\to 0} \psi^{6V'}[u,v;\xi]=0$, in agreement with the fact that $H_n^{6V'}[u,v;0]=1$ by definition.


\begin{remark}\label{thetaonerem}
In the case of the more general U-turn 6V model (with arbitrary value of the parameter $\theta$, we already showed in Remark \ref{thetarem} that the thermodynamic free energy of the model is independent of $\theta$, therefore identical to that of the 6V' model. The same argument may be applied to the one-point function, whose leading asymptotics is independent of $\theta$ as well, and therefore the {\it same} for U-turn 6V and 6V' models.
\end{remark}

\begin{remark}
Independently of Theorem \ref{freeconj}, eq. \eqref{psione} can be solved in terms of the generic function $g$ which determines the general solution \eqref{gensol}
to the Liouville equation with the correct symmetries and limits, namely such that $h(v)=g(v+\eta)$, with the expression:
$$ W[u,v]=\frac{g[u]-g[v+\eta]}{\sqrt{|g'(u)g'(v+\eta)|}} $$
Solving Eq.~\ref{psione} in the same manner as above, we obtain:
$$\psi[u,v;\xi]={\rm Log}\left( \frac{ (g(u)-g(v+\xi+\eta))\, g'(v+\eta)}{(g(u)-g(v+\eta))(g(v+\xi+\eta)-g(v+\eta))}\right) $$
In particular, we recover the solution for the 6V-DWBC case by picking $g(u)=\tan(\al u)$, which leads to
\begin{eqnarray*}W[u,v]&=& \frac{\sin(\al(u-v-\eta))}{\al}=W[u-v]\\ 
\psi[u,v;\xi]&=&{\rm Log}\left( \frac{\al\sin(\al (u - v - \eta -\xi))}{\sin(\al\xi)\,\sin(\al(u-v-\eta))}\right)=\psi[u-v;\xi] \end{eqnarray*}
in agreement with \eqref{W6v} and \eqref{psi6v}.
\end{remark}


\subsection{Paths}

\subsubsection{Partition function}

With the setting of Fig.~\ref{fig:alltgt} (bottom left, light blue domain), we wish to compute the partition function $Y_{k,\ell}$ of a single path of the 6V' model in the first quadrant $\Z_+^2$,
with starting point $(0,k)$ and endpoint $(\ell,0)$. The weights of the path are those of the 6V' model, namely $(b_o,c_o)$ for a path (going straight, turning) at a vertex with second coordinate $y=2j$, $j=0,1,2,...$
and $(b_e,c_e)$ for a path (going straight, turning) at a vertex with second coordinate $y=2j+1$, $j=0,1,2,...$ However the path crosses a domain of empty vertices, each receiving weights $a_e,a_o$ depending on the parity of their second coordinate $y$. Factoring an overall weight $(a_o)^{n\ell}(a_e)^{(n-1)\ell}$ which does not affect our study, the weights of the path steps must be divided by 
$a_e,a_o$ and finally read:
\begin{eqnarray}\label{moreweights}
&&b_0=\frac{b_o}{a_o}=\frac{\sin(u-v-\eta)}{\sin(u-v+\eta)}, \qquad c_0=\frac{c_o}{a_o}= \frac{\sin(2\eta)}{\sin(u-v+\eta)},\nonumber \\
&&b_1=\frac{b_e}{a_e}=\frac{\sin(u+v+\eta)}{\sin(u+v-\eta)},  \qquad c_1=\frac{c_e}{a_e}=\frac{\sin(2\eta)}{\sin(\eta-u-v)} 
\end{eqnarray}
for vertices with $y=2j$ and $y=2j+1$ respectively.
Note that the path has a horizontal step just before entering the first quadrant, and has a final vertical step.

The partition function $Y_{k,\ell}$ is computed by use of a transfer matrix technique. Each path is travelled from N,W to S,E, and the transfer matrix
is a $4\times 4$ matrix $T_{6V'}$ whose entries correspond to the vertex weight for the transition from the entering step at each visited vertex to the outgoing step, with the four possible configurations $(-,o),(\vert, o),(-,e),(\vert,e)$ of horizontal/vertical step ending at an odd/even vertex. Moreover we include an extra weight $z,w$ per horizontal, vertical outgoing step respectively.
The matrix $T_{6V'}$ reads:
$$T_{6V'}=\begin{pmatrix}
b_0 z & c_0 z & 0 & 0\\
0 & 0  & c_1 w & b_1 w\\
0 & 0  & b_1 z & c_1 z \\
c_0 w  &b_0 w & 0 & 0
\end{pmatrix}$$
We deduce the generating function for the $Y_{k,\ell}$:
\begin{eqnarray}
Y^{6V'}(z,w)&=&\sum_{k,\ell\geq 0} Y_{k,\ell}\, w^{k+1} z^{\ell}=(0,0,0,1) ({\mathbb I}-T_{6V'})^{-1} \begin{pmatrix}1\\0\\1\\0\end{pmatrix}  \nonumber \\
&=&\frac{w(c_0(1-b_1 z)+ c_1 w(b_0+(c_0^2-b_0^2)z))}{(1-b_0 z)(1-b_1 z)-w^2(b_0+(c_0^2-b_0^2)z)(b_1+(c_1^2-b_1^2)z)}\nonumber \\
&=&c_0\sum_{j\geq 0} w^{2j+1} \frac{(b_0+(c_0^2-b_0^2)z)^j(b_1+(c_1^2-b_1^2)z)^j}{(1-b_0 z)^{j+1}(1-b_1 z)^j} \nonumber \\
&&\qquad \qquad
 + c_1\sum_{j\geq 0} w^{2j+2} 
\frac{(b_0+(c_0^2-b_0^2)z)^{j+1}(b_1+(c_1^2-b_1^2)z)^j}{(1-b_0 z)^{j+1}(1-b_1 z)^{j+1}}\nonumber \\
&=&\sum_{k\geq 0} w^{k+1} c_\epsilon
\frac{ \left(\gamma_1(1+\gamma_3 z)\right)^{\frac{k+\epsilon}{2}} \left(\gamma_2(1+\gamma_4 z)\right)^{\frac{k-\epsilon}{2}}}{(1-\gamma_1 z)^{1+\frac{k-\epsilon}{2}} (1-\gamma_2 z)^{\frac{k+\epsilon}{2}}}
 \label{genpa6v}
\end{eqnarray}
where we have used the notation $\epsilon:=k$ mod 2 (with $\epsilon\in \{0,1\}$), and the following weights:
\begin{eqnarray}
&&\gamma_1=b_0=\frac{\sin(u-v-\eta)}{\sin(u-v+\eta)},\quad \gamma_2=b_1=\frac{\sin(u+v+\eta)}{\sin(u+v-\eta)}, \nonumber \\
&&\gamma_3=\frac{c_0^2-b_0^2}{b_0}=-\frac{\sin(u-v+3\eta)}{\sin(u-v-\eta)},\quad \gamma_4= \frac{c_1^2-b_1^2}{b_1}=-\frac{\sin(u+v+3\eta)}{\sin(u+v+\eta)}
\label{weights6vpath}
\end{eqnarray}
To obtain \eqref{genpa6v}, we have used the fact that the first step of path is horizontal with $y$ parity unspecified 
(and receives no weight $z$), and the last step is vertical, with $y=0$ (and receives the weight $w$).

\subsubsection{Asymptotics}\label{secasym6v}

We wish to take the large $n=N$ scaling limit with $\kappa=k/(2N)$ and $\lambda=\ell/N$ finite.
Further expanding \eqref{genpa6v} in powers of $z$, we find:
\begin{eqnarray}
Y_{k,\ell}&=&\sum_{P_1,P_2,P_3,P_4\geq 0\atop P_1+P_2+P_3+P_4=\ell} {\frac{k-\epsilon}{2}+P_1\choose P_1} {\frac{k+\epsilon-2}{2}+P_2\choose P_2}
{\frac{k+\epsilon}{2}\choose P_3} {\frac{k-\epsilon}{2}\choose P_4} \gamma_1^{P_1+\frac{k+\epsilon}{2}}\gamma_2^{P_2+\frac{k-\epsilon}{2}} \gamma_3^{P_3}\gamma_4^{P_4}
\nonumber \\
Y_{2\kappa N,\lambda N}&\simeq& \int_0^1 dp_2 dp_3 dp_4 e^{-N S_1^{6V'}(\kappa, p_2,p_3,p_4)}\nonumber \\
&&\!\!\!\!\!\!\!\!\!\!\!\!\!\!\!\!\!\!\!\!\!\!\!\!\!\!\!\!\!\!\!S_1^{6V'}(\kappa, p_2,p_3,p_4)=-(\kappa+\lambda-p_2-p_3-p_4)\,{\rm Log}(\kappa+\lambda-p_2-p_3-p_4)\nonumber\\
&&\!\!\!\!\!\!+(\lambda-p_2-p_3-p_4)\,{\rm Log}(\lambda-p_2-p_3-p_4)-(\kappa+p_2)\,{\rm Log}(\kappa+p_2)+p_2\,{\rm Log}(p_2)\nonumber\\
&&\!\!\!\!\!\!+p_3\,{\rm Log}(p_3)+(\kappa-p_3)\,{\rm Log}(\kappa-p_3)+p_4\,{\rm Log}(p_4)+(\kappa-p_4)\,{\rm Log}(\kappa-p_4)\nonumber\\
&&\!\!\!\!\!\!- (\kappa+\lambda-p_2-p_3-p_4)\,{\rm Log}(\gamma_1)-(\kappa+p_2)\,{\rm Log}(\gamma_2)-p_3\,{\rm Log}(\gamma_3)-p_4\,{\rm Log}(\gamma_4)
\label{asymptopath6v}
\end{eqnarray}
Here we have eliminated $P_1$ and replaced the remaining summations over $P_i$ by integrations over $p_i=P_i/n$ in $[0,1]$. Note that this covers the case of vahishing weights $\gamma_i$ for $i=3$ or $4$ as well: if $\gamma_i=0$ we simply suppress $P_i$ from the above 
expression, which in turn corresponds to taking the $p_i\to 0$ limit at finite $\gamma_i$ in \eqref{asymptopath6v}.


\subsection{Refined one-point functions and asymptotics}

%

\subsubsection{Refined partition function}

Let $Z_{n,k}^{6V'}[u,v]$ denote the {\it refined} partition function of the 6V' model on the rectangular grid of size $(2n-1)\times n$ with uniform weights (\ref{6voddweights}-\ref{6vevenweights}), 
in which the rightmost path is conditioned to first visit the
rightmost vertical line at a point at position $k\in[1,2n-1]$ (counted from bottom to top), before going vertically down until its endpoint, 
as illustrated in Fig.~\ref{fig:alltgt} (bottom left, pink domain, with the $k$ final steps removed). 
This quantity is easily related to the semi-homogeneous partition function $Z_n^{6V'}[u,v;\xi]$ as follows.
In the latter, only the weights of the last column (with spectral parameter $v_n=v+\xi$) are different, and depend on the parity of the vertex height. 
Let us denote by $({\bar a}_i,{\bar b}_i,{\bar c}_i)$, $i=o,e$ the relative 6V' weights (ratio of the value at $v+\xi$ by that at $v$): 
\begin{eqnarray*} && {\bar a}_o=\frac{\sin(u-v-\xi+\eta)}{\sin(u-v+\eta)},\quad  {\bar b}_o=\frac{\sin(u-v-\xi-\eta)}{\sin(u-v-\eta)}, \quad {\bar c}_o=1\\
&& {\bar a}_e=\frac{\sin(u+v+\xi-\eta)}{\sin(u+v-\eta)}, \quad  {\bar b}_e=\frac{\sin(u+v+\xi+\eta)}{\sin(u+v+\eta)}, \quad {\bar c}_e=1
\end{eqnarray*}

Contributions to $Z_{n,k}^{6V'}[u,v]$
have a last column with $k-1$ bottom vertices of type b (vertical step), the $k$-th vertex of type c (right turn), and the top $2n-1-k$ vertices of type a (empty). 
Splitting contributions according to the parity of the position of the point of entry into the last column of the rightmost path, we arrive at:
\begin{eqnarray*}
Z_n^{6V'}[u,v;\xi] &=&\sum_{j=1}^n Z_{n,2j-1}^{6V'}[u,v]({\bar b}_o{\bar b}_e)^{j-1} {\bar c}_o ({\bar a}_e {\bar a}_o)^{n-j}  +\sum_{j=1}^{n-1}Z_{n,2j}^{6V'}[u,v]
({\bar b}_o{\bar b}_e)^{j-1} {\bar b}_o {\bar c}_e {\bar a}_o ({\bar a}_e {\bar a}_o)^{n-j-1}\\
&=&({\bar a}_e {\bar a}_o)^{n-1}\sum_{j=1}^n \tau^{j-1} \{ Z_{n,2j-1}^{6V'}[u,v]+ Z_{n,2j}^{6V'}[u,v]\, \sigma \}
\end{eqnarray*}
where we have used the values ${\bar c}_o={\bar c}_e=1$ and the parameters
\begin{eqnarray*}
\tau&:=&\frac{{\bar b}_o{\bar b}_e}{{\bar a}_e {\bar a}_o}
=\frac{\sin(u-v-\xi-\eta)\sin(u+v+\xi+\eta)\sin(u-v+\eta)\sin(u+v-\eta)}{\sin(u-v-\eta)\sin(u+v+\eta)\sin(u-v-\xi+\eta)\sin(u+v+\xi-\eta)}\\
\sigma&:=&\frac{{\bar b}_o }{{\bar a}_e}=\frac{\sin(u-v-\xi-\eta)\sin(u+v-\eta)}{\sin(u-v-\eta)\sin(u+v+\xi-\eta)} 
\end{eqnarray*}

For use with the Tangent Method, we need to consider the refined one-point function $H_{n,k}[u,v]$ defined as the ratio of the partition function of the 6V' model
in which the topmost path ends at position $k$ with a horizontal last step between the $n-1$-st vertical and the rightmost vertical, to that of the usual 6V' partition function. 
Note that the numerator is slightly different from the refined partition function $Z_{n,k}^{6V'}[u,v]$ as the rightmost path does not continue with $k$ vertical steps after
hitting the rightmost vertical.  Consequently, we must replace the $k-1$ corresponding b-type weights with a-type weights:
$$H_{n,2j-1}[u,v]:=\left(\frac{a_o\,a_e}{b_o\, b_e}\right)^{j-1}\, \frac{Z_{n,2j-1}^{6V'}[u,v]}{Z_{n}^{6V'}[u,v]} ,\quad 
H_{n,2j}[u,v]:=\frac{a_o}{b_o} \,\left(\frac{a_o\,a_e}{b_o\, b_e}\right)^{j-1}\,  \frac{Z_{n,2j}^{6V'}[u,v]}{Z_{n}^{6V'}[u,v]}$$
In terms of the one-point function $H_n^{6V'}[u,v;\xi]$ \eqref{ratioref},  the above identity reads:
\begin{eqnarray*}H_n^{6V'}[u,v;\xi]=({\bar a}_e {\bar a}_o)^{n-1} \sum_{j=1}^n t^{j-1} \{ H_{n,2j-1}[u,v]+ H_{n,2j}[u,v]\, s \}\\
\end{eqnarray*}
where 
\begin{eqnarray*}
t&=&\tau \,\frac{b_o\,b_e}{a_o\, a_e} =\frac{\sin(u-v-\xi-\eta)\sin(u+v+\xi+\eta)}{\sin(u-v-\xi+\eta)\sin(u+v+\xi-\eta)}\\
s&=&\sigma \, \frac{b_o}{a_o} = \frac{\sin(u-v-\xi-\eta)\sin(u+v-\eta)}{\sin(u-v+\eta)\sin(u+v+\xi-\eta)}
\end{eqnarray*}

\subsubsection{Asymptotics}

We wish to estimate the leading behavior of the one-point function $H_{N,k}[u,v]$ for large $N$ and $\kappa=k/( 2N)$ finite. To this end, we use the asymptotics of the function $H_N^{6V'}[u,v;\xi]$ (Theorem \ref{asympto6vponept}) to estimate for large $N$:
\begin{eqnarray}
&&\frac{H_N^{6V'}[u,v;\xi]}{({\bar a}_e {\bar a}_o)^{N-1}} \simeq e^{-N\varphi^{6V'}[u,v;\xi]} \nonumber \\
&&\varphi^{6V'}[u,v;\xi]=\psi^{6V'}[u,v;\xi]+{\rm Log}\left( \frac{\sin(u-v-\xi+\eta)\sin(u+v+\xi-\eta)}{\sin(u-v+\eta)\sin(u+v-\eta)}\right) \nonumber \\
&&\qquad\qquad\quad\  =-{\rm Log}\left(\frac{\sin(2v)\sin(u-v-\xi-\eta)\,\sin(u+v+\xi+\eta)}{\sin(2v+\xi)\sin(u-v-\eta)\,\sin(u+v+\eta)}\right)\nonumber \\
&&\qquad\qquad -{\rm Log}\left(
\frac{\sin(\al\xi)\sin(\al(\xi+2v+2\eta))\sin(\al(u-v-\eta))\sin(\al(u+v+\eta))}{\al \,\sin(\xi)\sin(2\al (v+\eta))\sin(\al(u-v-\xi-\eta))\sin(\al(u+v+\xi+\eta))} \right)
\label{phi6v}
\end{eqnarray}

This leads finally to the following result.
\begin{thm}\label{6vpasythm}
The large $N$ asymptotics of the refined one-point function for the 6V'model are given by:
\begin{eqnarray} H_{N,2\kappa N}[u,v]&\simeq& \oint \frac{dt}{2i\pi t} e^{-N S_0^{6V'}(\kappa,t) }\nonumber \\
S_0^{6V'}(\kappa,t)&=& \varphi^{6V'}[u,v,\xi]+\kappa \, {\rm Log}(t) \label{action6v}
\end{eqnarray}
with $\varphi^{6V'}[u,v,\xi]$ as in \eqref{phi6v}, and
where $\xi$ can be thought of as an implicit function of the variable $t$, upon inversion of the relation
\begin{equation}\label{txi6v}
t=t_{6V'}[\xi]:=\frac{\sin(u-v-\xi-\eta)\sin(u+v+\xi+\eta)}{\sin(u-v-\xi+\eta)\sin(u+v+\xi-\eta)} 
\end{equation}
\end{thm}

The leading contribution to \eqref{action6v} is determined by the solution of the saddle point equation 
$\partial_t S_0^{6V'}(\kappa,t)=0$ or equivalently $\partial_\xi S_0^{6V'}(\kappa,t_{6V'}[\xi])=0$, leading to:
\begin{equation}\label{sapo6v} \kappa=\kappa_{6V'}[\xi]:=-\frac{t_{6V'}[\xi]}{\partial_\xi t_{6V'}[\xi]} \partial_\xi \varphi^{6V'}[u,v;\xi] 
\end{equation}
Explicitly we have:
\begin{eqnarray}
\kappa_{6V'}[\xi]
&=&\left\{\cot(u-v-\eta-\xi)+\cot(\xi)+\cot(\xi+2v)-\cot(u+v+\eta+\xi)\right.\nonumber \\
&&\!\!\!\!\!\!\!\!\!\!\!\!\!\!\!\left. -\al\big(\cot(\al(u-v-\eta-\xi))+\cot(\al\xi)+\cot(\al(\xi+2v+2\eta))-\cot(\al(u+v+\eta+\xi)) \big)\right\} \nonumber \\
&&\quad \times\  \frac{\sin(u + v -\eta+\xi) \sin(u + v +\eta +\xi)\sin(u - v -\eta-\xi) \sin(u - v +\eta -\xi)}{\sin(2\eta)\big(\cos(2\eta)-\cos(2u)\cos(2v+2\xi)\big)}\label{rsol}
\end{eqnarray}
with $\al=\frac{\pi}{\pi-2\eta}$ as usual.

\subsection{Arctic curves}

\subsubsection{NE branch} As explained above, the first application of the Tangent Method gives access to the portion of the arctic curve situated in the NE corner of the rectangular domain.

\begin{thm}\label{6VpNEthm}
The NE branch of the arctic curve for the 6V' model as predicted by the Tangent Method is given by the parametric equations
$$ x=X_{NE}^{6V'}[\xi]= \frac{B'[\xi]}{A'[\xi]} \qquad y=Y_{NE}^{6V'}[\xi]=B[\xi]-\frac{A[\xi]}{A'[\xi]}B'[\xi]$$
with the parameter range:
$$ 
\xi\in \big[\eta+|u|-v-\pi,0\big] 
$$
and where 
$$A[\xi]=2\,\frac{\sin(u-v-\eta-\xi)\sin(u-v+\eta-\xi)\sin(u+v-\eta+\xi)\sin(u+v+\eta+\xi)}{\sin(\xi-2\eta)\sin(\xi)\big(\cos(2\eta)- \cos(2u)\cos(2v+2\xi)\big)}$$
and $B[\xi]=2\kappa_{6V'}[\xi]$, with $\kappa_{6V'}[\xi]$ is as in \eqref{rsol}.
\end{thm}
\begin{proof}
We may now bring together the ingredients of the Tangent Method. We determine the family of tangents $F_\xi(x,y)=y+A[\xi]x-B[\xi]$ defined in Sect. \ref{sectan}.
We already identified the intercept $B[\xi]=2\kappa_{6V'}[\xi]$ with $\kappa_{6V'}[\xi]$ given by \eqref{rsol}. 
To determine the slope $A[\xi]=2\kappa/\lambda$, we must find the leading contribution to the total partition function
\begin{eqnarray*} 
&&\sum_{k=1}^{2n-1} H_{N,k}[u,v] \, Y_{k,\ell} \simeq \int_0^1 d\kappa H_{N,2\kappa N}[u,v]\, Y_{2\kappa N,\lambda N}\simeq 
\int_0^1 d\kappa dp_2 dp_4,dp_5 e^{-N S^{6V'}(\kappa,p_2,p_4,p_5,t)} \\
&&S^{6V'}(\kappa,p_2,p_4,p_5,t):=S_0^{6V'}(\kappa,t)+S_1^{6V'}(\kappa,p_2,p_4,p_5)
\end{eqnarray*}
with $S_0^{6V'}(\kappa,t)$ as in \eqref{action6v} and $S_1^{6V'}(\kappa,p_2,p_3,p_4)$ as in \eqref{asymptopath6v}. As in the 6V case, the saddle-point equation 
$\partial_\xi S^{6V'}=0$ is solved by \eqref{rsol}, and amounts to parameterizing $\kappa=\kappa_{6V'}[\xi]$ in terms of the parameter $\xi$.
The saddle-point equations
$\partial_\kappa S^{6V'}=\partial_{p_2}S^{6V'}=\partial_{p_3}S^{6V'}=\partial_{p_4}S^{6V'}=0$ give rise to the system of algebraic equations:
\begin{eqnarray*}
\frac{t}{\gamma_1\gamma_2}&=& \frac{(p_2+\kappa)(\kappa+\lambda-p_2-p_3-p_4)}{(p_3-\kappa)(p_4-\kappa)}\\
\frac{\gamma_1}{\gamma_2}&=&\frac{(p_2+\kappa)(\lambda-p_2-p_3-p_4)}{p_2(\kappa+\lambda-p_2-p_3-p_4)}\\
\frac{\gamma_1}{\gamma_3}&=& \frac{(\kappa-p_3)(\lambda-p_2-p_3-p_4)}{p_3(\kappa+\lambda-p_2-p_3-p_4)}\\
\frac{\gamma_1}{\gamma_4}&=& \frac{(\kappa-p_4)(\lambda-p_2-p_3-p_4)}{p_4(\kappa+\lambda-p_2-p_3-p_4)}\\
\end{eqnarray*}
Substituting the values of the weights $\gamma_i$ \eqref{weights6vpath} and $t=t_{6V'}[\xi]$ \eqref{txi6v}, we find the unique solution such that $\lambda,\kappa>0$:
\begin{eqnarray}
\frac{p_2}{\kappa}&=& -\frac{\sin(u+v+\eta)\sin(\xi)}{\sin(2\eta)\sin(u+v-\eta+\xi)} \qquad 
\frac{p_3}{\kappa}= \frac{\sin(u-v-3\eta)\sin(\xi)}{\sin(2\eta)\sin(u-v-\eta-\xi)} \nonumber \\
\frac{p_4}{\kappa}&=& \frac{\sin(u+v+3\eta)\sin(\xi)}{\sin(2\eta)\sin(u+v+\eta+\xi)}\nonumber \\
\frac{\kappa}{\lambda}&=&
\frac{\sin(u-v-\eta-\xi)\sin(u-v+\eta-\xi)\sin(u+v-\eta+\xi)\sin(u+v+\eta+\xi)}{\sin(\xi-2\eta)\sin(\xi)\big(\cos(2\eta)- \cos(2u)\cos(2v+2\xi)\big)}\nonumber \\
&&\label{valal}
\end{eqnarray}
Using the parametrization $\kappa=\kappa_{6V'}[\xi]$,
we may interpret the last equation as determining $\lambda$ as a function $\lambda_{6V'}[\xi]$ of the parameter $\xi$, where:
\begin{equation}\label{lamofxi}
\lambda_{6V'}[\xi]:=\kappa_{6V'}[\xi]\, \frac{\sin(\xi-2\eta)\sin(\xi)\big(\cos(2\eta)- \cos(2u)\cos(2v+2\xi)\big)}{\sin(u-v-\eta-\xi)\sin(u-v+\eta-\xi)\sin(u+v-\eta+\xi)\sin(u+v+\eta+\xi)}
\end{equation}
To summarize, we have found the most likely exit point $\kappa$
as an implicit function of the arbitrary parameter $\lambda$, via the parametric equations $(\kappa,\lambda)=(\kappa_{6V'}[\xi],\lambda_{6V'}[\xi])$, which 
results in the family of tangent lines with equations $F_\xi(x,y)=0$.
The theorem follows from the expressions \eqref{acurve}, by identifying the slope $A[\xi]=2\kappa_{6V'}[\xi]/\lambda_{6V'}[\xi]$, while the range of the parameter $\xi$ corresponds to imposing $A[\xi] \in [0,\infty)$.
\end{proof}

\subsubsection{SE branch} As mentioned in Sect.~\ref{obsec}, a simple transformation of the model gives access to the portion of the arctic curve situated in the SE corner of the rectangular domain: we must change parameters $(u,v)\mapsto (u^*,v^*)=(-u,-v-\pi)$ and coordinates $(x,y)\mapsto (x,2-x-y)$. 

\begin{thm}\label{6VpSEthm}
The SE branch of the arctic curve for the 6V' model is given by the parametric equations
$$ x=X_{SE}^{6V'}[\xi]={X_{NE}^{6V'}}^*[\xi] \qquad y=Y_{SE}^{6V'}[\xi]=2-{Y_{NE}^{6V'}}^*[\xi]\qquad (\xi\in [\eta+|u|+v,0]  )$$
with $X_{NE}^{6V'},Y_{NE}^{6V'}$ as in Theorem \ref{6VpNEthm}, and where the superscript $*$ stands for the transformation $(u,v)\mapsto (u^*,v^*)=(-u,-v-\pi)$,
which we have also applied to the parameter range.
\end{thm}

\begin{remark}\label{symrem}
In the case $v=-\frac{\pi}{2}=v^*$, we note that the equation of the tangent is invariant under $u\to -u=u^*$. We deduce that the arctic curve is symmetric w.r.t. the line $y=1$, and that the SE branch
is simply the reflection of the NE branch: $X_{SE}=X_{NE}$, $Y_{SE}=2-Y_{NE}$. This is no longer true when $v\neq -\frac{\pi}{2}$.
\end{remark}

\subsection{Examples}
In this section, we illustrate Theorems \ref{6VpNEthm} and \ref{6VpSEthm} with some concrete examples.

\subsubsection{The ``6V" case $u=0$, $v=-\frac{\pi}{2}$}

\begin{figure}
\begin{center}
\begin{minipage}{0.4\textwidth}
        \centering
        \includegraphics[width=4cm]{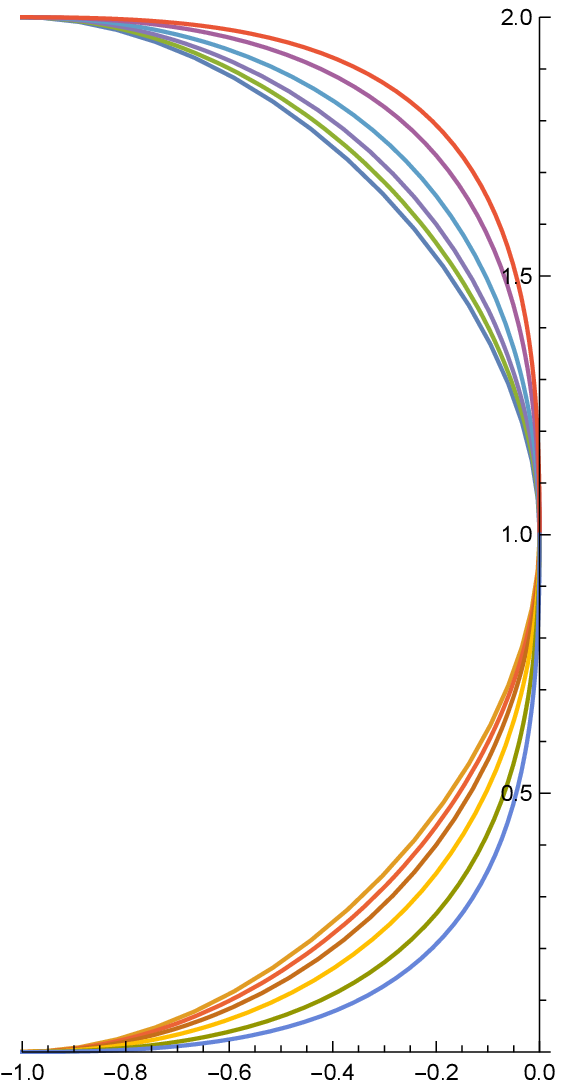} 
    \end{minipage}\hfill
    \begin{minipage}{0.6\textwidth}
        \centering
        \includegraphics[width=7.8cm]{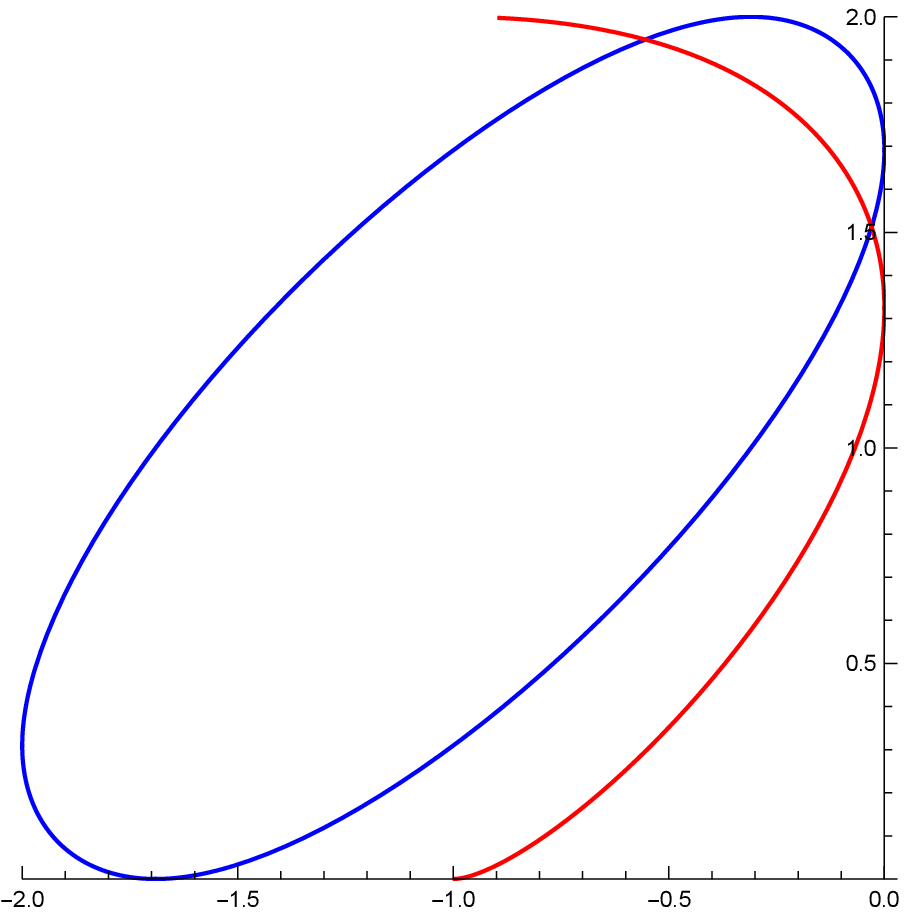}  
    \end{minipage}
\end{center}
\caption{\small Arctic curves for the 6V' model with parameter $u=0$. Left: case $v=-\frac{\pi}{2}$, with $\eta$ ranging from $0^+$ (outermost curve) to $\frac{\pi}{2}^-$
(innermost curve): all curves are symmetric w.r.t. the line $y=1$, and coincide with those of the 6V-DWBC model (NE/SE portions). 
Right: Arctic curve of the 6V' model for $\eta=\frac{\pi}{3}$, $u=0$ and $v=-\frac{\pi}{2}-\frac{\pi}{12}$ (red curve)
compared with the arctic curve of the 6V-DWBC model (scaled by a factor of 2) for the same value of 
$\eta$ and  the value $u=\frac{\pi}{2}+\frac{\pi}{12}$ leading to the same Boltzmann weights (blue curve).}
\label{fig:arctic1}
\end{figure}

The condition $u=0$ implies that all horizontal spectral parameters are equal, and that the Boltzmann weights (\ref{6voddweights}-\ref{6vevenweights}) lose their dependence on the parity of the row (upon taking $\rho_e=\rho_o=\rho$). In fact this gives a mapping to the weights \eqref{6vweights} of the ordinary 6V model via $(u_{6V'},v_{6V'})\mapsto (0,-u_{6V})$.
We may wonder how the U-turn boundary condition has affected the thermodynamics of the 6V-DWBC model. In fact, extending the usual connection between ASM and VSASM, it is easy to identify the 6V' model at $u=0$ with a 6V-DWBC model on a grid of ``double" size $2n+1\times 2n+1$, and whose configurations are vertically symmetric, i.e. invariant under reflection 
w.r.t. a vertical line. As noted in Remark \ref{rem6v}, the parameter $u$ in the 6V-DWBC case may be interpreted as an anisotropy parameter. Indeed,  the value $u=\frac{\pi}{2}$ corresponds 
for the 6V-DWBC model to 
identical weights $a=b$ which imply invariance of the partition function under reflection w.r.t. a horizontal line. However, when $u\neq \frac{\pi}{2}$, this is no longer true, as 
the weights $a\neq b$ are interchanged in the reflection. As a consequence, the tangency points of the arctic curve to the boundary of the domain move away from their symmetric positions. We expect therefore a connection between 6V-DWBC and 6V' models only at the isotropic point $u_{6V}=\frac{\pi}{2}$, corresponding to $(u_{6V'},v_{6V'})=(0,-\frac{\pi}{2})$.
Note that this point corresponds to the $\tau$-enumeration of VSASM (for the 6V' side) and ASM (for the 6V side), with $\tau= 4\sin^2(\eta)$.

\begin{thm}\label{tenum}
For arbitrary $0<\eta<\frac{\pi}{2}$, the arctic curve for the 6V' model with $(u_{6V'},v_{6V'})=(0,-\frac{\pi}{2})$ as obtained via the Tangent Method assuming Conjecture \ref{freeconj} holds is identical to that of the 6V-DWBC model with $u_{6V}=\frac{\pi}{2}$ in the NE/SE sector, up to global rescaling.
\end{thm}
\begin{proof}
As our choice of parameters is invariant under the symmetry $*$ for both the 6V' case $(u_{6V'},v_{6V'})^*=(u_{6V'},v_{6V'})=(0,-\frac{\pi}{2})$ and the 6V case $u_{6V}^*=u_{6V}=\frac{\pi}{2}$, we simply have to compare the envelope of the corresponding families of tangent lines leading to the NE branches, as given by Theorems \ref{6VNEthm} and \ref{6VpNEthm}. We have the two families (we add a superscript 6V,6V' to avoid ambiguities):
$$ y+A^{6V}[\xi]x -B^{6V}[\xi]=0 \quad {\rm and} \quad y+A^{6V'}[\xi]x -B^{6V'}[\xi] =0 $$
We find:
$$ \lim_{u\to 0,v\to-\frac{\pi}{2}} A^{6V'}[\xi]=\lim_{u\to \frac{\pi}{2}} A^{6V}[\xi],\qquad \lim_{u\to 0,v\to-\frac{\pi}{2}} B^{6V'}[\xi] =2 \lim_{u\to \frac{\pi}{2}} B^{6V}[\xi] $$
while the 6V and 6V' ranges of the parameter $\xi$ coincide with $\xi \in [\eta-\frac{\pi}{2},0]$. We deduce that upon rescaling of $x$ and $y$ by a factor of $2$ the two families are identical,
and conclude that $(X_{NE}^{6V'}[\xi],Y_{NE}^{6V'}[\xi])=2 (X_{NE}^{6V}[\xi],Y_{NE}^{6V}[\xi])$. The SE branch identification follows immediately from our remark on the symmetry $*$,
leading to $(X_{SE}^{6V'}[\xi],Y_{SE}^{6V'}[\xi])=2 (X_{SE}^{6V}[\xi],Y_{SE}^{6V}[\xi])$ as well, and the Theorem follows.
\end{proof}

A particular case of Theorem \ref{tenum} corresponds to the uniform case where the 6V-DWBC model boils down to the enumeration of ASM, and the 6V' model to that of VSASM. 
The arctic curves for both these cases were derived in \cite{CP2009,CP2010} and \cite{DFLAP} respectively, and shown to coincide.

For illustration, we have represented in Fig.~\ref{fig:arctic1} (left) the corresponding arctic curves for some values of $\eta$ ranging from $0^+$ to $\frac{\pi}{2}^-$: the curves are identical to the NE/SE portions of the arctic curve of the 6V-DWBC model, upon a rescaling by a global factor of $2$. 
For $\eta\to 0^+$, we find  the following limiting arctic curve:
$$(X_{NE},Y_{NE})\vert_{\eta\to 0}= \left( \frac{1}{\pi}(2\xi-\sin(2\xi)),1-\frac{1}{\pi}(2\xi+\sin(2\xi))\right)\quad (\xi\in [-\frac{\pi}{2},0])\\
$$
The limit  $\eta\to \frac{\pi}{2}^-$ is singular, as the parameter $\al=\pi/(\pi-2\eta)$ diverges. However, one can take a double scaling limit
$\eta=\frac{\pi}{2}-\epsilon$, $\xi=\epsilon \zeta$, and $\epsilon\to 0$, in which case the limiting curve reads:
\begin{eqnarray*}X_{NE}\vert_{\epsilon\to 0}&=&
\frac{(2+\zeta)^2(\cos(2\pi\zeta)-1+2\pi\zeta^2(\pi(1-\zeta^2)\cos(\pi\zeta)+2\zeta \sin(\pi\zeta))}{4(1+\zeta+\zeta^2)\sin^2(\pi \zeta)}\\
Y_{NE}\vert_{\epsilon\to 0}&=&
\frac{(1+\zeta)^2(3\sin^2(\pi\zeta)+\pi(1-\zeta)^2(\pi\zeta(2+\zeta)\cos(\pi\zeta)-2(1+\zeta)\sin(\pi\zeta))}{2(1+\zeta+\zeta^2)\sin^2(\pi \zeta)}  
\end{eqnarray*}
for $\zeta\in (-1,0]$.

By contrast, in the anisotropic case where $u_{6V'}=0$ but $v_{6V'}\neq -\frac{\pi}{2}$ (and $u_{6V}=-v_{6V'}\neq \frac{\pi}{2}$), the arctic curves no longer coincide.
For illustration, 
the predicted NE and SE portions of arctic curve of the 6V' model at $\eta=\frac{\pi}{3}$, $u_{6V'}=0$, $v_{6V'}= -\frac{\pi}{2}-\frac{\pi}{12}$ are depicted in Fig.~\ref{fig:arctic1} (right), together with the arctic curve of the 6V-DWBC model with the same values of the weights (i.e. with same $\eta$ and $u_{6V}=-v_{6V'}=\frac{\pi}{2}+\frac{\pi}{12}$): the resulting curves are very different. In particular, the 6V' curve is anchored at the endpoints $(-1,0)$ and $(-1,2)$ with horizontal tangents, whereas the 6V curve has horizontal tangents at different points $\big(2(1-\frac{2}{\sqrt{3}}),2\big)\simeq (-.309,2)$ and $\big(4(\frac{1}{\sqrt{3}}-1),0\big)\simeq (-1.69,0)$.

\subsubsection{The ``Free fermion" case $\eta=\frac{\pi}{4}$}

\begin{figure}
\begin{center}
\begin{minipage}{0.5\textwidth}
        \centering
        \includegraphics[width=5cm]{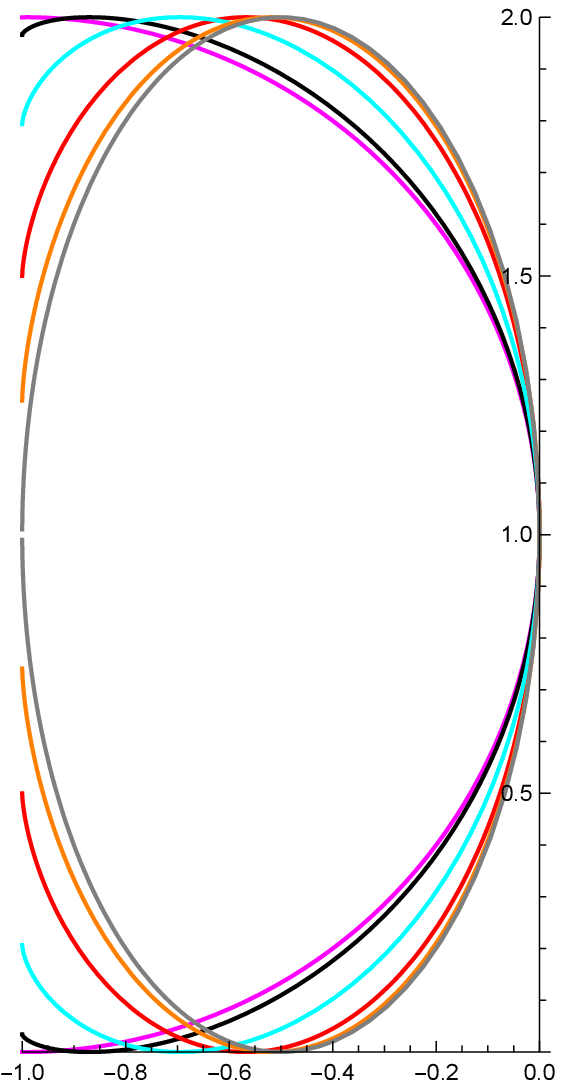} 
    \end{minipage}\hfill
    \begin{minipage}{0.5\textwidth}
        \centering
        \includegraphics[width=5cm]{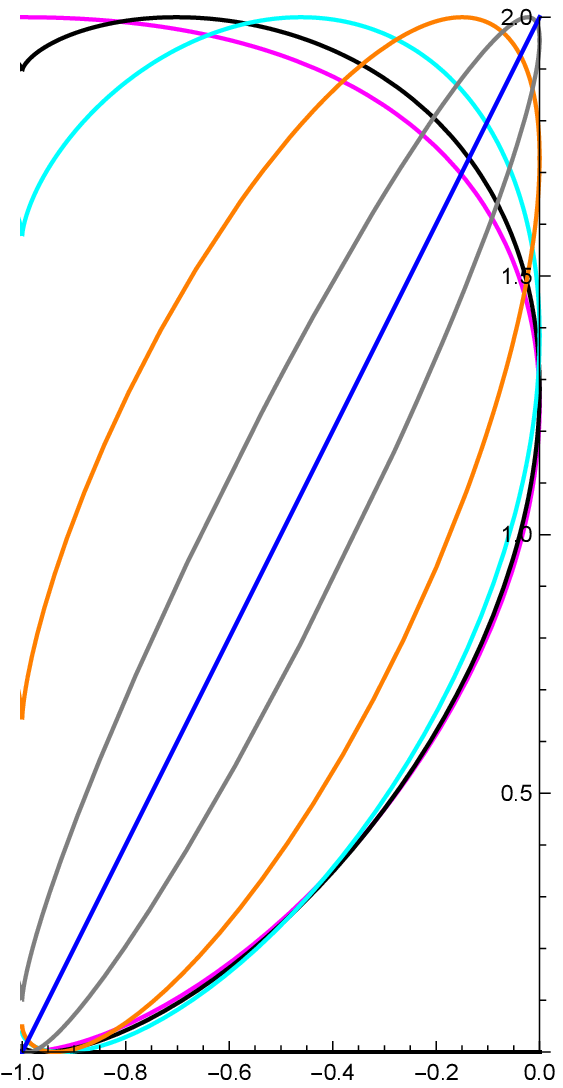}  
    \end{minipage}
\end{center}
\caption{\small Arctic curves for the free fermion case $\eta=\frac{\pi}{4}$ of the 6V' model. Left: symmetric case $v=-\frac{\pi}{2}$, 
with $u$ ranging from $0^+$ (outermost curve on the vertical $x=-1$) to $\frac{\pi}{4}^-$
(innermost curve on the vertical $x=-1$): all curves are symmetric w.r.t. the line $y=1$. 
Right: asymmetric case $v=-\frac{\pi}{2}-\frac{\pi}{12}$, with $u$ ranging from $0^+$ (outermost curve on the vertical $x=-1$) 
to $\frac{\pi}{6}^-$ (innermost curve on the vertical $x=-1$ degenerating to a segment).}
\label{fig:arctic2}
\end{figure}

This case is nicer in the sense that arctic curves are expected to be analytic. In particular, we checked that the SE portion of the arctic curve is indeed the analytic continuation of the NE one. In Fig.\ref{fig:arctic2} (left) we represent arctic curves for $\eta=\frac{\pi}{4}$ and the isotropic value $v=-\frac{\pi}{2}$ with $u$ ranging from $0^+$ to $\frac{\pi}{4}^-$.
The $u=0$ arctic curve is given by $(x,y)=\big(\cos(2\xi)-1,\sin(2\xi)+1\big)$: it is the half-circle $(x+1)^2+(y-1)^2=1$ with $x\geq -1$, first obtained in \cite{PRarctic}. The case $u=\frac{\pi}{4}$ is singular. As before, 
we consider the double-scaling limit $u=\frac{\pi}{4}+\epsilon$ and $\xi=\epsilon \, \zeta$, leading to the limiting curve:
$$(x,y)=\left(- \frac{\zeta^2}{1+\zeta^2},\frac{(1+\zeta)^2}{1+\zeta^2}\right) $$
equal to the ellipse $(2x+1)^2 +(y-1)^2 =1$ inscribed in the rectangle $[-1,0]\times [0,2]$. We see that the gap between the endpoints of the arctic curve on the vertical $x=-1$
ranges from $2$ (semi-circle case) to $0$ (ellipse case). This type of arctic curve was also encountered when considering lozenge tilings (an archetypical free fermion model)
with free boundary conditions in \cite{DFR}.

In Fig.~\ref{fig:arctic2} (right) we represent arctic curves for $\eta=\frac{\pi}{4}$ and a sample anisotropic value $v=-\frac{\pi}{2}-\frac{\pi}{12}$ with $u$ ranging from 
$0^+$ to $\frac{\pi}{6}^-$. The $u=0$ arctic curve is the quartic:
$$(x,y)=\left(-\sin^2(\xi)\big( 1+\frac{\cos^2(\xi)}{\cos^2(\frac{\pi}{6}-\xi)}\big),\frac{1}{2}\,\frac{\cos^2(\frac{\pi}{6}+\xi)}{\cos^2(\frac{\pi}{6}-\xi)}\big(1+\sqrt{3}\sin(\frac{\pi}{3}+2\xi)\big)\right) $$
The limit  $u\to \frac{\pi}{6}^-$ is singular, but the double scaling limit $u=\frac{\pi}{6}-\epsilon$ and $\xi= \zeta \sqrt{\epsilon}$ and $\epsilon\to 0$ leads to the segment
$$ (x,y)=\left(-1+\frac{1}{1+\frac{4}{\sqrt{3}}\zeta^2},\frac{2}{1+\frac{4}{\sqrt{3}}\zeta^2}\right) \qquad (\zeta\in [0,\infty) )$$
that joins point $(-1,0)$ to $(0,2)$.

\subsubsection{20V case}

\begin{figure}
\begin{center}
\includegraphics[width=8cm]{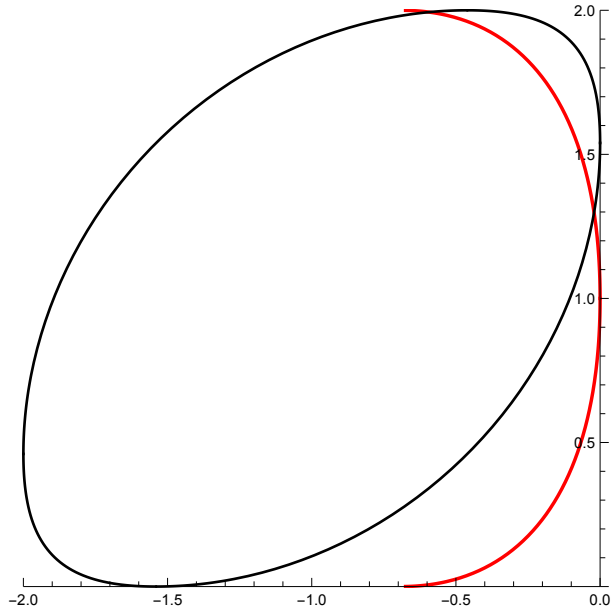}
\end{center}
\caption{\small Arctic curve (NE and SE portions) of the ``20V point" of the 6V' model, with $\eta=u=\frac{\pi}{8}$, $v=-\frac{\pi}{2}$ (symmetric curve in red) together with the arctic curve of the associated 6V model, with $\eta=\frac{\pi}{8}$, $u=\frac{5\pi}{8}$, scaled by a factor of $2$.}
\label{fig:arctic3}
\end{figure}

This case corresponds to $\eta=\frac{\pi}{8}$, $u=\eta=\frac{\pi}{8}$ and $v=-4\eta=-\frac{\pi}{2}$, by analogy with the 6V model with DWBC whose partition function is identical to
that of the uniformly weighted 20V model with DWBC1,2 studied in Refs. \cite{DFGUI,BDFG}. The corresponding NE/SE portions of arctic curve are depicted in Fig.~\ref{fig:arctic3}.

\subsubsection{Generic case}
\begin{figure}
\begin{center}
\includegraphics[width=5cm]{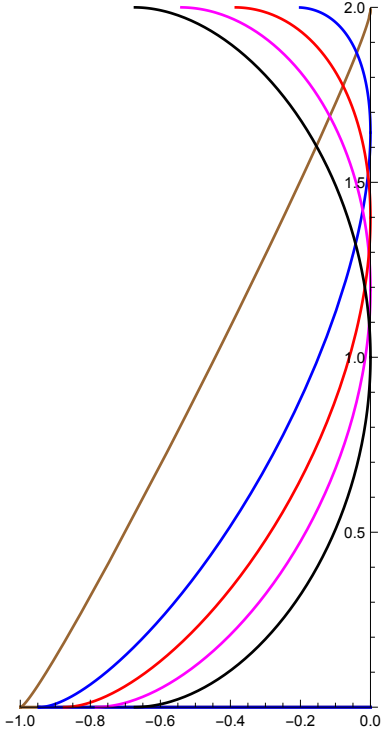}
\end{center}
\caption{\small Arctic curve (NE and SE portions) of the 6V' model  with $\eta=\frac{\pi}{3}$, $u=\frac{\pi}{12}$ and $v$ varying from $-\frac{\pi}{2}$ (leftmost curve on top) to $-\frac{\pi}{2}-\frac{\pi}{12}$ (rightmost curve on top).}
\label{fig:arctic4}
\end{figure}

We present a ``generic case" in Fig.~\ref{fig:arctic4} with $\eta=\frac{\pi}{3}$, $u=\frac{\pi}{12}$ and $v$ varying from $-\frac{\pi}{2}$ to $-\frac{\pi}{2}-\frac{\pi}{12}$. As before the case $v=-\frac{\pi}{2}-\frac{\pi}{12}$ is singular, but may be investigated via a double scaling limit, leading to the segment joining $(-1,0)$ to $(0,2)$.

\section{20V model with DWBC3}\label{20vsec}

\subsection{Partition function and one-point function}

In Ref.~ \cite{DF20V} the partition function of the 20V-DWBC3 model was related to that of the 6V' model, by use of the integrability of the weights \eqref{weights20V}. More precisely, let us denote by $Z_n^{20V}[u,\bv]$ the semi-homogeneous
partition function of the 20V-DWBC3 model, with all horizontal spectral parameters equal to $\eta+u$, all diagonal
ones to $-u$ and arbitrary vertical spectral parameters $\bv=v_1,v_2,...,v_n$, and by 
$Z_n^{6V'}[u,\bv]$ the partition function of the 6V' model with horizontal spectral parameters all equal to $u$ and arbitrary vertical spectral parameters $\bv$. We have:
\begin{thm}\label{20v6vthm}{\cite{DF20V}}
The following relation holds for all $n\geq 1$:
\begin{equation}Z_n^{20V}[u,\bv]=\al^{n(3n-1)/2} Z_n^{6V'}[u,\bv] \sin(2u+2\eta)^{n(3n-1)/2}  \prod_{i=1}^n \sin(u-v_i-\eta)^{i-1} \sin(\eta-u-v_i)^{i}
\end{equation}
\end{thm}

In the homogeneous case where all $v_i=v$ for all $i$, this reduces to:

\begin{equation}\label{vingtvpf}Z_n^{20V}[u,v]=
\al^{n(3n-1)/2} Z_n^{6V'}[u,v] \sin(2u+2\eta)^{n(3n-1)/2} \sin(u-v-\eta)^{n(n-1)/2} \sin(\eta-u-v)^{n(n+1)/2} 
\end{equation}

Next we define the one-point function $H_n^{20V}[u,v;\xi]$ as the ratio:
\begin{eqnarray}\label{inho20v}
H_n^{20V}[u,v;\xi]&:=&\frac{Z_n^{20V}[u,v;\xi]}{Z_n^{20V}[u,v]}\nonumber \\
&=&
\left(\frac{\sin(u-v-\xi-\eta)}{\sin(u-v-\eta)}\right)^{n-1} \, \left(\frac{\sin(\eta-u-v-\xi)}{\sin(\eta-u-v)}\right)^n \, H_n^{6V'}[u,v;\xi]
\end{eqnarray}
where in $Z_n^{20V}[u,v;\xi]$ we have kept $v_1=v_2=\cdots=v_{n-1}=v$ but relaxed the last value $v_n=v+\xi$.
Like in the 6V and 6V' cases, this function will be a crucial ingredient of the Tangent Method.

\subsection{Refined one-point functions and asymptotics}

\subsubsection{Refined partition function}\label{ref20vsec}

\begin{figure}
\begin{center}
\includegraphics[width=7.5cm]{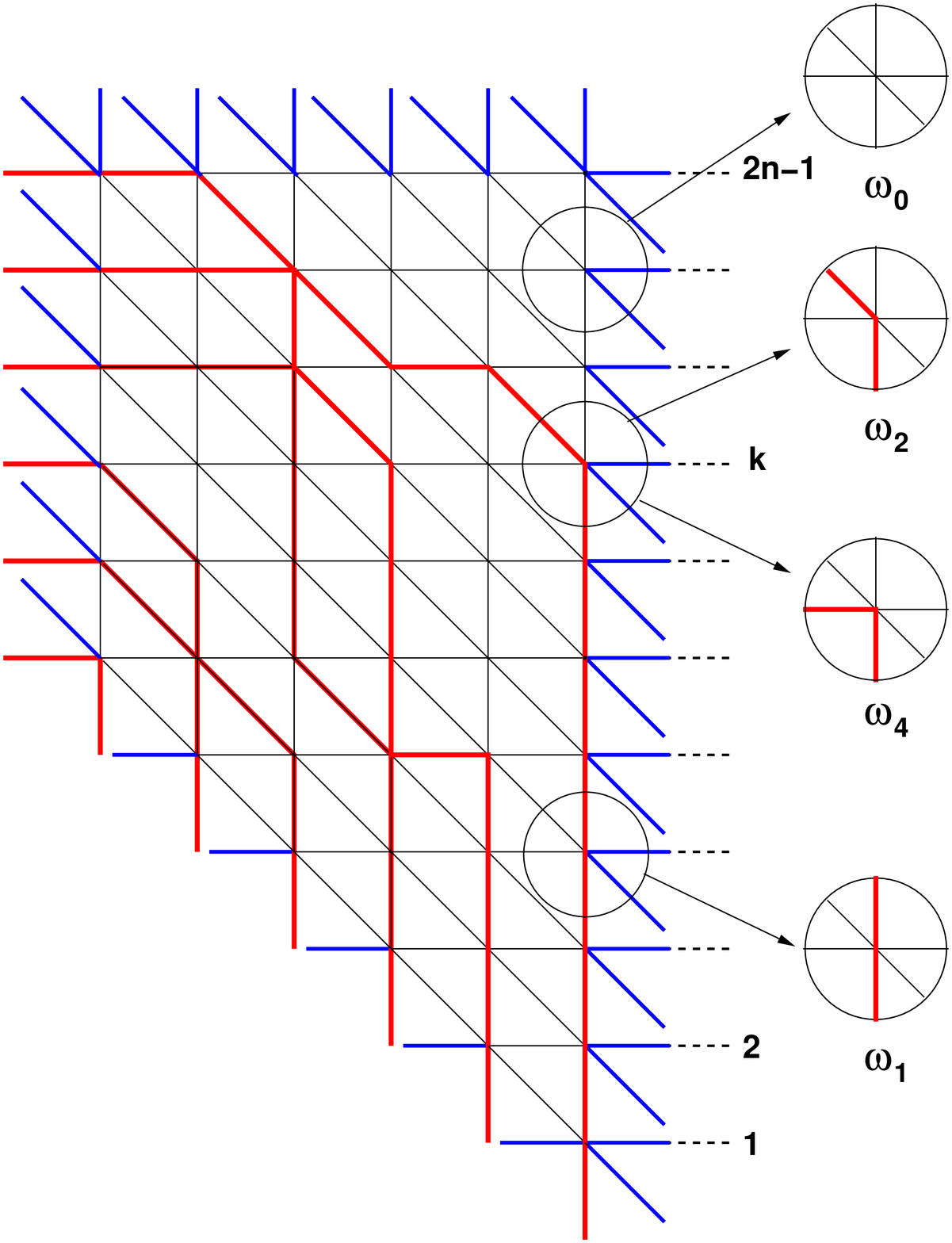}
\end{center}
\caption{\small A sample contribution to the refined partition function $Z_{n,k}^{20V}[u,v]$. In this particular example, the contribution pertains to $Z_{n,k}^{20V \backslash}[u,v]$. The medallions detail the various weights involved in the last column.}
\label{fig:ref20v}
\end{figure}

Let $Z_{n,k}^{20V}[u,v]$ denote the partition function of the 20V model on the quadrangle $Q_n$ with uniform weights \eqref{weights20V}, 
in which the rightmost path is conditioned to first visit the rightmost column
at a point at position $k\in [1,2n-1]$ (see Fig.~\ref{fig:ref20v} for an illustration).  
We may split this partition function into $Z_{n,k}^{20V}[u,v]=Z_{n,k}^{20V-}[u,v]+Z_{n,k}^{20V \backslash}[u,v]$ according to whether the topmost path  accesses 
the point $k$ via a horizontal $-$ or diagonal $\backslash$ step, before terminating with $k$ vertical steps until its endpoint. This quantity is easily related to the 
partially inhomogeneous partition function $Z_{n}^{20V}[u,v;\xi]$ \eqref{inho20v}. Recall that for the latter the weights are homogeneous with parameters $u,v$ except for the
$n$-th column in which $v$ is replaced by $v+\xi$. Let ${\bar \omega_i}:=\omega_i[u,v+\xi]/\omega_i[u,v]$ be the {\it relative} Boltzmann weights for the last column, 
as compared to the homogeneous values. Specifically, using the weights:
\begin{eqnarray*}
&&{\bar \omega_0}=\frac{\sin(u-v-\xi+\eta)\, \sin(\eta-u-v-\xi)}{\sin(u-v+\eta)\, \sin(\eta-u-v)}, \quad
{\bar \omega_2}=\frac{ \sin(u-v-\xi-\eta)}{ \sin(u-v-\eta)}
\nonumber \\
&&{\bar \omega_1}= \frac{\sin(u-v-\xi-\eta)\, \sin(-u-v-\xi-\eta)}{\sin(u-v-\eta)\, \sin(-u-v-\eta)}, \quad 
{\bar \omega_4}=\frac{\sin(\eta-u-v-\xi)}{\sin(\eta-u-v)}
\end{eqnarray*}
we find the following relation, expressing the decomposition of the contributions to $Z_{n}^{20V}[u,v;\xi]$ according to the configurations of their 
topmost path (see Fig.~\ref{fig:ref20v} for an illustration):
\begin{equation}\label{sumrule}
\sum_{k=1}^{2n-1} \left({\bar\omega}_4 \,Z_{n,k}^{20V -}[u,v]+{\bar \omega}_2\,Z_{n,k}^{20V \backslash}[u,v]\right) {\bar \omega}_0^{2n-k-1} {\bar\omega}_1^{k-1}
=Z_{n}^{20V}[u,v;\xi]
\end{equation}
Introducing the parameters
$$\tau:=\frac{{\bar\omega}_1}{{\bar \omega}_0},\qquad \sigma:=\frac{{\bar \omega}_2}{{\bar\omega}_4}$$
this reads:
$$Z_{n}^{20V}[u,v;\xi] ={\bar\omega}_4 \, {\bar \omega}_0^{2n-2} \sum_{k=1}^{2n-1} \tau^{k-1}\, (Z_{n,k}^{20V -}[u,v]+\sigma\, Z_{n,k}^{20V \backslash}[u,v])$$

\subsubsection{Refined one-point function}

As in the 6V' case, 
the corresponding (normalized) refined one-point functions $H_{n,k}^{20V -}[u,v],H_{n,k}^{20V \backslash}[u,v]$ are ratios of slightly modified refined partition functions 
to the original homogeneous partition function $Z_n^{20V}[u,v]$.
The corresponding configurations have a topmost path that stops at the point $k$ after a last step from the $n-1$-th vertical to the $n$-th one (see Fig.~\ref{fig:alltgt} top right, pink domain). Compared to 
$Z_{n,k}^{20V -}[u,v],Z_{n,k}^{20V \backslash}[u,v]$, we must remove the last $k$ vertical steps of the topmost path, 
and thus replace the $k$ corresponding weights by $1$ 
(instead of $\omega_4,\omega_2$) for the turning vertex, and by $\omega_0$ (instead of $\omega_1$) for the $k-1$ vertices crossed by the path:
\begin{equation}\label{1pt20v}
H_{n,k}^{20V -}[u,v]=\frac{1}{\omega_4}\left(\frac{\omega_0}{\omega_1}\right)^{k-1} \frac{Z_{n,k}^{20V -}[u,v]}{Z_n^{20V}[u,v]},\qquad 
H_{n,k}^{20V \backslash}[u,v]=\frac{1}{\omega_2}\left(\frac{\omega_0}{\omega_1}\right)^{k-1} \frac{Z_{n,k}^{20V \backslash}[u,v]}{Z_n^{20V}[u,v]}
\end{equation}

We deduce the relation
\begin{equation}\label{relaonept20v}
H_n^{20V}[u,v;\xi]=\frac{Z_n^{20V}[u,v;\xi]}{Z_n^{20V}[u,v]}=\omega_4[u,v;\xi]\,{\bar \omega}_0^{2n-2} 
\sum_{k=1}^{2n-1} t^{k-1} (H_{n,k}^{20V -}[u,v]+s\, H_{n,k}^{20V \backslash}[u,v])
\end{equation}
where we have used the parameters
\begin{equation}\label{t20v}
t= \tau \frac{\omega_1}{\omega_0}=\frac{\sin(u-v-\xi-\eta)\, \sin(-u-v-\xi-\eta)}{\sin(u-v-\xi+\eta)\, \sin(\eta-u-v-\xi)}=:t_{20V}[\xi],\quad s=\sigma \frac{\omega_2}{\omega_4}=
\frac{ \sin(u-v-\xi-\eta)}{\sin(\eta-u-v-\xi)}
\end{equation}
We note that the function $t_{20V}[\xi]$ is identical to $t_{6V'}[\xi]$ of the 6V' model \eqref{txi6v}.

\subsubsection{Relation to 6V' one-point function}

Using eq.\eqref{inho20v}, and noting moreover that ${\bar a}_o{\bar a}_e={\bar \omega}_0$, we may express:
\begin{equation}\label{relutfin}
\frac{H_n^{20V}[u,v;\xi]}{{\bar \omega}_0^{2n-1}}=\left(\frac{\sin(u-v-\xi-\eta)}{\sin(u-v-\eta)}\right)^{n-1} \, 
\left(\frac{\sin(u-v+\eta)}{\sin(u-v-\xi+\eta)}\right)^n\, \frac{H_n^{6V'}[u,v;\xi]}{({\bar a}_o{\bar a}_e)^{n-1}}
\end{equation}

\subsubsection{Asymptotics}

We now turn to large $n=N$ asymptotics of the one-point functions \eqref{1pt20v} with the scaled exit point position $\kappa =k/(2N)$ kept finite. 
We first note that the relation \eqref{relutfin} yields the large $N$ asymptotics
\begin{eqnarray}\frac{H_N^{20V}[u,v;\xi]}{{\bar \omega}_0^{2N-1}}&\simeq& e^{-N\,\varphi^{20V}[u,v,\xi]} \nonumber \\
\varphi^{20V}[u,v;\xi]&=&\varphi^{6V'}[u,v,\xi]-{\rm Log}\left(\frac{\sin(u-v-\xi-\eta)\sin(u-v+\eta)}{\sin(u-v-\xi+\eta)\sin(u-v-\eta)}\right) \nonumber \\
&&\!\!\!\!\!\!\!\!\!\! =-{\rm Log}\left(\frac{\sin(2v)\sin^2(u-v-\xi-\eta)\,\sin(u+v+\xi+\eta)\,\sin(u-v+\eta)}{\sin(2v+\xi)\sin^2(u-v-\eta)\,\sin(u+v+\eta)\,\sin(u-v-\xi+\eta)}\right)\nonumber \\
&&\!\!\!\!\!\!\!\!\!\!  -{\rm Log}\left(
\frac{\sin(\al\xi)\sin(\al(\xi+2v+2\eta))\sin(\al(u-v-\eta))\sin(\al(u+v+\eta))}{\al \,\sin(\xi)\sin(2\al (v+\eta))\sin(\al(u-v-\xi-\eta))\sin(\al(u+v+\xi+\eta))} \right)\label{asymptoh20v}
\end{eqnarray}

As the parameter $s$ is finite and independent of $k$, using the relation \eqref{relaonept20v}, the connection between $H_n^{20V}[u,v;\xi]$ to the 6V' one-point function \eqref{relutfin} and finally
the asymptotics \eqref{inho20v}, we get identical leading behaviors for both one-point functions.

\begin{thm}\label{20vasyonethm}
The large $n=N$ scaling limit of the refined one-point functions $H_{n,k}^{20V -}[u,v]$ and $H_{n,k}^{20V \backslash}[u,v]$
reads:
\begin{eqnarray*}
H_{N,2\kappa N}^{20V -}[u,v]&\simeq& H_{N,2\kappa N}^{20V \backslash}[u,v]\simeq \oint \frac{dt }{2i\pi t} e^{-N\, S_0^{20V}(\kappa,t)} \\
S_0^{20V}(\kappa,t)&:=& \varphi^{20V}[u,v;\xi]+2\kappa\,{\rm Log}(t) 
\end{eqnarray*}
where $\varphi^{20V}[u,v;\xi]$ is as in \eqref{asymptoh20v}, and 
in which the variables $t$ and $\xi$ are related via $t=t_{20V}[\xi]$ \eqref{t20v}. 
\end{thm}

As before, the integral is dominated at large $N$ by the solution of the saddle-point equation
$\partial_t S_0^{20V}(\kappa,t)=0$, or equivalently, changing integration variables to $\xi$: $\partial_\xi S_0^{20V}(\kappa,t_{20V}[\xi])=0$. Using the identification $t_{20V}[\xi]=t_{6V'}[\xi]$, this is easily solved as
\begin{eqnarray}
&&\!\!\!\!\!\!\!\!\!\!\! \kappa= \kappa_{20V}[\xi]:=-\frac{1}{2}\,\frac{t_{6V'}[\xi]}{\partial_\xi t_{6V'}[\xi]}\, 
\partial_\xi \varphi^{20V}[u,v;\xi]\nonumber \\
&&\!\!\!\!\!\!\!\!\!\!\! =\frac{\kappa_{6V'}[\xi]}{2}+\frac{\cot(u-v-\xi+\eta)-\cot(u-v-\xi-\eta)}{2\,\sin(2\eta)}\,\frac{\sin^2(u-v-\xi+\eta)\sin^2(u-v-\xi-\eta)}{\cos(2u)\cos(2\xi+2v)-\cos(2\eta)}\nonumber \\
\label{kappa20v}
\end{eqnarray}
with $\kappa_{6V'}[\xi]$ as in (\ref{sapo6v}-\ref{rsol}).

\subsection{Paths}
$\ $

\subsubsection{Partition function}

With the setting of Fig.~\ref{fig:alltgt} (top right, light blue domain), we wish to compute the partition function $Y_{k,\ell}(\beta_1,\beta_2)$ of a single (Schr\"oder) path of the 20V model in the first quadrant $\Z_+^2$,
with starting point $(0,k)$ and endpoint $(\ell,0)$. We include a weight $\beta_1,\beta_2$ according to the configuration of the step taken before entering the path domain (last step in the pink domain, respectively horizontal or diagonal).

The paths receive homogeneous 20V weights \eqref{weights20V}, with 
horizontal, vertical, diagonal uniform spectral parameters $u+\eta,\ v,\ -u$ respectively, while
all vertices not visited by the path receive the weight $\omega_0$. As in the previous cases, we may factor out an 
unimportant overall factor $\omega_0^{k\ell}$ (where $k\ell$ is the area of the light blue rectangle $[0,\ell]\times [0,k]$ in
Fig.~\ref{fig:alltgt} top right), and weight 
the vertices visited by the path by and extra factor $\frac{1}{\omega_0}$. 

The partition function $Y_{k,\ell}(\beta_1,\beta_2)$ is computed by use of a transfer matrix technique (see \cite{BDFG} Appendix B for details with slightly different definitions). Each path is travelled from N,W to S,E, and the transfer matrix
is a $3\times 3$ matrix whose entries correspond to the vertex weight for the transition from the entering step at each visited vertex to the outgoing step. 
The three states are $(-,\backslash,\vert)$ for respectively a horizontal, diagonal, vertical step ending at the transition vertex.
Moreover we include an extra weight $z,zw,w$ per horizontal, diagonal, vertical outgoing step respectively. Note that the step prior to entering the quadrant (exit from the rectangular domain) may be
either horizontal (with an extra weight $\beta_1$) or diagonal (with an extra weight $\beta_2$), while the last step is vertical.
The transfer matrix $T_{20V}$ reads:
$$T_{20V}=\frac{1}{\omega_0} \begin{pmatrix}
\omega_6 z & \omega_5 z & \omega_4 z\\
\omega_5 z w & \omega_3 z w & \omega_2 z w \\
\omega_4 w & \omega_2 w & \omega_1 w
\end{pmatrix}$$
The generating function for the $Y_{k,\ell}$ reads
$$\sum_{k,\ell\geq 0} Y_{k,\ell}(\beta_1,\beta_2) \,z^k w^{\ell+1} = (0,0,1) ({\mathbb I}- T_{20V})^{-1} \begin{pmatrix} \beta_1\\ \beta_2\\ 0\end{pmatrix} $$
This is a rational fraction with denominator $\det({\mathbb I}- T_{20V})=1-\al_1 w-\al_2 z -\al_3 z w-\al_4 zw^2-\al_5 z^2w-\al_6 z^2w^2$, where
\begin{eqnarray}
&&\al_1=\frac{\omega_1}{\omega_0},\quad \al_2=\frac{\omega_6}{\omega_0},\quad \al_3=\frac{\omega_0\omega_3+\omega_4^2-\omega_1\omega_6}{\omega_0^2}\nonumber\\
&& \label{denom}\\
&&\al_4=\frac{\omega_2^2-\omega_1\omega_3}{\omega_0^2} ,\quad \al_5=\frac{\omega_5^2-\omega_6\omega_3}{\omega_0^2}, \quad \al_6=\frac{2\omega_2\omega_4\omega_5+\omega_1\omega_6\omega_3-\omega_3\omega_4^2-\omega_1\omega_5^2-\omega_6\omega_2^2}{\omega_0^3} \nonumber
\end{eqnarray}

\subsubsection{Asymptotics}

We now consider the large $n=N,k,\ell$ limit, with $\kappa= k/(2N)$ and $\lambda=\ell/N$ fixed.  Like in Sect. \ref{secasym6v} above, the asymptotics of $Y_{k,\ell}$ are determined by the denominator \eqref{denom}, and read (see also Ref. \cite{BDFG} appendix B for details):
\begin{eqnarray}
&&\qquad \qquad\quad \  Y_{2\kappa N,\lambda N}\simeq \int_0^1 dp_3 dp_4 dp_5 dp_6 e^{-NS_1^{20V}(\kappa,p_3,p_4,p_5,p_6)}\nonumber \\
&&S_1^{20V}(\kappa,p_3,p_4,p_5,p_6)=-(2\kappa+\lambda-p_3-2p_4-2p_5-3p_6){\rm Log}(2\kappa+\lambda-p_3-2p_4-2p_5-3p_6)\nonumber \\
&&\qquad\qquad\qquad\qquad +(2\kappa-p_3-2p_4-p_5-2p_6){\rm Log}\left(\frac{2\kappa-p_3-2p_4-p_5-2p_6}{\al_1}\right)\nonumber \\
&&\qquad\qquad\qquad\qquad+(\lambda-p_3-p_4-2p_5-2p_6){\rm Log}\left(\frac{\lambda-p_3-p_4-2p_5-2p_6}{\al_2}\right)\nonumber \\
&&\qquad\qquad\qquad\qquad+\sum_{i=3}^6 p_i{\rm Log}\left(\frac{p_i}{\al_i}\right)
\label{asymptopath20v}
\end{eqnarray}
As before this also covers the case of vanishing weights $\al_i$ by taking the limit $p_i\to 0$ at finite $\al_i$ in the above.

\subsection{Arctic curves}

\begin{thm}\label{20VNEthm}
The NE branch of the arctic curve for the 20V-DWBC3 model on the quadrangle $\cQ_n$ is predicted by the Tangent Method to be:
$$x=X_{NE}^{20V}[\xi]=\frac{B'[\xi]}{A'[\xi]}, \qquad y=Y_{NE}^{20V}[\xi]= B[\xi]-\frac{A[\xi]}{A'[\xi]}B'[\xi]
$$
where $B[\xi]=2 \kappa_{20V}[\xi]$ with $\kappa_{20V}[\xi]$ as in \eqref{kappa20v}, and where $A[\xi]$ is given by
\begin{eqnarray}
A[\xi]
&=& \frac{\cos(2\eta)-\cos(u+v+\eta)\cos(u+v-\eta+2\xi)}{\cos(2\eta)-\cos(2u)\cos(2v+2\xi)}\, \frac{\sin(u-v-\eta-\xi)\sin(u-v+\eta-\xi)}{\sin(\xi)\sin(\xi-2\eta)}\nonumber \\
\label{A20v}
\end{eqnarray}
and with the parameter range:
$$ 
\xi\in \big[\eta+u-v-\pi,0\big] 
$$
\end{thm}
\begin{proof}
We may now bring together the ingredients of the Tangent Method. We determine the family of tangents $F_\xi(x,y)=y+A[\xi]x-B[\xi]$ defined in Sect. \ref{sectan}.
We already identified the intercept $B[\xi]=2\kappa_{20V}[\xi]$ with $\kappa_{20V}[\xi]$ given by \eqref{kappa20v}. 
To determine the slope $A[\xi]=2\kappa/\lambda$, we must find the leading contribution to the total partition function
\begin{eqnarray*} 
&&\sum_{k=1}^{2n-1} H_{N,k}[u,v] \, Y_{k,\ell} \simeq \int_0^1 d\kappa H_{N,2\kappa N}[u,v]\, Y_{2\kappa N,\lambda N}\simeq 
\int_0^1 d\kappa dp_3 dp_4,dp_5 dp_6 e^{-N S^{20V}(\kappa,p_3,p_4,p_5,p_6,t)} \\
&&S^{6V'}(\kappa,p_2,p_4,p_5,t):=S_0^{20V}(\kappa,t)+S_1^{20V}(\kappa,p_3,p_4,p_5,p_6)
\end{eqnarray*}
with $S_0^{20V}(\kappa,t)$ as in \eqref{action6v} and $S_1^{20V}(\kappa,p_3,p_4,p_5,p_6)$ as in \eqref{asymptopath20v}. As in the 6V case, the saddle-point equation 
$\partial_\xi S^{20V}=0$ is solved by \eqref{kappa20v}, and amounts to parameterizing $\kappa=\kappa_{20V}[\xi]$ in terms of the parameter $\xi$.
The saddle-point equations
$\partial_\kappa S^{20V}=\partial_{p_3}S^{20V}=\partial_{p_4}S^{20V}=\partial_{p_5}S^{20V}=\partial_{p_6}S^{20V}=0$ give rise to the system of algebraic equations:
\begin{eqnarray*}
\frac{t}{\al_1}&=&\frac{2\kappa+\lambda-p_3-2p_4-2p_5-3p_6}{2\kappa-p_3-2p_4-p_5-2p_6} \\
\frac{\al_1\al_2\,p_3}{\al_3}&=&\frac{(2\kappa-p_3-2p_4-p_5-2p_6)(\lambda-p_3-p_4-2p_5-2p_6)}{2\kappa+\lambda-p_3-2p_4-2p_5-3p_6}\\
\frac{\al_1^2\al_2\,p_4}{\al_4}&=&\frac{(2\kappa-p_3-2p_4-p_5-2p_6)^2(\lambda-p_3-p_4-2p_5-2p_6)}{(2\kappa+\lambda-p_3-2p_4-2p_5-3p_6)^2}\\
\frac{\al_1\al_2^2\,p_5}{\al_5}&=&\frac{(2\kappa-p_3-2p_4-p_5-2p_6)(\lambda-p_3-p_4-2p_5-2p_6)^2}{(2\kappa+\lambda-p_3-2p_4-2p_5-3p_6)^2}\\
\frac{\al_1^2\al_2^2\,p_6}{\al_6}&=&\frac{(2\kappa-p_3-2p_4-p_5-2p_6)^2(\lambda-p_3-p_4-2p_5-2p_6)^2}{(2\kappa+\lambda-p_3-2p_4-2p_5-3p_6)^3}
\end{eqnarray*}
Substituting the values of $t=t_{20V}[\xi]$ \eqref{t20v} and of the weights $\al_i$ \eqref{denom} expressed using \eqref{weights20V}:
\begin{eqnarray*}
&& \al_1=\frac{\sin(u-v-\eta)\sin(u+v+\eta)}{\sin(u-v+\eta)\sin(u+v-\eta)}, \quad \al_2=\frac{\sin(2u)\sin(u-v-\eta)}{\sin(2u+2\eta)\sin(u-v+\eta)}\\
&& \al_3= \frac{2\sin(2\eta)\sin(2u)\big(\sin(u-v+{\scriptstyle \frac{\pi}{4}})\sin(u+v+{\scriptstyle \frac{\pi}{4}})-\sin^2(2\eta)\big)}{\sin(2u+2\eta)\sin(u-v+\eta)\sin(u+v-\eta)}\\
&& \al_4=\frac{\sin(2u)\sin(u-v-\eta)\sin(u+v+3\eta)}{\sin(2u+2\eta)\sin(u-v+\eta)\sin(\eta-u-v)}\\
&&\al_5= \frac{\sin(2u-2\eta)\sin(u-v-\eta)\sin(u+v+\eta)}{\sin(2u+2\eta)\sin(u-v+\eta)\sin(\eta-u-v)}\\
&&\al_6=\frac{\sin(u-v-3\eta)\sin(u+v+3\eta)\sin(2u-2\eta)}{\sin(u-v+\eta)\sin(u+v-\eta)\sin(2u+2\eta)}
\end{eqnarray*} 
we find the unique solution such that $\lambda,\kappa>0$:
\begin{eqnarray}
\frac{p_3}{\kappa}&=& \frac{2\sin(\xi-2\eta)\sin(\xi)\sin(u+v+\xi-\eta)\sin(u+v+\xi+\eta)}{\sin(2\eta)\sin(u-v-\eta)\big(\cos(2u)\cos(u+v+\eta)-\cos(2\eta)\cos(u-v-2\xi+\eta)\big)}\nonumber \\
&&\qquad \times\, \frac{\cos^2(2u)-\cos(4\eta)-\sin(2u)\sin(2v+2\eta)}{\cos(2\eta)-\cos(u+v+\eta)\cos(u+v+2\xi-\eta)}\nonumber \\
\frac{p_4}{\kappa}&=& \frac{\sin(2u)\sin(u+v+3\eta)\sin(u-v-\xi+\eta)}{\sin^2(2\eta)\sin(u-v-\xi-\eta)\big(\cos(2u)\cos(u+v+\eta)-\cos(2\eta)\cos(u-v-2\xi+\eta)\big)} \nonumber \\
&&\qquad \times \ \frac{\sin(\xi-2\eta)\sin(\xi)\sin^2(u+v+\xi-\eta)}{\cos(2\eta)-\cos(u+v+\eta)\cos(u+v+2\xi-\eta)} \nonumber \\
\frac{p_5}{\kappa}&=& \frac{2\sin(2u-2\eta)\sin(u+v+\eta)}{\sin^2(2\eta)\big(\cos(2u)\cos(u+v+\eta)-\cos(2\eta)\cos(u-v-2\xi+\eta)\big)}\nonumber \\
&&\qquad \times \ \frac{\sin^2(\xi)\sin^2(u+v+\xi+\eta)}{\cos(2\eta)-\cos(u+v+\eta)\cos(u+v+2\xi-\eta)} \nonumber \\
\frac{p_6}{\kappa}&=& \frac{2\sin(2u-2\eta)\sin(u-v-3\eta)\sin(u+v+3\eta)\sin^2(\xi)}{\sin^2(2\eta)\sin(u-v-\eta)\big(\cos(2u)\cos(u+v+\eta)-\cos(2\eta)\cos(u-v-2\xi+\eta)\big)}\nonumber \\
&&\qquad \times \ \frac{\sin(u-v-\xi+\eta)\sin(u+v+\xi-\eta)\sin(u+v+\xi+\eta)}{\sin(u-v-\xi-\eta)\big(\cos(2\eta)-\cos(u+v+\eta)\cos(u+v+2\xi-\eta)\big)}\nonumber \\
\frac{\kappa}{\lambda}&=&
\frac{\sin(u-v-\xi-\eta)\sin(u-v-\xi+\eta)\big(\cos(2\eta)-\cos(u+v+\eta)\cos(u+v+2\xi-\eta)\big)}{2\sin(\xi-2\eta)\sin(\xi)\big(\cos(2\eta)- \cos(2u)\cos(2v+2\xi)\big)}\nonumber \\
&&\label{valal}
\end{eqnarray}
Using the parametrization $\kappa=\kappa_{20V}[\xi]$, we may interpret the last equation as determining $\lambda$ as a function $\lambda_{20V}[\xi]$ of the parameter $\xi$, where:
\begin{eqnarray}
\lambda_{20V}[\xi]&:=&\kappa_{20V}[\xi]\, \frac{2 \sin(\xi)\sin(\xi-2\eta)}{\sin(u-v-\eta-\xi)\sin(u-v+\eta-\xi)} \nonumber \\
&&\qquad \times\ \frac{\cos(2\eta)-\cos(u+v+\eta)\cos(u+v-\eta+2\xi)}{\cos(2\eta)-\cos(2u)\cos(2v+2\xi)}\label{lam20vofxi}
\end{eqnarray}
To summarize, we have found the most likely exit point $\kappa$
as an implicit function of the arbitrary parameter $\lambda$, via the parametric equations $(\kappa,\lambda)=(\kappa_{20V}[\xi],\lambda_{20V}[\xi])$, which 
results in the family of tangent lines $F_\xi(x,y)=0$.
The theorem follows from the expressions \eqref{acurve}, by identifying the slope $A[\xi]=2\kappa_{20V}[\xi]/\lambda_{20V}[\xi]$, while the range of the parameter $\xi$ corresponds to imposing $A[\xi] \in [0,\infty)$.
\end{proof}

As explained in Section \ref{obsec}, the SE branch of the arctic curve is easily obtained by applying the transformation
$(u,v)\mapsto (u^*,v^*)=(u,-v-\pi)$ and the change of coordinates $(x,y)\mapsto (x,2-x-y)$.

\begin{thm}\label{20VSEthm}
The SE branch of the arctic curve for the 20V-DWBC3 model is given by the parametric equations
$$ x=X_{SE}^{20V}[\xi]={X_{NE}^{20V}}[\xi]^* \qquad 
y=Y_{SE}^{20V}[\xi]=2-X_{NE}^{20V}[\xi]^*-{Y_{NE}^{20V}}[\xi]^*\qquad (\xi\in \big[\eta+u+v,0\big]  )$$
with $X_{NE}^{20V},Y_{NE}^{20V}$ as in Theorem \ref{20VNEthm}, and where the superscript $*$ stands for the transformation $(u,v)\mapsto (u^*,v^*)=(u,-v-\pi)$,
which we have also applied to the range of $\xi$.
\end{thm}

\subsection{Examples}

$\ $
\begin{figure}
\begin{center}
\begin{minipage}{0.5\textwidth}
        \centering
        \includegraphics[width=5cm]{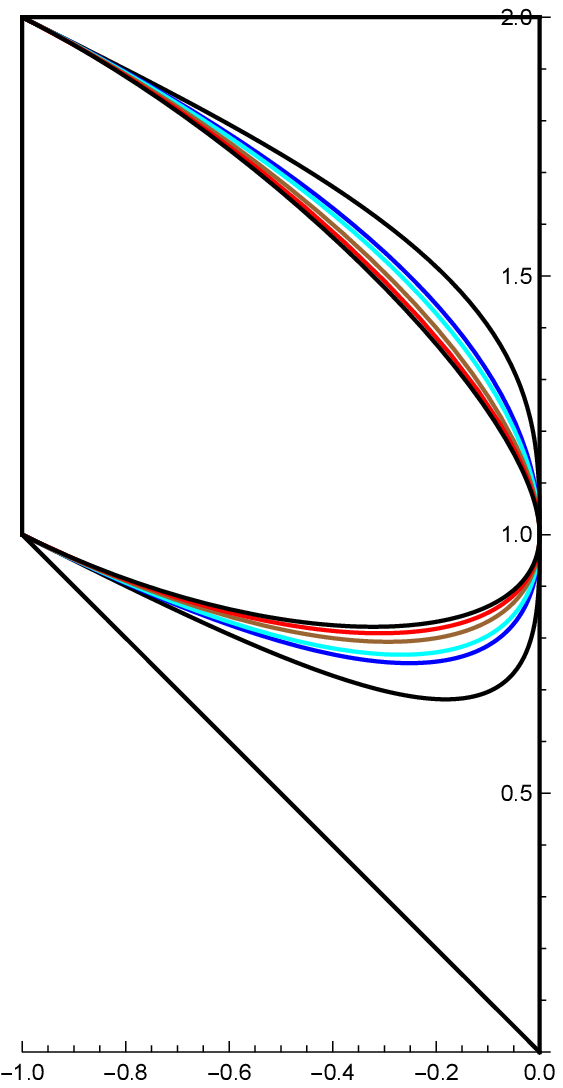} 
    \end{minipage}\hfill
    \begin{minipage}{0.5\textwidth}
        \centering
        \includegraphics[width=5cm]{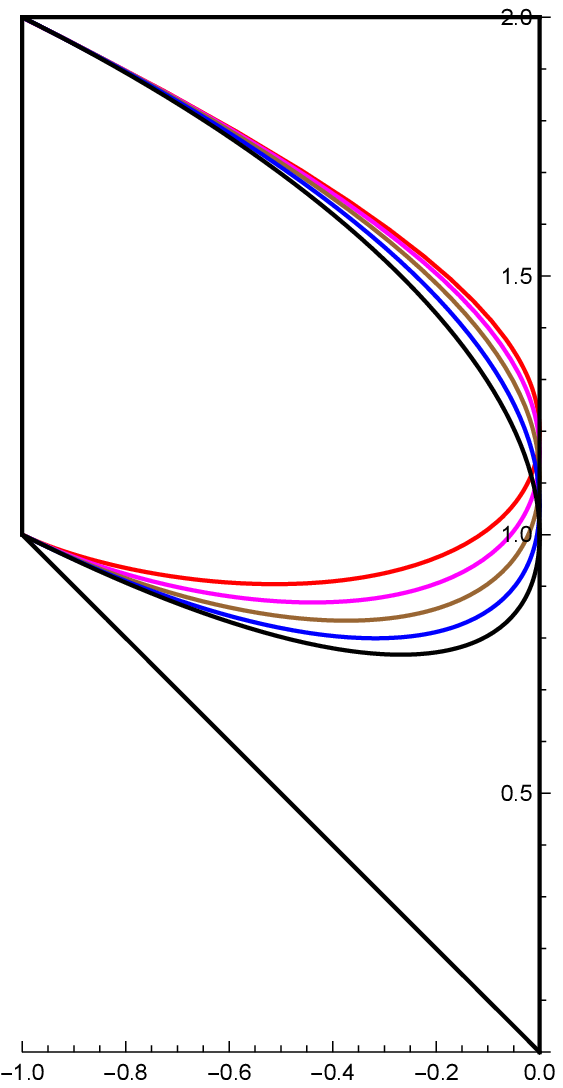}  
    \end{minipage}
\end{center}
\caption{\small Left: Arctic curve of the 20V-DWBC3 in the cases $u=0$, $v=-\frac{\pi}{2}$ and $\eta$ varying from $0+$ (outermost curve) to $\frac{\pi}{2}^-$ (innermost curve). Right: Arctic curve of the 20V-DWBC3 in the cases $u=0$, $\eta=\frac{\pi}{6}$ and $v$ varying between $-\frac{\pi}{2}-\frac{\pi}{6}$ (topmost NE curve) and $-\frac{\pi}{2}$ (bottommost).}
\label{fig:arctic5}
\end{figure}

We now illustrate the results of Theorems \ref{20VNEthm} and \ref{20VSEthm} in a few examples.

\subsubsection{\bf Case $u=0$.}
In this case the arctic curve is entirely made of its NE and SE portions, as it touches the W boundary at points $(-1,1)$
and $(-1,2)$, both with a tangent of slope $-1/2$ (corresponding to $A=1/2$) for all values of $\eta,v$.
We have represented in Fig.~\ref{fig:arctic5} (left) the arctic curves for the self-dual value $v=-\frac{\pi}{2}$ and for $\eta$ ranging from $0^+$ to $\frac{\pi}{2}^-$.  The arctic curve for $\eta=0$ reads:
$$ (X_{NE}[\xi],Y_{NE}[\xi])=\left(\frac{2\xi-\sin(2\xi)}{\pi},1-2\frac{\xi}{\pi} \right) \qquad (\xi\in [-\frac{\pi}{2},0])$$
The limit $\eta\to\frac{\pi}{2}^-$ is singular, however we find a finite result by setting $\eta=\frac{\pi}{2}-\epsilon$
and $\xi=\epsilon \zeta$, and then sending $\epsilon\to 0$, with the result:
\begin{eqnarray*}X_{NE}&=&
\frac{(2+\zeta)^2(\cos(2\pi\zeta)-1+2\pi\zeta^2(\pi(1-\zeta^2)\cos(\pi\zeta)+2\zeta \sin(\pi\zeta))}{4(1+\zeta+\zeta^2)\sin^2(\pi \zeta)}\\
Y_{NE}&=&1+\frac{1}{\zeta}-\frac{\pi(1-\zeta^2)}{2\sin(\pi\zeta)}\\
&&+\frac{(2+\zeta)(2\zeta^2+2\zeta-1)(\cos(2\pi\zeta)-1+2\pi\zeta^2(\pi(1-\zeta^2)\cos(\pi\zeta)+2\zeta \sin(\pi\zeta))}{8(1+\zeta+\zeta^2)\sin^2(\pi \zeta)}  
\end{eqnarray*}
In all these cases, the SE branch is given by $(X_{SE},Y_{SE})=(X_{NE},2-X_{NE}-Y_{NE})$ as $v=v^*$.

We also represent non-selfdual cases in  Fig.~\ref{fig:arctic5} (right), for $u=0$, $\eta=\frac{\pi}{6}$ and $v$ varying between $-\frac{\pi}{2}-\frac{\pi}{6}$ and $-\frac{\pi}{2}$. We see that the tangency point on the vertical $x=0$ moves away from the self-dual point $(0,1)$, and that the curves are no longer nested as in the $v=-\frac{\pi}{2}$ case.

\subsubsection{\bf Uniform case.}

\begin{figure}
\begin{center}
\begin{minipage}{0.4\textwidth}
        \centering
        \includegraphics[width=4.cm]{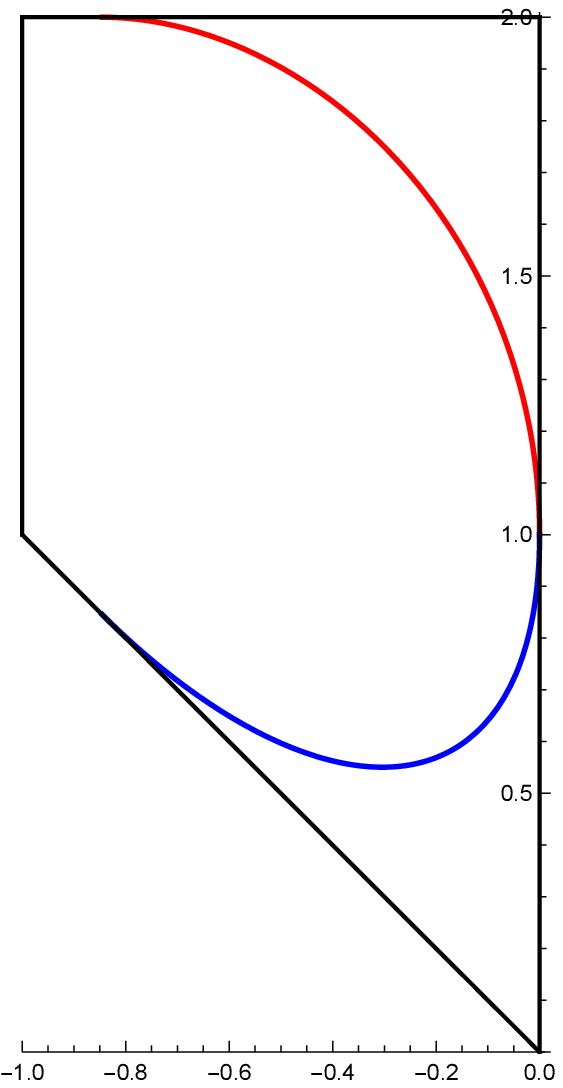} 
    \end{minipage}\hfill
    \begin{minipage}{0.59\textwidth}
        \centering
        \includegraphics[width=6.cm]{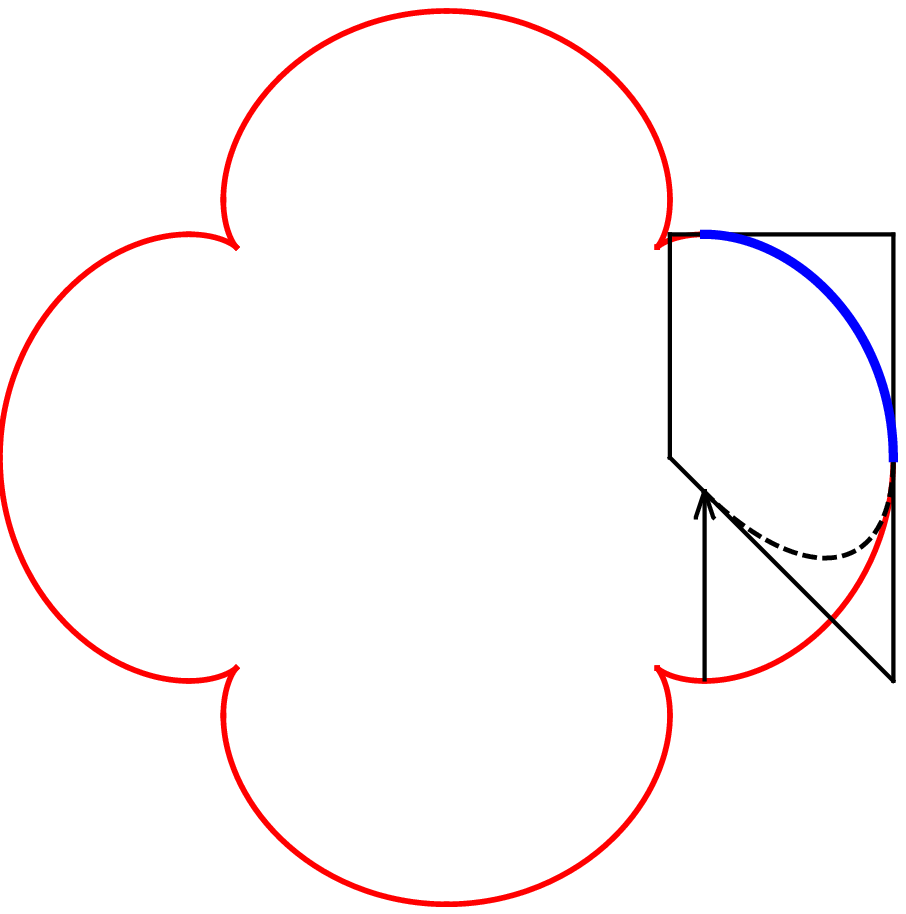}  
    \end{minipage}
\end{center}
\caption{\small Left: Arctic curve (NE=red and SE=blue portions) of the uniform 20V-DWBC3 model on its rescaled domain (black), 
corresponding to $\eta=\frac{\pi}{8}$, $u=\frac{\pi}{8}$ and $v=-\frac{\pi}{2}$. Right: Arctic curve of the uniform 20V-DWBC3 model (NE branch in thicker blue line, SE branch in dashed black), together with the analytic continuation of its NE portion (in red). The arrow indicates the shear transformation from the latter to the SE branch.}
\label{fig:arctic5b}
\end{figure}

As is easily checked on the weights \eqref{weights20V}, the uniform case corresponds to $\eta=\frac{\pi}{8}=u$, $v=-\frac{\pi}{2}$ and $\nu=\sqrt{2}$ \eqref{combipoint20v}. The NE and SE portions of the arctic curve predicted by Theorems \ref{20VNEthm} and \ref{20VSEthm}
have a vertical tangent at $(0,1)$, a horizontal tangent at $\big( \frac{2}{3}(\sqrt{3}-3),2\big)\simeq (-.845,2)$ and a diagonal tangent of slope $-1$ at $\big( \frac{2}{3}(\sqrt{3}-3),\frac{2}{3}(3-\sqrt{3})\big)\simeq (-.845,.845)$. We have represented in Fig.~\ref{fig:arctic5} (left)  the NE and SE portions of the arctic  curve together with the rescaled quadrangular domain $\lim_{n\to \infty} \cQ_n/n$.

As pointed out before, the Tangent Method does not allow to predict the NW and SW portions of the arctic curve. It is interesting however to notice that the NE portion of the curve is algebraic.
With a suitable shift of the origin to the point $(-2,1)$, namely substituting $(x,y)\to (x-2,y+1)$, we obtain the following algebraic equation:
\begin{equation}\label{algebraic}
3^6 (x^2+y^2-{\scriptstyle \frac{2}{3}})^5 -5^3\, 3^3 (x^2+y^2-{\scriptstyle \frac{2}{3}})^3-2\, 3^2\,5^4 (x^2+y^2-{\scriptstyle \frac{2}{3}})^2-2^2\, 5^5 (x^2+y^2-4 x^2 y^2)=0 
\end{equation}
We have represented this algebraic curve in Fig.~\ref{fig:arctic5b} (right) together with the NE portion of the uniform 20V-DWBC3 curve, and the scaled quadrangular domain (in black). We see that the SE portion of the arctic curve (dashed black curve) is obtained as the shear of the analytic continuation of the NE portion (red curve).

\subsubsection{\bf Free fermion case}

\begin{figure}
\begin{center}
\begin{minipage}{0.5\textwidth}
        \centering
        \includegraphics[width=5cm]{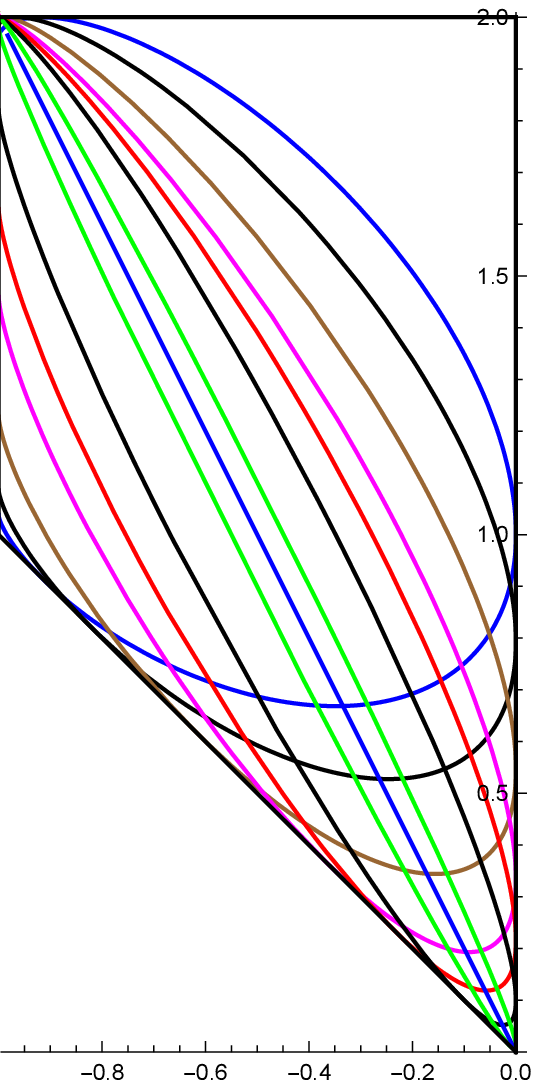} 
    \end{minipage}\hfill
    \begin{minipage}{0.5\textwidth}
        \centering
        \includegraphics[width=5cm]{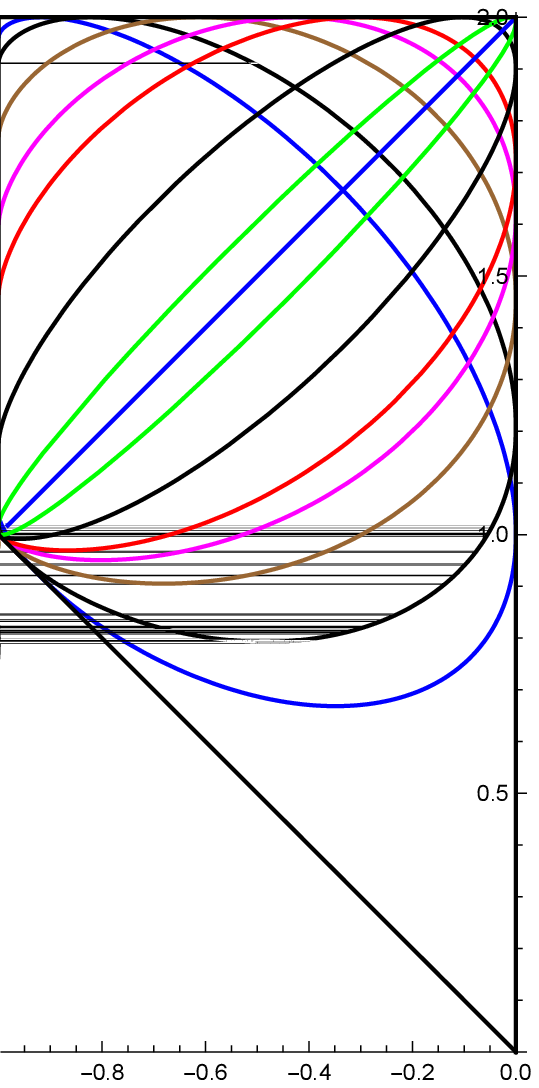}  
    \end{minipage}
\end{center}
\caption{\small Left: Arctic curve of the free fermion 20V-DWBC3 model for $\eta=\frac{\pi}{4}$, $u=\frac{\pi}{16}$ and $v$ varying from $-\frac{\pi}{2}$ (topmost) to $-\frac{\pi}{2}+3\frac{\pi}{16}$ (bottommost). Right: same, but with $v$ varying from $-\frac{\pi}{2}-3\frac{\pi}{16}$(topmost) to $-\frac{\pi}{2}$ (bottommost).}
\label{fig:arctic6}
\end{figure}

In view of the connection to the 6V' model (with same values of $\eta,u,v$) it is clear that 
$\eta=\frac{\pi}{4}$ plays the role of free fermion point. In particular, we expect the arctic curve to be analytic. As a highly non-trivial check, we have verified that at $\eta=\frac{\pi}{4}$ and for all allowed values of $u,v$ the SE branch is the analytic continuation of the NE branch. Like in the 6V' case, we also get access to the NW and SW branches via analytic continuation. We have represented in Fig.~\ref{fig:arctic6} a sequence of cases with $\eta=\frac{\pi}{4}$, $u=\frac{\pi}{16}$
and  $v$ varying (1) from $-\frac{\pi}{2}$ to $-\frac{\pi}{2}+\frac{3\pi}{16}$ (left) and (2) from $-\frac{\pi}{2}-\frac{3\pi}{16}$ to $-\frac{\pi}{2}$ (right). We see that in the case (1) the curve is anchored at the point $(-1,2)$ while the other end along the vertical $x=-1$ varies along the W boundary. The reverse phenomenon is observed in the case (2), where the curve is anchored at the point $(-1,1)$ and its other end varies along the W boundary.

%

\subsubsection{\bf Generic case}

\begin{figure}
\begin{center}
\begin{minipage}{0.5\textwidth}
        \centering
        \includegraphics[width=5cm]{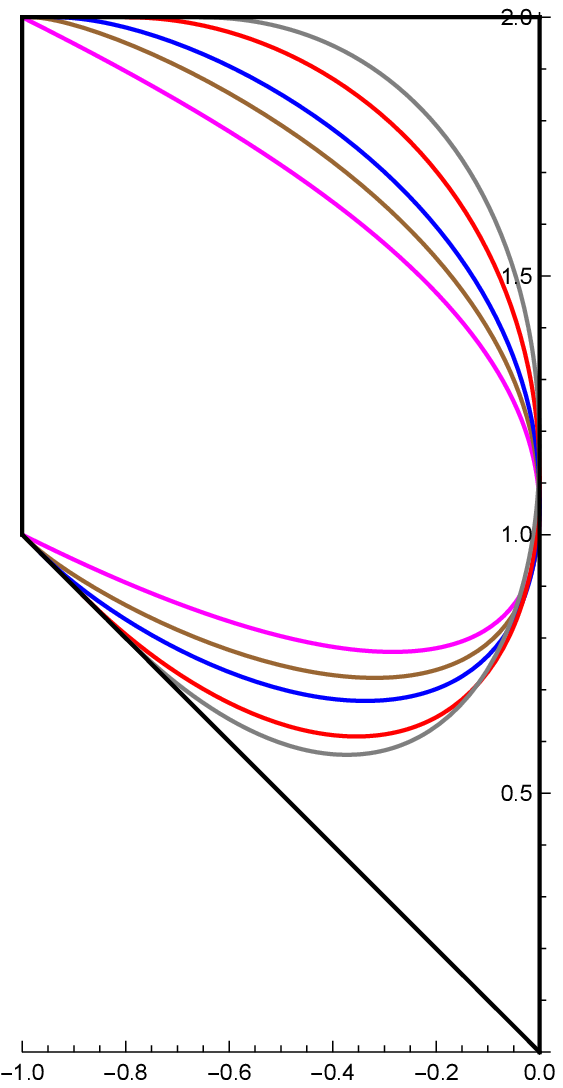} 
    \end{minipage}\hfill
    \begin{minipage}{0.5\textwidth}
        \centering
        \includegraphics[width=5cm]{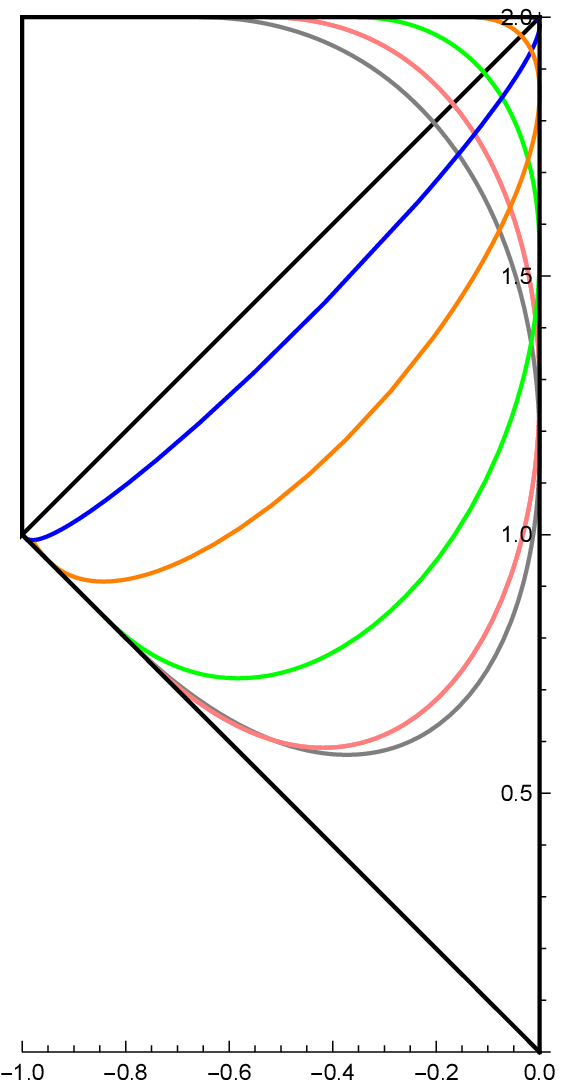}  
    \end{minipage}
\end{center}
\caption{\small Left: Arctic curve of the 20V-DWBC3 model for $\eta=\frac{\pi}{8}$, $v=-\frac{\pi}{2}-\frac{\pi}{32}$ and $u$ varying from $0$ (innermost) to $\frac{3\pi}{16}$. Right: same, but with $u$ varying from $\frac{3\pi}{16}$(bottommost) to $\frac{11\pi}{32}$ (segment).}
\label{fig:arctic7}
\end{figure}

We finally present in Fig.~\ref{fig:arctic7} a ``generic" case with no special symmetry: $\eta=\frac{\pi}{8}$, and $v=-\frac{\pi}{2}-\frac{\pi}{32}\neq v^*=-\frac{\pi}{2}+\frac{\pi}{32}$, for $u$ varying from $0$ to $\frac{11\pi}{32}=\pi+v-\eta$. The last value is singular,
and must be approached as $u=\frac{11\pi}{32}-\epsilon$, $\xi=\epsilon^{1/2}\,\zeta$ with $\epsilon\to 0$. The result
is a line segment joining the points $(-1,1)$ and $(0,2)$.

\section{Aztec triangle domino tilings}\label{DTsec}

\subsection{Partition function and one-point functions}

In Ref.~ \cite{DF20V}, a correspondence was established between the 20V-DWBC3 model on $\cQ_n$ and the Domino Tiling problem of the Aztec triangle $\cT_n$. First it was shown that the models share the same uniformly weighted partition function (total number of configurations):
$$ Z_n^{20V}=Z_n^{DT} $$
Next this correspondence was refined by considering the (uniformly weighted) 20V-DWBC3 refined partition functions
$Z_{n,k}^{20V}$, $k=1,2,...,2n-1$, equal to the refined partition function of Sect.~\ref{ref20vsec} with the parameters \eqref{combipoint20v}.
Their counterparts are the refined partition functions of the Domino tiling problem $Z_{n,k}^{DT}$ defined in a silimar manner, using the non-intersecting Schr\"oder path formulation of Sect.~\ref{dtrefsec}, as the number of configurations in which the topmost path is conditioned to first enter the last vertical at position $k=0,1,...,n-1$, 
before ending with $k$ vertical steps
(see the pink domain in Fig.~\ref{fig:alltgt} (bottom right) for an illustration).
In Ref.~\cite{DF20V}, it was shown that
\begin{equation}\label{usefuldt} Z_{n,k}^{DT}= Z_{n,n+k+1}^{20V}+Z_{n,n+k}^{20V}
\end{equation}
This implies the following relation between the corresponding refined one-point functions
$H_{n,k}^{20V}=Z_{n,k}^{20V}/Z_n^{20V}$ and $H_{n,k}^{DT}=Z_{n,k}^{DT}/Z_n^{DT}$:
\begin{equation}\label{relaonept20vdt}H_{n,k}^{DT}= H_{n,n+k+1}^{20V}+H_{n,n+k}^{20V}
\end{equation}

\subsection{Arctic curves}

\subsubsection{Asymptotics of the one-pont function} As usual, we explore the asymptotics of the refined one-point function
$H_{n,k}^{DT}$ for the Domino Tiling model in the scaling limit of large $n=N$ and $\kappa=k/N$ finite.
The relation \eqref{relaonept20vdt} allows immediately to express:

\begin{thm}\label{DTasyonethm}
The large $n=N$ asymptotics of the refined one-point function
$H_{n,k}^{DT}$ for the Domino Tiling model reads:
\begin{equation}H_{N,\kappa N}^{DT}= 2 \, H_{N,2 \kappa' N}^{20V}, \qquad \kappa'=\frac{1+\kappa}{2} 
\end{equation}
and
\begin{eqnarray*}H_{N,\kappa N}^{DT}&\simeq& \oint \frac{dt}{2i\pi t} e^{-NS_0^{DT}(\kappa,t)} \\
S_0^{DT}(\kappa,t)&=& S_0^{20V}({\scriptstyle \frac{1+\kappa}{2}},t)=\varphi^{20V}[u,v;\xi] +(1+\kappa)\, {\rm Log}(t)
\end{eqnarray*}
where the variables $\xi$ and $t$ are dependent through the relation $t=t_{20V}[\xi]$ \eqref{t20v}.
\end{thm}

Similarly to the 6V' and 20V cases, the saddle-point equation in the variable $\xi$
reads $\partial_\xi S_0^{DT}(\kappa,t_{20V}[\xi])=0$, with  the solution:
\begin{equation}\label{kappadt}
\kappa=\kappa_{DT}[\xi]:= 2 \kappa_{20V}[\xi]-1 \qquad (\xi\in [-\frac{\pi}{4},0])\end{equation}
with $\kappa_{20V}[\xi]$ as in \eqref{kappa20v}, and
where the range of $\xi$ ensures that $\kappa_{20V}\in [1,2]$ hence $\kappa_{DT}\in [0,1]$.

\subsubsection{Asymptotics of Path partition function}

By definition, and comparing Fig.~\ref{fig:alltgt} top right and bottom right (light blue domains), 
we have in the uniform case: $Y_{k,\ell}^{DT}= Y_{k,\ell}^{20V}$. We deduce the asymptotics
$$Y_{\kappa N ,\lambda N}^{DT} \simeq \int_0^1 dp_3 e^{-NS_1^{DT}(\kappa,p_3)}, \qquad S_1^{DT}(\kappa,p_3)=S_1^{20V}(\kappa,p_3) $$
with $S_1^{20V}$ the uniform weight version of \eqref{asymptopath20v}:
$$ S_1^{20V}(\kappa,p_3)=-(\kappa+\lambda-p_3){\rm Log}(\kappa+\lambda-p_3)+(\kappa-p_3){\rm Log}(\kappa-p_3)+(\lambda-p_3){\rm Log}(\lambda-p_3)+p_3{\rm Log}(p_3)$$

\subsubsection{Arctic curves via the Tangent Method}

Strictly speaking, the Tangent Method only predicts the NE portion of the arctic curve. However, the domino tiling problem is of the ``free fermion" class, as it involves only non-intersecting lattice paths (or alternatively the dual is just a dimer model, for which the general results of \cite{KOS} apply). As such, it has an analytic arctic curve, hence we may safely use the analytic continuation of the NE portion predicted by the Tangent Method.

\begin{figure}
\begin{center}
\begin{minipage}{0.66\textwidth}
        \centering
        \includegraphics[width=7cm]{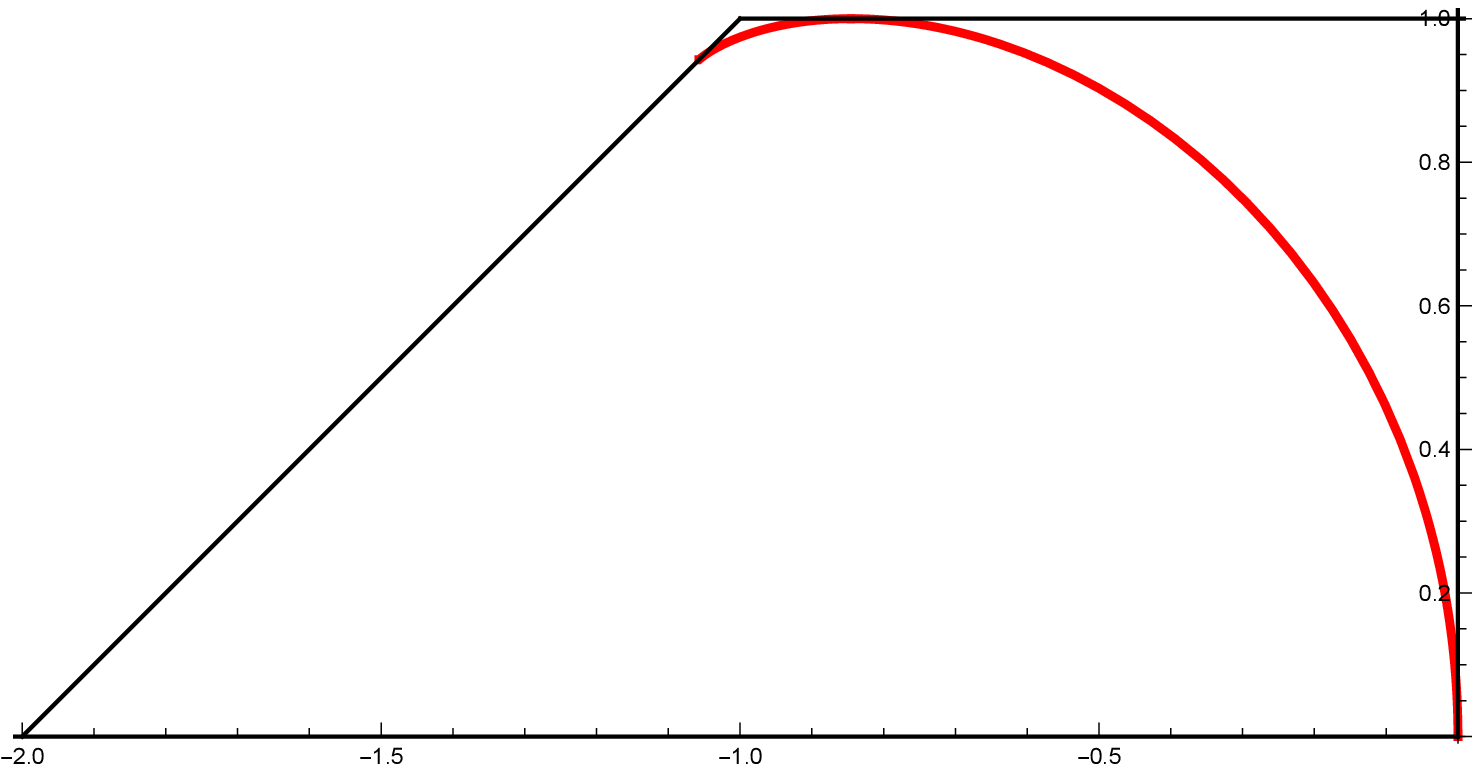} 
    \end{minipage}
    \begin{minipage}{0.33\textwidth}
        \includegraphics[width=4cm]{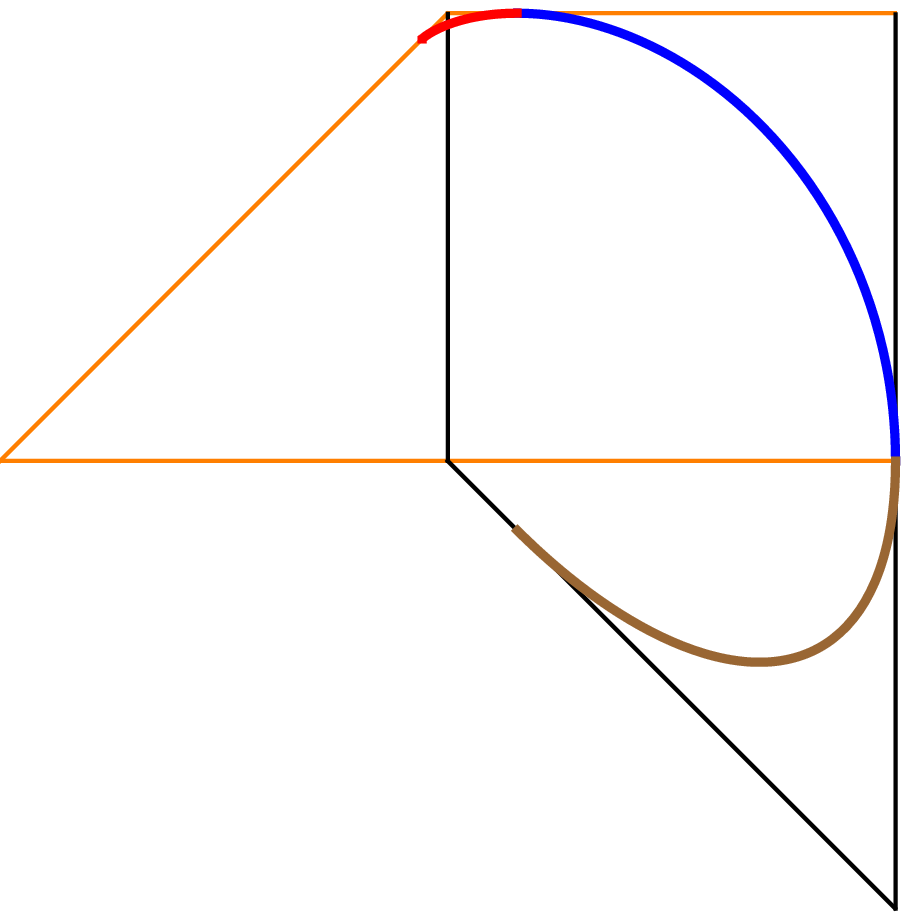}  
    \end{minipage}
\end{center}

\caption{\small Left: Arctic curve of the uniformly weighted Domino Tiling problem of the Aztec triangle, tangent to the NW, N and E boundaries. Right: comparison with the arctic curve of the 20V-DWBC3 model: the blue portion is the common NE branch of the two curves, represented with their respective rescaled domains.}
\label{fig:arcticdt}
\end{figure}

\begin{thm}\label{DTthm}
The arctic curve for the uniform Domino Tilings of the Aztec triangle, as predicted via the Tangent Method,  reads
$$ x=X^{DT}[\xi]= \frac{B'[\xi]}{A'[\xi]} \qquad y=Y^{DT}[\xi]=B[\xi]-\frac{A[\xi]}{A'[\xi]}B'[\xi]\qquad (\xi\in [-\frac{3\pi}{8},0])$$
where
$$B[\xi]:=\kappa_{DT}[\xi]\qquad {\rm and} \qquad A[\xi]:= -\cot(2 \xi) $$
with $\kappa_{DT}[\xi]$ as in \eqref{kappadt}.
\end{thm}
\begin{proof}
The rescaled tangent lines are now through the points $(0,\kappa)$ and $(\lambda, 0)$, governed by the equation $y+A x-B=0$ with $A=\kappa/\lambda$, $B=\kappa$.
We have already determined the most likely exit point $\kappa=\kappa_{DT}[\xi]$ \eqref{kappadt}, leading to $B[\xi]=\kappa_{DT}[\xi]$. To determine $A[\xi]$
we solve the saddle-point equations $\partial_{\kappa} S^{DT}(\kappa,t,p_3)=\partial_{p_3} S^{DT}(\kappa,t,p_3)=0$, in terms of the total action $S^{DT}(\kappa,t,p_3):=S_0^{DT}(\kappa,t)+S_1^{DT}(\kappa,p_3)$.
These read
$$ 
t=\frac{\kappa-p_3}{\kappa+\lambda-p_3},\qquad \frac{p_3(\kappa+\lambda-p_3)}{(\kappa-p_3)(\lambda-p_3)}=1 
$$
and are easily solved into
$$ \frac{p_3}{\kappa_{DT}[\xi]}=\frac{t[\xi]-1}{2t[\xi]}=\frac{\sin(\xi)}{\sqrt{2}\, \sin(\xi-\frac{\pi}{4})},\qquad \frac{\kappa_{DT}[\xi]}{\lambda}=\frac{2t[\xi]}{t[\xi]^2-1}=-\cot(2 \xi)=A[\xi]$$
The range of parameter $\xi$ for the NE portion of arctic curve is $\xi\in [-\frac{\pi}{4},0]$, ensuring that $\kappa_{DT}[\xi]\in [0,1]$, however as noted above we may extend the range
to cover the entire domain, which corresponds to $\kappa_{DT}[\xi]\in [0,2]$, namely $\xi\in [-\frac{3\pi}{8},0]$, and
the Theorem follows.
\end{proof}


We illustrate the result of Theorem \ref{DTthm} in Fig.~\ref{fig:arcticdt} (left). Note that the curve has a vertical tangent at the origin, and a horizontal tangent at the point
$\big(\frac{2}{3}(\sqrt{3}-3),1\big)$, while it ends tangencially on the diagonal NW boundary at the point 
$\big(2\frac{\sqrt{2}}{3}-2,2\frac{\sqrt{2}}{3}\big)$.

We note that the curve of Theorem \ref{DTthm} is a portion of an algebraic curve. In fact, changing the origin to $(-2,0)$ by applying the substitution $(x,y)\to (x-2,y)$, we find that this curve is given by  the {\it same} equation \eqref{algebraic} as in the uniform 20V-DWBC3 case. This is illustrated in Fig.~\ref{fig:arcticdt} (right) where we have represented both rescaled domains, and their common NE branch of the arctic curve (in blue).


\section{Conclusion}\label{seconc}

In this paper, we have presented the Tangent Method derivation of arctic curves for the disordered phase of the 6V-DWBC, 6V', 20V-DWBC3 models, as well as for the Domino Tilings of the Aztec Triangle. The main ingredient used is the large size asymptotics of refined one-point functions, which we derived from the form of the thermodynamic free energy of the 6V' model in the disordered phase (Theorem \ref{freeconj}), and then deducing all relevant asymptotics from there. 
Our method however only predicts the NE and SE branches of the relevant arctic curves.
It would be desirable to find the remaining NW and SW branches of the arctic curves when applicable. 

\begin{figure}
\begin{center}
\begin{minipage}{0.49\textwidth}
        \centering
        \includegraphics[width=5.5cm]{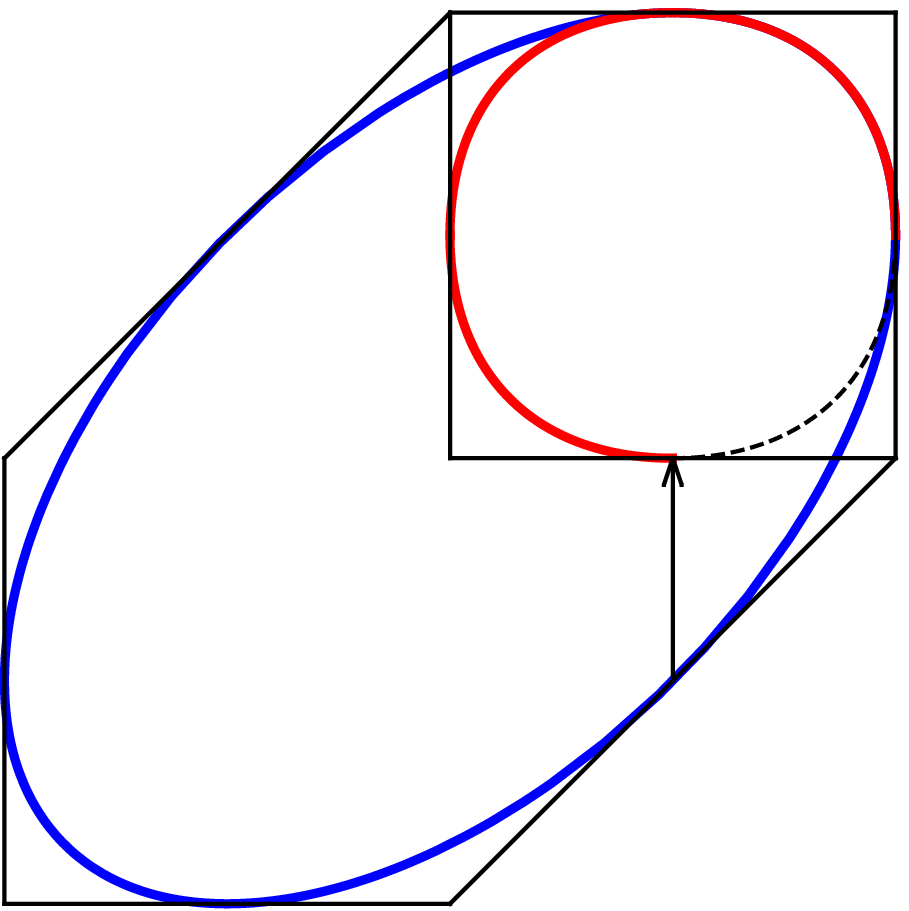} 
    \end{minipage}
    \begin{minipage}{0.5\textwidth}
        \includegraphics[width=5cm]{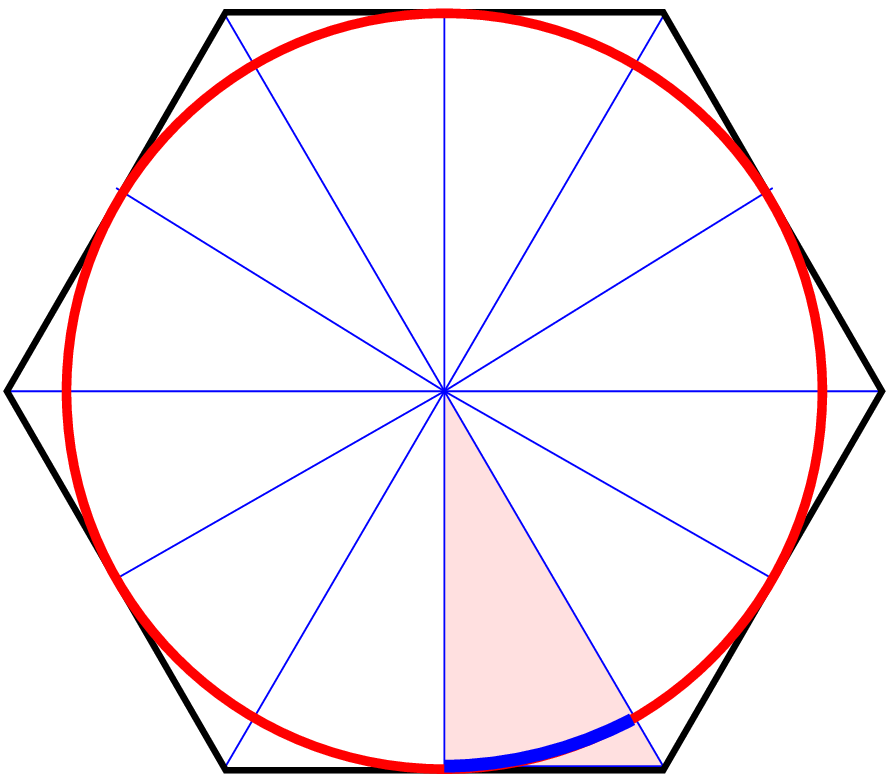}  
    \end{minipage}
\end{center}
\caption{\small Left: the arctic curve of the uniform 6V-DWBC/ASM model (red and dashed black curves inside the square) and the analytic continuation of the NE branch (blue ellipse inscribed in a hexagon); the arrow indicates the shear mapping the latter to the SE branch (dashed black curve). Right: the arctic curve for TSSCPP (blue curve inside the pink triangular domain), and that for the lozenge tiling of the regular hexagon obtained by multiple reflections (red circle).}
\label{fig:asmgen}
\end{figure}

The results for the 6V' and 20V-DWBC3 models of the present paper complement earlier results on the 6V-DWBC \cite{COSPO,DFLAP} and 20V-DWBC1,2 models \cite{BDFG}, which display non-analytic arctic curves as well. The key to the non-analyticity can be traced back to the symmetries of the systems, allowing for determining their SE branch in terms of the NE branch of another system obtained by applying an involution $*$ to its weights together with a geometric transformation of the plane involving a reflection and possibly a shear. Note also that our results for the 6V' model also apply to the more general case of U-turn boundary 6V model, which is expected to share the {\it same} arctic curves.

Finally, let us compare the situation of the 20V-DWBC3 model to that of ASMs, with the known enumeration formula:
$$ ASM_n= \prod_{j=0}^{n-1} \frac{(3j+1)!}{(n+j)!} \ ,$$
a formula strikingly reminiscent of \eqref{20vpf}. The analogy goes further: we have found that the NE/SE portion of arctic curve for large uniform 20V-DWBC3 configurations is piecewise algebraic, the SE portion being equal to a shear transformation of the analytic continuation of the NE portion (see Fig.~\ref{fig:arctic5b} right). The same holds for ASMs, whose NE/SE portion of arctic curve is piecewise elliptic, the SE portion being obtained by a shear transformation of the ellipse containing the NE portion (see Fig.~\ref{fig:asmgen} left). The algebraic curve \eqref{algebraic} clearly plays a role similar to this ellipse.

Finally, recall that ASMs of size $n$ are also in same number as TSSCPP \cite{tsscpp}, which can be viewed as rhombus tilings of a regular hexagon with edges of length $2n$, which satisfy all the symmetries of the hexagon. The triangular fundamental domain under these symmetries occupies $\frac{1}{12}$-th of the hexagon, which is recovered by successive reflections (see Fig.~\ref{fig:asmgen} right for an illustration). As such, the arctic curve for TSSCPP was argued in \cite{DFR} to be identical to that of the full hexagon without any symmetry constraint, i.e. the inscribed circle in the uniform case. There is a clear analogy between TSSCPP and the Domino Tilings of the Aztec triangle, in which the algebraic curve \eqref{algebraic} plays the role of this circle.

\begin{figure}
\begin{center}
\includegraphics[width=7cm]{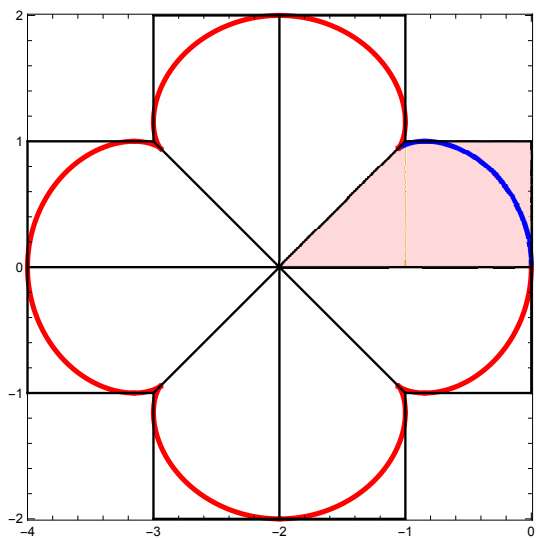}
\end{center}
\caption{\small The expected arctic curve of the uniformly weighted Domino Tiling problem of Ciucu's cruciform region (in red), is obtained as the multiple reflection of the arctic curve of the Aztec triangle (in blue). The resulting clover-shaped curve is the analytic continuation to the whole plane of the blue portion.}
\label{fig:cloverdt}
\end{figure}

Recently Ciucu \cite{CIUconj} noticed a relation between the number of domino tilings of the Aztec triangle $\cT_n$ and that of a cruciform domain $C^{n-1,n,n,n-2}_{2n-1,2n-1}$, obtained by ``symmetrization", namely a
succession of ``reflections" of the original Aztec triangle. We believe that the curve \eqref{algebraic}, which is the analytic continuation of the arctic curve for the triangle, is in fact the complete arctic curve for the rescaled large $n$ cruciform domain. As visual evidence, we have displayed both curves in Fig.~\ref{fig:cloverdt}, together with the original asymptotic Aztec triangle (shaded in pink) and its 7 reflected copies.
Fig.~\ref{fig:cloverdt} suggests that, similarly to the TSSCPP case, the Aztec triangle could be the fundamental domain for symmetric tilings of a crosslike shaped domain probably similar to that considered by Ciucu.

\bibliographystyle{amsalpha} 

\bibliography{ArcticNew}
\end{document}